\documentclass[11pt]{amsart}
\usepackage{lmodern}
\usepackage[svgnames]{xcolor}
\usepackage{amssymb}
\usepackage[arrow,curve,rotate]{xypic}
\usepackage{gastex}
\usepackage{enumerate}
\usepackage{amsthm}
\usepackage[T1]{fontenc}
\usepackage[latin1]{inputenc}
\usepackage{enumitem}
\usepackage[mathcal]{euscript}
\usepackage{comment}
\usepackage{mathrsfs}
\usepackage{ae,aecompl}
\usepackage[normalem]{ulem}
\usepackage{relsize}

\usepackage{soul}

\usepackage{xspace}
\definecolor{lightblue}{rgb}{.90,.95,1}
\definecolor{lightgreen}{rgb}{.90,1,.95}
\definecolor{lightyellow}{rgb}{1,1,.8}
\definecolor{lightred}{rgb}{1,.9,.9}

\usepackage[textwidth=4cm,colorinlistoftodos]{todonotes}
\setenumerate[1]{label=$(\alph*)$,leftmargin=*}
\setenumerate[2]{label=$(\roman*)$,leftmargin=*}
\setitemize[1]{label=$\bullet$,leftmargin=*}
\setitemize[2]{label=$\circ$,leftmargin=*}
\setitemize[3]{label=$--$,leftmargin=*}

\advance\footskip by 1cm

\newtheorem{Thm}{Theorem}[section]
\newtheorem{Prop}[Thm]{Proposition}
\newtheorem{Lemma}[Thm]{Lemma}
\newtheorem{Cor}[Thm]{Corollary}
 {
   \theoremstyle{definition}
 \newtheorem{Remark}[Thm]{Remark}
 }

\theoremstyle{definition}
\newtheorem{Def}[Thm]{Definition}
\newtheorem{Example}[Thm]{Example}

\numberwithin{equation}{section}

\def\pv#1{\ensuremath{{\sf#1}}}
\def\Cl#1{\ensuremath{\mathcal {#1}}}
\def\Om#1#2{\ensuremath{\overline\Omega_{#1}{\sf#2}}}

\def\malcev{\mathop{\raise1pt\hbox{\footnotesize$\bigcirc$\kern-8pt\raise1pt
      \hbox{\tiny$m$}\kern1pt}}}

\def\ev {l}
\def\teu{\mathrm t}
\def\beu {\mathrm i}
\def\Dom{\mathop{\mathrm {Dom}}}
\def\z#1{\ensuremath{{{#1}}^{\mathbb Z}}}
\def\step{\pv{step}}
\def\stat{\pv{stat}}
\def\rstab{\pv{rstab}}
\def\lstab{\pv{lstab}}

\let\cal=\mathcal

\makeatletter
\renewcommand*\l@subsection{\@tocline{1}{0pt}{2.5pc}{5pc}{}}
\makeatother

\title{The linear nature of pseudowords}

\author{J. Almeida}
\address{CMUP, Departamento de Matem\'atica,
  Faculdade de Ci\^encias, Universidade do Porto, 
  Rua do Campo Alegre 687, 4169-007 Porto, Portugal.}
\email{jalmeida@fc.up.pt}

\author{A. Costa}
\address{CMUC, Department of Mathematics, University of Coimbra,
  Apartado 3008, EC Santa Cruz, 3001-501 Coimbra, Portugal.}
\email{amgc@mat.uc.pt}

\author{J. C. Costa}
\address{CMAT, Dep. Matem\'atica e Aplica\c c\~oes, Universidade do Minho, Campus de Gualtar, 4710-057 Braga, Portugal.}
\email{jcosta@math.uminho.pt}

\author{M. Zeitoun}
\address{Univ. Bordeaux, LaBRI, UMR 5800, F-33400 Talence, France.}
\email{mz@labri.fr}

\subjclass[2010]{Primary 20M07, 20M05. Secondary 06A05, 68Q45, 37B10}

\keywords{Relatively free profinite semigroup, aperiodic semigroup,
equidivisible semigroup, labeled linear order, pseudoword, pseudovariety.}

\begin{document}

\begin{abstract}
  Given a pseudoword over suitable pseudovarieties, we associate to it
  a labeled linear order determined by the factorizations of the
  pseudoword. We show that, in the case of the pseudovariety of
  aperiodic finite semigroups, the pseudoword can be recovered from
  the labeled linear order.
\end{abstract}

\maketitle

 \tableofcontents

\section{Introduction}
\label{sec:introduction}

Since the publication of Eilenberg's textbook~\cite{Eilenberg:1976},
a large body of finite semigroup theory is in fact the theory of
pseudovarieties of semigroups. Besides its own mathematical interest,
it draws motivation from the connections with computer science through
Eilenberg's correspondence between pseudovarieties of semigroups and
varieties of regular languages. As pseudovarieties are classes of
finite semigroups, only in very special cases do they contain most
general members on a given finite set of generators, that is relatively free
semigroups, namely semigroups on $n$ generators in the pseudovariety
such that every other member of the pseudovariety on $n$ generators
is their homomorphic image. To obtain relatively free structures,
one needs to step away from finiteness into the more general framework
of profinite semigroups, and indeed such a tool has been shown to lead
to useful insights and has found many applications
\cite{Almeida:1994a, Almeida:2003cshort, Almeida&Volkov:2001a,
  Pin:2009S, Weil:2002,Rhodes&Steinberg:2009qt}.

As topological semigroups, relatively free profinite semigroups $S$ over
a finite alphabet $A$ are generated by $A$, which means that elements of $S$ are
arbitrarily well approximated by words in the letters of~$A$.
Thus, the elements of~$S$ may be
considered a sort of generalization of words on the alphabet~$A$, which
are sometimes called pseudowords.  Of course, $S$ may satisfy nontrivial
identities, which means that different words may represent the same
element of~$S$, although in the most interesting examples of
pseudovarieties, this is not the case.  Now, words on the alphabet $A$ may
be naturally viewed as $A$-labeled finite linear orders, a perspective
that underlies many fruitful connections with finite model theory
\cite{Thomas:1997}.
For some pseudovarieties, such as
\pv R, of all finite \Cl R-trivial semigroups, and \pv{DA}, of all finite
semigroups in which the idempotents are the only regular elements,
representations of the corresponding finitely generated
relatively free profinite semigroups
by labeled linear orders have been obtained and significantly applied
\cite{Almeida&Weil:1996b,Moura:2009a}. The purpose of this paper is to
investigate such a linear nature of pseudowords for pseudovarieties with
suitable properties. Our main motivation is to understand pseudowords over
the pseudovariety \pv A, of all finite aperiodic semigroups.

The key properties of the pseudovariety \pv A that play a role in this
paper are of a combinatorial nature: the corresponding variety of
languages is closed under concatenation and the cancelability of first
and last letters. The first of these properties entails a very useful
feature of the corresponding finitely generated relatively free profinite semigroups,
namely equidivisibility, which means that different factorizations of
the same pseudoword have a common refinement. This condition already
forces a linear quasi-order on the factorizations of a given
pseudoword, and this is the starting point for the whole paper. The
cancelability condition leads to special types of factorizations,
which we call step points, to which a letter is naturally associated.
The corresponding linear order has interesting order and topological
properties, such as being compact for the interval topology. The step
points are the isolated points and there are only countably many of
them. All other points are called stationary and, in contrast, there
may be uncountably many of them. Perhaps somewhat surprisingly, there
is no correlation between the number of stationary points and how low
pseudowords fall in the \Cl J-order.

Our main result is that the linear order of factorizations with
alphabet-labeled step points provides a faithful representation of
pseudowords over~\pv A. We also obtain a characterization of the
partially labeled linear orders that appear in this way, albeit in
terms of properties involving finite aperiodic semigroups. A natural
goal for future work consists in looking for a characterization of the
image of the representation which is independent of such semigroups,
as has been done in the case of the pseudovarieties \pv R and
\pv{DA}~\cite{Almeida&Weil:1996b,Moura:2009a}.

While this paper was being written, Gool and Steinberg developed a
different approach on the pseudowords over \pv A, applying Stone
duality and model theory to view them as elementary equivalence
classes of labeled linear orders~\cite{Gool&Steinberg:2016}. They
worked specially with saturated models. In our paper, the models that
appear in the image of the representation are not saturated in
general.

We also mention the articles~\cite{Huschenbett&Kufleitner:2014}
and~\cite{Kufleitner&Wachter:2016}, where labeled linear orders were
assigned only to a special class of pseudowords, the $\omega$-terms,
and were used to solve the word problem for $\omega$-terms in several
pseudovarieties, either for the first time, or with new proofs, as in
the case of \pv A, treated in \cite{Huschenbett&Kufleitner:2014}.

The paper is organized as follows. After a section of preliminaries,
Section~\ref{sec:class-exampl-relat-1} introduces the key notion of
equidivisible semigroup in the context of relatively free profinite
semigroups, with an emphasis on pseudovarieties closed under
concatenation. Several results of the paper apply to all such
pseudovarieties, but at a certain point our hypothesis restricts
to~\pv A. In the next four sections, we develop more on the tools and
the language necessary for the main results. In
Section~\ref{sec:cluster-words}, we give our faithful representation
of pseudowords over \pv A as labeled linear orders. The following
three sections relate to the proof of this representation (the first
two of them having independent interest). This is followed by a study
of the effect of the multiplication in the image of the
representation, and by a characterization of the image. The paper
closes with Section~\ref{sec:card-prot-l_2w} where, among other
things, it is shown that the ordered set of the real numbers can be
embedded in the ordered set of the stationary points of a pseudoword
over a finitely cancelable pseudovariety containing $\pv {LSl}$. This
is done via a connection with symbolic dynamics.

\section{Preliminaries}
\label{sec:preliminaries}

We assume some familiarity with pseudovarieties of semigroups and
relatively free profinite
semigroups~\cite{Almeida:2003cshort,Almeida:1994a,Rhodes&Steinberg:2009qt}.
For the reader's convenience, some notation and terminology is
presented here. The following is a list of some of the pseudovarieties
we will be working with:
\begin{itemize}
\item $\pv I$: all trivial semigroups;
\item $\pv S$: all finite semigroups;
\item $\pv A$: all finite aperiodic semigroups;
\item $\pv N$: all finite nilpotent semigroups;
\item $\pv D$: all finite semigroups in which the idempotents are
  right zeros;
\item $\pv {LSl}$: all finite local semilattices.
\end{itemize}
In the whole paper, $A$ denotes a finite alphabet.
Let $\pv V$ be a pseudovariety of semigroups. The free pro-$\pv V$ semigroup generated by $A$ is denoted $\Om AV$.
Its elements are \emph{pseudowords} over $\pv V$.
When $\pv V\neq\pv I$, as the associated  generating mapping $A\to\Om AV$ is
injective, one considers $A$ to be contained in $\Om AV$.
If $\varphi\colon A\to S$ is a generating mapping of a pro-$\pv V$
semigroup, then we denote by $\varphi_{\pv V}$ 
the unique continuous homomorphism $\Om AV\to S$ extending $\varphi$.

If $\pv V$ contains $\pv N$, then the subsemigroup of
$\Om AV$ generated by $A$ is isomorphic to $A^+$ and its elements are
the isolated points of $\Om AV$, in view of which
$A^+$ is considered to be contained in $\Om AV$, and the elements of $A^+$
and $\Om AV\setminus A^+$ are respectively called
the \emph{finite} and \emph{infinite} pseudowords over $\pv V$.

By a \emph{topological semigroup}, we mean a semigroup endowed with a
topology that makes the semigroup multiplication continuous. Unlike
some authors, we require that a compact space be Hausdorff. By a
\emph{compact semigroup}, we mean a compact topological semigroup.
See~\cite{Carruth&Hildebrant&Koch:1983}.

We denote by $S^I$ the monoid obtained from the semigroup~$S$ by
adjoining to $S$ an element denoted by $1$ which acts as the identity.
Every semigroup homomorphism $\varphi\colon S\to T$ is extended to
a semigroup homomorphism $S^I\to T^I$, also denoted $\varphi$, such that $\varphi(1)=1$.
If $S$ is a topological semigroup, then $S^I$
is viewed as a topological monoid whose topology is the sum
of the topological spaces $S$ and $\{1\}$, whence $1$ is an isolated
point of~$S^I$.

We use the standard notation for Green's relations and its
quasi-orders on a semigroup $S$.
Hence, $s\leq_{\Cl R}t$,
$s\leq_{\Cl L}t$
and $s\leq_{\Cl J}t$
respectively mean $s\in tS^I$, $s\in S^It$ and $s\in S^ItS^I$,
$\Cl R$, $\Cl L$, $\Cl J$ 
are the associated equivalence relations, $\mathcal
D=\mathcal R\vee\mathcal L$, and $\mathcal H=\mathcal R\cap\mathcal L$.

A semigroup $S$ has \emph{unambiguous $\leq_{\Cl L}$-order} if, for every
$x,y,z\in S$,  $x\leq_{\Cl L}y$ and $x\leq_{\Cl L}z$ implies
$y\leq_{\Cl L}z$ or $z\leq_{\Cl L}y$. One also has the dual notion of
\emph{unambiguous $\leq_{\Cl R}$-order}.  An \emph{unambiguous} semigroup
is a semigroup with unambiguous $\leq_{\Cl R}$-order and unambiguous
$\leq_{\Cl L}$-order.  The next proposition is an important tool to show
one of our main results.

\begin{Prop}\label{p:residual-B-semigroups}
  Let $A$ be a finite alphabet.
  Let $u,v\in\Om AA$. Then $u=v$
  if and only if $\varphi_{\pv A}(u)=\varphi_{\pv A}(v)$,
  for every mapping $\varphi$ from $A$
  onto an unambiguous finite aperiodic semigroup.
\end{Prop}

\begin{proof}
  The ``only if'' direction of the statement is immediate.
  To establish the ``if'' direction, it suffices to show that it is
  the inverse limit of $A$-generated unambiguous finite aperiodic
  semigroups.

  It is well known that every $A$-generated finite aperiodic semigroup
  is a homomorphic image of an unambiguous $A$-generated finite
  aperiodic semigroup, namely its \emph{Birget-Rhodes expansion} (also
  called \emph{iterated Rhodes expansion}), cut down to the set of
  generators~$A$~\cite{Birget:1984,Grillet:1995}. Since pairs of distinct
  points of~\Om AA may be separated by continuous homomorphisms into
  finite aperiodic semigroups, the result follows.  
\end{proof}

\section{Equidivisibility  and pseudovarieties closed under concatenation}
\label{sec:class-exampl-relat-1}

A language $L\subseteq A^+$ is said to be \emph{\pv V-recognizable} if
there is a homomorphism $\varphi:A^+\to S$ into a semigroup $S$
from~\pv V such that $L=\varphi^{-1}\varphi(L)$. We say that a
pseudovariety $\pv V$ of semigroups is \emph{closed under
  concatenation} if, for every finite alphabet~$A$, whenever $L$ and
$K$ are $\pv V$-recognizable languages of $A^+$,
the set $LK$ is
also a $\pv V$-recognizable language of $A^+$.

\begin{Thm}\label{t:factorizing}
  The following conditions are equivalent for a pseudovariety~\pv V of
  semigroups.
  \begin{enumerate}
  \item\label{item:factorizing-1} \pv V is closed under concatenation;
  \item\label{item:factorizing-2} $\pv A\malcev\pv V=\pv V$;
  \item\label{item:factorizing-3}
    \pv V contains \pv N and the multiplication in \Om AV is an
    open mapping for every finite alphabet~$A$.
  \end{enumerate}
\end{Thm}

The
equivalence~\ref{item:factorizing-1}$\Leftrightarrow$\ref{item:factorizing-3}
in Theorem~\ref{t:factorizing} is
from~\cite[Lemma~2.3]{Almeida&ACosta:2007a}. The difficult part of the
theorem is the equivalence
\ref{item:factorizing-1}$\Leftrightarrow$\ref{item:factorizing-2},
which is a particular case of a more general result established by
Chaubard, Pin and Straubing \cite{Chaubard&Pin&Straubing:2006}. The
latter, in turn, extends an earlier result of
Straubing~\cite{Straubing:1979a}, establishing that a nontrivial
pseudovariety $\pv V$ of \emph{monoids} satisfies $\pv A\malcev\pv
V=\pv V$ if and only if, for every finite alphabet~$A$, whenever $L$
and $K$ are $\pv V$-recognizable languages of $A^\ast$, the set $LK$
is also a $\pv V$-recognizable language of $A^\ast$. In the case of
semigroups, the absence in Theorem~\ref{t:factorizing} of reference to
the pseudovariety $\pv I$ of trivial semigroups is not surprising if
we take into account that $A^+$ is $\pv I$-recognizable but not $A^+
A^+$, where we view these languages as languages of $A^+$.
 
Sch\"utzenberger~\cite{Schutzenberger:1965} proved that a language
over a finite alphabet is \pv A-recognizable if and only if it is
star-free, in the sense that it can be obtained from finite languages
by using only finite Boolean operations and concatenation. In
particular, it follows that \pv A is closed under concatenation. As
important classes of examples of pseudovarieties closed under
concatenation that include $\pv A$, one has the complexity
pseudovarieties~$\pv C_n$ (cf.~\cite[Definition
4.3.10]{Rhodes&Steinberg:2009qt})
and every
pseudovariety $\overline{\pv H}$ formed by the finite semigroups whose
subgroups belong to the pseudovariety of groups $\pv H$.

Combined with Theorem~\ref{t:factorizing},
the next lemma, which will be quite useful in the sequel,
provides yet another characterization
of the pseudovarieties closed under concatenation.
A weaker version of the direct implication
was proved in~\cite[Lemma~2.5]{Almeida&ACosta:2007a}.

\begin{Lemma}\label{l:a-condition-on-metric-semigroups}
  Let $S$ be a topological semigroup whose
  topology is defined by a  metric.
  The following conditions
  are equivalent:
  \begin{enumerate}
  \item \label{item:a-condition-on-metric-semigroups-1}
    The multiplication in $S$ is an open mapping;
  \item \label{item:a-condition-on-metric-semigroups-2} For every
    $u,v\in S$, if $(w_n)_n$ is a sequence of elements of $S$
    converging to $uv$, then there are sequences $(u_n)_n$ and
    $(v_n)_n$ of elements of $S^I$ such that $w_n=u_nv_n$, $\lim
    u_n=u$, and $\lim v_n=v$.
  \end{enumerate}
\end{Lemma}

\begin{proof}
  Consider a metric $d$ inducing the topology of $S$.
  We denote by $B(t,\varepsilon)$ the open
  ball in $S$ with center $t$ and radius $\varepsilon$.

  \ref{item:a-condition-on-metric-semigroups-1}$\Rightarrow$\ref{item:a-condition-on-metric-semigroups-2}:
  Let $k$ be a positive integer.
  Since the  multiplication is an open mapping, the
  set $B\bigl(u,\frac{1}{k}\bigr)
  B\bigl(v,\frac{1}{k}\bigr)$ is an open neighbourhood 
  of $uv$. Hence there is $p_k$ such that
  $w_n\in B\bigl(u,\frac{1}{k}\bigr)
    B\bigl(v,\frac{1}{k}\bigr)$
    if $n\geq p_k$.
  Let $n_k$ be the strictly increasing sequence recursively defined by
  $n_1=p_1$ and  $n_k=\max\{n_{k-1}+1,p_k\}$ whenever $k>1$.
  Then there are sequences $(u_n)_n$ and $(v_n)_n$ satisfying the following
  conditions: if $n_k\leq n<n_{k+1}$ then
  $u_n\in B\bigl(u,\frac{1}{k}\bigr)$,
  $v_n\in B\bigl(v,\frac{1}{k}\bigr)$, and $w_n=u_nv_n$;
  and if $n<n_1$ then $u_n=1$ and $v_n=w_n$.
  The pair of sequences $(u_n)_n$ and $(v_n)_n$ satisfies
  condition~\ref{item:a-condition-on-metric-semigroups-2}.

  \ref{item:a-condition-on-metric-semigroups-2}$\Rightarrow$\ref{item:a-condition-on-metric-semigroups-1}:
  We want to prove that $B(s,\varepsilon)B(t,\varepsilon)$ is open,
  for every $s,t\in S$ and $\varepsilon>0$. Let $(w_n)_n$ be a
  sequence of elements of $S$ converging to an element of
  $B(s,\varepsilon)B(t,\varepsilon)$. Let $u\in B(s,\varepsilon)$ and
  $v\in B(t,\varepsilon)$ be such that $\lim w_n=uv$. Take sequences
  $(u_n)_n$ and $(v_n)_n$ as in the statement of
  condition~\ref{item:a-condition-on-metric-semigroups-2}. There is
  $N$ such that $d(u_n,u)<\varepsilon-d(u,s)$ for all $n\geq N$. Then
  $d(u_n,s)\leq d(u_n,u)+d(u,s)<\varepsilon$ for all $n\geq N$.
  Similarly, $d(v_n,t)<\varepsilon$ for all sufficiently large $n$.
  Therefore, since $w_n=u_nv_n$, we have $w_n\in
  B(s,\varepsilon)B(t,\varepsilon)$ for all sufficiently large $n$,
  which proves that $B(s,\varepsilon)B(t,\varepsilon)$ is open.
\end{proof}

A semigroup $S$ is said to be
\emph{equidivisible}~\cite{McKnight&Storey:1969,Lallement:1979} if,
for every equality of the form $xy=uv$, with $x,y,u,v\in S$, there
exists $t\in S^I$ such that, either $xt=u$ and $y=tv$, or $x=ut$ and
$ty=v$. Clearly, free semigroups and groups are equidivisible.
Moreover, all completely simple semigroups are equidivisible.
Actually, a semigroup $S$ is completely simple if and only if, for
every $x,y,u,v\in S$ such that $xy=uv$, there are $t,s\in S$ such that
$xt=u$, $y=tv$, $x=us$ and $sy=v$~\cite{McKnight&Storey:1969}. Note
that every equidivisible semigroup is unambiguous. The converse is not
true: for instance, free bands are unambiguous, which follows easily
from the solution of the word problem for free bands (see, for
instance~\cite[Section~5.4]{Almeida:1994a}) but not equidivisible for
more than one free generator since, if $a,b$ are two distinct free
generators in a free band then, for $x=a, y=b, u=v=ab$, we have
$xy=uv$, yet $y>_{\mathcal L}v$ and $x>_{\mathcal R}u$.
More generally, it is shown in~\cite[Section
15]{Rhodes&Steinberg:2002} that, if $\pv V$ is a pseudovariety of
semigroups such that $\pv V=\pv {RB}\malcev \pv V$, where $\pv {RB}$
denotes the pseudovariety of finite rectangular bands, then $\Om AV$
is unambiguous, for every finite alphabet~$A$.

Let us say that a pseudovariety of semigroups $\pv V$ is
\emph{equidivisible} if $\Om AV$ is equidivisible, for every finite
alphabet $A$. The following result was established
by the first two authors~\cite{Almeida&ACosta:2016a}, where $\malcev$
denotes the Mal'cev product, \pv{LI} the pseudovariety of all 
finite locally trivial semigroups, and \pv{CS} the pseudovariety of
all finite completely simple semigroups.

\begin{Thm}\label{t:characterization-of-equidivisible-pseudovarieties}
  A pseudovariety of semigroups $\pv V$ is equidivisible if and only
  if\/  $\pv V=\pv{LI}\malcev\pv V$ or $\pv V\subseteq\pv{CS}$.
\end{Thm}

In particular, every
pseudovariety closed under concatenation is equidivisible. Many of our
results below are formulated not in terms of pseudovarieties but more
abstractly for free profinite semigroups with suitable properties, which
are satisfied for free profinite semigroups over pseudovarieties that are
closed under concatenation or, sometimes, more generally,
equidivisible.

\section{The quasi-order of 2-factorizations}
\label{sec:linear-orders}

By a \emph{quasi-order} on a set we mean a reflexive transitive
relation. In case the relation is additionally anti-symmetric, the
quasi-order is called a \emph{partial order}. A quasi-ordering
$(X,\leq)$, in the sense of a set $X$ with a quasi-order $\le$, is
said to be \emph{total}, or \emph{linear} if $x\leq y$ or $y\leq x$,
for every $x,y\in X$.

\subsection{Definition and properties}
\label{sec:defin-prop}

Let $S$ be a semigroup.
A \emph{$2$-factorization} of $s\in S$ is a pair $(u,v)$ of elements of $S^I$
such that $s=uv$.
We denote the set of $2$-factorizations of $s$ by $\mathfrak{F}(s)$.
We introduce in $\mathfrak{F}(s)$
a relation $\leq$ defined by $(u,v)\leq(u',v')$ if there exists $t\in  S^I$ such
that $ut=u'$ and $v=tv'$, in which case we say that $t$
is a \emph{transition} from $(u,v)$ to $(u',v')$.
The relation $\leq$ is a quasi-order. Concerning transitivity,
we have more precisely that if $t$ is a transition from
$(u,v)$ to $(u',v')$
and $t'$ is a transition from $(u',v')$ to $(u'',v'')$,
then $tt'$ is a transition from $(u,v)$ to $(u'',v'')$.

Given a quasi-order $\leq$ on a set $P$, we denote by $\sim$
the equivalence relation on~$P$ induced by $\leq$
and we write $p<q$ if $p\leq q$ but not $p\sim q$.
Denote
by $\prec$ the relation on $P$ such that $p\prec q$
if and only if
$q$ is a \emph{successor} of $p$ (equivalently, $p$ is
a \emph{predecessor} of $q$), that is,
$p\prec q$ if and only if $p<q$
and
$p\leq r\leq q \Rightarrow (r\sim p\vee r\sim q)$.

For every element $s$ of a semigroup $S$, the quotient set
$\mathfrak{F}(s)/{\sim}$ is denoted $\mathfrak{L}(s)$. We denote the
quotient mapping $\mathfrak{F}(s)\to \mathfrak{L}(s)$ by~$\chi$. The
partial order on $\mathfrak{L}(s)$ induced by the quasi-order $\leq$ on
$\mathfrak{F}(s)$ is also denoted by $\leq$. For
$p,q\in\mathfrak{L}(s)$, we also write $p\prec q$ if $p$ is a
predecessor of~$q$. Sometimes we will also consider the unions
$\mathfrak {F}(S)=\bigcup_{s\in S}\mathfrak {F}(s)$ and
$\mathfrak {L}(S)=\bigcup_{s\in S}\mathfrak {L}(s)$.

The following result is immediate.

\begin{Lemma}
  \label{l:linear-quasi-ordering}
  A semigroup $S$ is equidivisible if and only if $\mathfrak{L}(s)$ is
  linearly ordered for every $s\in S$.
\end{Lemma}

The previous lemma
is the departing point
motivating this paper. For a good supporting reference on the theory
of linear orderings, see~\cite{Rosenstein:1982}.

We proceed to extract
from topological assumptions on $S$
some consequences on the quasi-order of
$2$-factorizations.
In what follows, $\mathfrak {F}(s)$ is viewed as a topological
subspace of $S^I\times S^I$.

\begin{Lemma}\label{l:if-S-is-compact-then-the-quasi-order-is-closed}
  If $S$ is a compact semigroup, then, for every
  $s\in S$, the quasi-order $\leq $ on $\mathfrak {F}(s)$
  is a closed subset of $\mathfrak {F}(s)\times \mathfrak {F}(s)$.
\end{Lemma}

\begin{proof}
  Suppose $(p_i,q_i)_{i\in I}$ is a convergent net of elements of
  $\mathfrak {F}(s)\times \mathfrak {F}(s)$ with limit $(p,q)$ and
  such that $p_i\leq q_i$ for every $i\in I$. Then, for each $i\in I$,
  there is $t_i\in S^I$ making a transition from $p_i$ to $q_i$. Since
  $S^I$ is compact, the net $(t_i)_{i\in I}$ has a subnet converging
  to some $t\in S^I$.
  Then, by continuity of multiplication on $S^I$, one deduces that
  indeed $p\leq q$, with $t$ being a transition from $p$ to $q$.
\end{proof}

We shall denote the open intervals of a quasi-ordered
set $P$ by
\begin{equation*}
  {]}{\leftarrow},p{[}=\{r\in P: r<p\},
  \quad
  {]}p,{\rightarrow}{[}=\{r\in P: p<r\},
  \quad
  {]}p,q{[}={]}p,{\rightarrow}{[}\cap {]}{\leftarrow},q{[},
\end{equation*}
for every $p,q\in P$. Considering the relation $\leq$, we also have
the intervals of the form ${]}{\leftarrow},p{]}=\{r\in P: r\leq
p\}$, ${[}p,{\rightarrow}{[}=\{r\in P: p\leq r\}$, and so on.
Recall that the \emph{order topology} of a linearly ordered set $P$ is
the topology with subbase the sets of the form
${]}{\leftarrow},p{[}$ and ${]}p,{\rightarrow}{[}$. In particular, we
consider the order topology on $\mathfrak{L}(s)$.

\begin{Prop}
  \label{p:topology-of-Fw}
  Let $S$ be a compact equidivisible semigroup. For every $s\in S$,
  the mapping $\chi\colon \mathfrak{F}(s)\to \mathfrak {L}(s)$ is
  continuous.
\end{Prop}

\begin{proof}
  It is sufficient to show that the sets of both forms
  $\chi^{-1}({]}{\leftarrow},q{[})$ and
  $\chi^{-1}({]}q,{\rightarrow}{[})$ are open. By duality, we are
  actually reduced to show that $\chi^{-1}({]}q,{\rightarrow}{[})$ is
  open. Since $\mathfrak{L}(s)$ is linearly ordered by
  Lemma~\ref{l:linear-quasi-ordering},
  the complement of this last set is
  $\chi^{-1}({]}{\leftarrow},q{]})$, which we therefore want to show
  to be closed. Consider a net $(r_i)_{i\in I}$ of elements of
  $\chi^{-1}({]}{\leftarrow},q{]})$, converging to some $r\in
  \mathfrak{F}(s)$. Let $\hat q\in\chi^{-1}(q)$. Then $r_i\leq \hat q$
  for every $i\in I$. It follows from
  Lemma~\ref{l:if-S-is-compact-then-the-quasi-order-is-closed} that
  $r\leq \hat q$, that is, $r\in \chi^{-1}({]}{\leftarrow},q{]})$,
  showing that $\chi^{-1}({]}{\leftarrow},q{]})$ is closed.
\end{proof}

\begin{Cor}
  \label{c:Ls-is-compact}
  Let $S$ be a compact equidivisible semigroup. Then, for every $s\in S$,
  the order topology of $\mathfrak{L}(s)$ is compact.
  Moreover, if the space $S$ is metrizable, then
  the space $\mathfrak{L}(s)$ is also metrizable and
  the set of isolated points of $\mathfrak{L}(s)$ is countable.
\end{Cor}

\begin{proof}
  Since $S$ is compact, $\mathfrak {F}(s)$ is compact, being the
  preimage in $S^I\times S^I$ under multiplication of the closed set
  $\{s\}$, it is a closed subset of a compact space, whence compact.
  Since $\mathfrak{L}(s)$ is clearly
  Hausdorff, it follows from Proposition~\ref{p:topology-of-Fw} that $\mathfrak
  {L}(s)$ is compact.

  Suppose that $S$ is metrizable. Then $\mathfrak {F}(s)$ is
  metrizable, being a subspace of a product of two metrizable spaces.
  Since the continuous image of a compact metric space in a Hausdorff
  space is metrizable (\cite[Corollary 23.2]{Willard:1970}), it also follows
  from Proposition~\ref{p:topology-of-Fw} that $\mathfrak {L}(s)$ is
  metrizable.
  As a compact metrizable space, $\mathfrak {L}(s)$ has a dense
  countable subset.
  Since isolated points belong to every dense subset, they
  form a countable set.
\end{proof}

Recall that a linearly ordered set $L$ is said to be \emph{complete}
if every subset of $L$ which is bounded above has a least upper bound
(i.e., a supremum) or, equivalently, if every subset of $L$ which is
bounded below has a greatest lower bound (i.e., an
infimum)~\cite[Section 2.4]{Rosenstein:1982}.
  
\begin{Prop}
  \label{p:complete-quasi-ordering}
  Suppose $S$ is a compact equidivisible semigroup.
  Then the linearly ordered set $\mathfrak{L}(s)$ is
  complete.
\end{Prop}

\begin{proof}
  Let $X$ be a subset of $\mathfrak{L}(s)$. Consider the subset $Y$
  of $\mathfrak{L}(s)$ of lower bounds of~$X$ and assume that it is
  nonempty.
  As $S$ is equidivisible, we know by
  Lemma~\ref{l:linear-quasi-ordering} that the quasi-order $\leq$ on
  the set $\chi^{-1}(U)$ is linear, whence in particular this set is
  directed.
  Therefore, we may consider, in the product space
  $S^I\times S^I$, the net $(p)_{p\in \chi^{-1}(U)}$. By compactness,
  this net has a subnet $(p_i)_{i\in I}$ converging to some element
  $q$ of $\mathfrak{F}(s)$.
  We claim that $\chi(q)=\max U$. If $r\in
  \chi^{-1}(X)$, then we have $p\leq r$ for all $p\in \chi^{-1}(U)$,
  and so $q\leq r$, by
  Lemma~\ref{l:if-S-is-compact-then-the-quasi-order-is-closed},
  showing that $q\in \chi^{-1}(U)$. On the other hand, if $p\in
  \chi^{-1}(U)$, then, by the definition of subnet, there exists
  $i_0\in I$ such that $p\leq p_i$ for all $i\geq i_0$. Hence we have
  $p\leq q$, again by
  Lemma~\ref{l:if-S-is-compact-then-the-quasi-order-is-closed}. This
  proves the claim that $\chi(q)=\max U$, and so $\chi(q)$ is the
  infimum of $X$.
\end{proof}

\subsection{The category of transitions}
\label{sec:categ-fact-trans}

A directed graph with vertex set~$V$ and edge set~$E$, which are
assumed to be disjoint, is given by mappings $\alpha,\omega:E\to V$
assigning to each edge $s$ its source $\alpha(s)$ and
its target $\omega(s)$.
A \emph{semigroupoid} is a directed graph, with a nonempty set of edges, endowed with
a partial associative binary operation on the set of its edges such that if $s$ and $t$ are edges, then
$st$ is defined if and only if
$\omega(s)=\alpha(t)$, in which case
$\alpha(st)=\alpha(s)$ and $\omega(st)=\omega(t)$.

Semigroupoids can be viewed as generalizations of semigroups, which in
turn can be viewed as one-vertex semigroupoids.
In particular, Green's relations generalize straightforwardly to Green's relations between the edges in semigroupoids.
For instance, in a semigroupoid $S$, $s\leq_{\Cl J}t$ means that the edge $t$ is a factor of the edge $s$ and $s\mathrel{\Cl J}t$ means that $s\leq_{\Cl J}t$ and $t\leq_{\Cl J}s$.
A \emph{subsemigroupoid} of the semigroupoid $S$ is a subgraph $T$ of $S$, with a nonempty set of edges, such that
$s,t\in T$ implies $st\in T$ whenever $\omega(s)=\alpha(t)$.
Also, an \emph{ideal} of a semigroupoid $S$
is a subsemigroupoid $I$ of $S$ such that for every $t\in I$ and every
  $s\in S$, $\omega(s)=\alpha(t)$ implies $st\in I$, 
  and $\omega(t)=\alpha(s)$ implies $ts\in I$.

A \emph{category} is a semigroupoid such that,
for each vertex $v$, there is a loop $1_v$ at $v$
satisfying $1_vs=s$ and $t1_v=t$ for every
edge $s$ starting in $v$ and every edge $t$ ending in $v$.
This coincides with the notion of small category from Category Theory, except that we compose in the opposite direction. In doing so, we are following a common convention in Semigroup Theory, see for example~\cite{Tilson:1987}.

If the sets of edges and vertices
of a semigroupoid are both endowed with compact topologies,
for which the semigroupoid operation and the mappings
$\alpha$ and $\omega$ are continuous,
then the semigroupoid is said to be \emph{compact}.

Let $S$ be an arbitrary semigroup.
To each $s\in S$, we associate a category~$\Cl T(s)$,
the~\emph{category of transitions for $s$}, as follows:
\begin{enumerate}
\item the set of vertices of $\Cl T(s)$ is $\mathfrak {F}(s)$;
\item we have an edge
  $(u,v,t,x,y)$ from $(u,v)$ to $(x,y)$, which we may denote
  $(u,v)\xrightarrow{t}(x,y)$,
  if $t$ is a transition
  from $(u,v)$ to $(x,y)$ (thus implying $(u,v)\leq (x,y)$); we say
  that $t$ is the \emph{label} of the edge;
\item  multiplication of consecutive edges is done by multiplying
their labels, that is, the product of
$(u_1,v_1)\xrightarrow{t_1}(u_2,v_2)$
an
$(u_2,v_2)\xrightarrow{t_2}(u_3,v_3)$
is $(u_1,v_1)\xrightarrow{t_1t_2}(u_3,v_3)$.
\end{enumerate}

Note that the sets of vertices of the strongly connected components of the
category $\Cl T(s)$ are precisely the $\sim$-classes of $\mathfrak
F(s)$.

The \emph{category of transitions for $S$},
denoted $\Cl T(S)$, is the coproduct category~$\bigcup_{s\in S}\Cl T(s)$.
We denote by $\Lambda$ the faithful functor $\Cl T(S)\to S^I$
mapping each edge $(u,v,t,x,y)$ to~$t$. We say that $\Lambda$ is
the \emph{labeling} functor associated to $\Cl T(S)$.
We remark that if $S$ is a compact semigroup, then $\Cl T(S)$ is a
compact category, with the vertex and edge sets
respectively endowed with the subspace topology of $(S^I)^2$ and of
$(S^I)^5$. Note that $\Lambda$ is continuous.

Suppose that in $\mathfrak{L}(s)$ we have $p\leq q$.
An element $t\in S^I$ will be called a
\emph{transition} from $p$ to $q$ if $t$ is a transition from an element
of $p$ to an element of $q$, in which case we use
the notation $p\xrightarrow{t}q$.

For future reference, it is convenient to
register the following remark,
concerning the relationship between $\Cl T(u)$ and $\Cl T(uv)$.

\begin{Remark}
  \label{r:trivial-preservation}
  Let $u,v$ be elements of a semigroup $S$.
  If $(\alpha,\beta)\xrightarrow t(\gamma,\delta)$
  is an edge of~$\Cl T(u)$,
  then $(\alpha,\beta v)\xrightarrow t(\gamma,\delta v)$ is an edge
  of~$\Cl T(uv)$.
\end{Remark}

This remark is applied in the proof of the following lemma,
which in turn will later be used in the proof of Theorem~\ref{t:J-class-type-2}.

\begin{Lemma}
  \label{l:suffix-transition}
  Let $S$ be an equidivisible semigroup.
  Consider two edges $\sigma$ and $\tau$
  of $\Cl T(S)$
  with the same target
  and such that $\alpha(\sigma)<\alpha(\tau)$.
  Then the label of $\sigma$ is a suffix of the label of $\tau$.
\end{Lemma}

\begin{proof}
  Let $\sigma$ be the edge
  $(\alpha,\beta)\xrightarrow t(\varepsilon,\varphi)$ and
  $\tau$ be the edge
  $(\gamma,\delta)\xrightarrow z(\varepsilon,\varphi)$,
  with $(\alpha,\beta)<(\gamma,\delta)$. The following equalities
  hold: $\varepsilon=\alpha t=\gamma z$, $\beta=t\varphi$, and
  $\delta=z\varphi$. From the equality $\alpha t=\gamma z$ and by
  equidivisibility, we deduce that if $z$ is not a suffix of $t$, then
  there exists $s\in S$ such that $\gamma s=\alpha$ and $st=z$, that
  is, we have an edge $(\gamma,z)\xrightarrow s(\alpha,t)$ in $\Cl
  T(\varepsilon)$. By Remark~\ref{r:trivial-preservation}, there is an
  edge $(\gamma,\delta)\xrightarrow s(\alpha,\beta)$ in~$\Cl T(S)$,
  which contradicts the hypothesis that
  $(\alpha,\beta)<(\gamma,\delta)$. Hence $z$ is a suffix
  of~$t$.
\end{proof}

\section{The minimum ideal semigroupoid
  and the $\Cl J$-class
  associated to a $\sim$-class}

In a strongly connected compact semigroupoid $C$, there is an
underlying \emph{minimum ideal semigroupoid $K(C)$} which may be
defined as follows. Consider any vertex $v$ of $C$ and the local
semigroup $C_v$ of $C$ at~$v$, that is, the semigroup formed by the
loops at $v$. Then $C_v$ is a compact semigroup, and therefore it has a minimum ideal $K_v$. Let $K(C)$ be the subsemigroupoid of~$C$ with the same set of vertices of $C$ and whose edges are those edges as~$C$ which admit some (and therefore
every) element of~$K_v$ as a factor. The next lemma is folklore.

\begin{Lemma}
  \label{lem:single-J-class}
  If $C$ is a strongly connected compact semigroupoid, then $K(C)$ is
  a closed ideal of $C$ whose definition does not depend on the choice
  of~$v$. Moreover, the edges
  in $K(C)$ are $\Cl J$-equivalent, more precisely they are \Cl J-below every edge of $C$.
\end{Lemma}

Let $(u,v)\in \mathfrak {F}(S)$.
An element $z\in S^I$ \emph{stabilizes} $(u,v)$ if
$z$ labels a loop of $\Cl T(S)$ at $(u,v)$.
Note that the set $M_{(u,v)}$ of stabilizers of
$(u,v)$ is a monoid and that $M_{(u,v)}$
is the isomorphic image,
under the labeling functor $\Lambda$,
of the local
monoid of $\Cl T(S)$ at $(u,v)$.

Assume $S$ is a compact semigroup. For $p\in\mathfrak {L}(S)$, let
$\Cl T_p$ be the strongly connected component of $\Cl T(S)$ whose
vertices are the elements of $p$. We denote by $\Cl K_p$ the minimum
ideal semigroupoid $K(\Cl T_p)$. Since $\Lambda\colon{\Cl{T}}(S)\to
S^I$ is a (continuous) functor, where $S^I$ is viewed as the set of
edges of single vertex semigroupoid, in view of
Lemma~\ref{lem:single-J-class} the set of labels of edges in~${\Cl
  K}_p$ is contained in a single {\Cl J}-class of~$S^I$, which we
denote $J_p$. For every $(u,v)\in p$,
the minimum ideal of $M_{(u,v)}$, which we denote 
$I_{(u,v)}$, is the image under $\Lambda$ of the minimum ideal
of the local monoid of $\Cl T(S)$ at $(u,v)$,
whence $I_{(u,v)}\subseteq J_p$.
Note that $J_p$ is regular, since $I_{(u,v)}$ is itself regular.
The set $J_p$ can also be characterized
as the set of \Cl J-minimum transitions from $p$ to itself,
as seen in the next lemma.

\begin{Lemma}\label{l:a-trivial-characterization-of-Jp}
  Let $S$ be a compact semigroup, and let $p\in\mathfrak {L}(S)$.
  Then $t$ is a transition from $p$ to $p$
  if and only if $t$ is a factor of the elements of $J_p$.
\end{Lemma}

\begin{proof}
  Let $(u,v)\xrightarrow t (x,y)$ be a transition between elements of
  $p$. Since $(u,v)\sim (x,y)$, there is a transition
  $(x,y)\xrightarrow s (u,v)$. The loop $(u,v)\xrightarrow {ts} (u,v)$
  is a factor of every element $\varepsilon$ in the minimum ideal of
  the local monoid at $(u,v)$. Therefore, $ts$ is a factor of $\Lambda
  (\varepsilon)\in J_p$.

  Conversely, suppose that
  $t$ is factor of the elements of $J_p$.
  Then there is a loop
  $(u,v)\xrightarrow z (u,v)$
  in $\mathcal K_p$ such that $z=xty$ for some $x,y\in S^I$.
  In $\mathcal T(S)$
  we have the following path:
  $(u,v)
  \xrightarrow {x}
  (ux,tyv)
  \xrightarrow {t}
  (uxt,yv)
  \xrightarrow {y}
  (u,v) $.
  Therefore,
  $
    (ux,tyv)
  \xrightarrow {t}
  (uxt,yv)
  $
  is an edge of $\mathcal T_p$.
\end{proof}

We next give some results that further highlight the role of
idempotent stabilizers of $2$-factorizations of elements of $S$,
specially those idempotents in a $\Cl J$-class of the form $J_p$.

Recall that a semigroup is \emph{stable} if
$\mathcal{J}\cap{\leq_{\mathcal{L}}}=\mathcal{L}$ and
$\mathcal{J}\cap{\leq_{\mathcal{R}}}=\mathcal{R}$. In particular, any
compact semigroup is stable, see for
instance~\cite{Rhodes&Steinberg:2009qt}.

\begin{Lemma}
  \label{l:minimum-idempotents-labeling-loops-at-same-point}
  Let $S$ be a stable unambiguous semigroup. Let $e,f$ be idempotents stabilizing an element $(u,v)$ of $\mathfrak{F}(S)$. If $e\mathrel {\Cl J} f$ then
  $e=f$.
\end{Lemma}

\begin{proof}
  The hypothesis gives $ue=u=uf$
  and $ev=v=fv$. Since $S$ is unambiguous, from $ue=uf$ we
  get $e\leq_{\Cl L}f$ or $f\leq_{\Cl L}e$. By
  stability, as $e\mathrel{\Cl J}f$, it follows that $e\mathrel{\Cl
    L}f$. Dually, from $ev=fv$ we get $e\mathrel{\Cl R}f$. Hence $e=f$.
\end{proof}

\begin{Cor}\label{c:minimum-idempotents-labeling-loops-at-same-point}
  Let $S$ be a compact unambiguous semigroup,
  $p\in \mathfrak L(S)$ and $(u,v),(x,y)\in p$.
  The edge $(u,v)\xrightarrow t (x,y)$ of $\Cl T(S)$
  belongs to $\mathcal K_p$ if and only if
  $t\in J_p$.
\end{Cor}

\begin{proof}
  The ``only if'' part holds by definition of $J_p$.
  Conversely, suppose that $t\in J_p$.
  Denote by $\varepsilon$ the edge
  $(u,v)\xrightarrow t (x,y)$.
  As $J_p$ is regular, there
  is an idempotent $e\in J_p$ such that $t=et$.
  Since $t$ is a prefix of $v$, we have $v=ev$,
  thus we may consider the edges
  $(u,v)\xrightarrow e (ue,v)$,
  $(ue,v)\xrightarrow e (ue,v)$
  and $(ue,v)\xrightarrow t (x,y)$,
  respectively denoted by $\alpha$, $\beta$ and $\gamma$.
  Observe that $\varepsilon=\alpha\beta\gamma$,
  and so it suffices to show that the loop $\beta$ belongs
  to $\mathcal K_p$.
  The ideal $\mathcal K_p$ contains
  the minimum ideal of the local monoid of $\Cl T(S)$ at
  $(ue,v)$.
  The latter contains an idempotent, of the form
  $(ue,v)\xrightarrow f(ue,v)$ for some $f\in J_p$.
  But $f=e$ by Lemma~\ref{l:minimum-idempotents-labeling-loops-at-same-point} and therefore $\varepsilon\in\mathcal K_p$.
\end{proof}

Note that in the next lemma one does not assume that $S$ is unambiguous.

\begin{Lemma}
  \label{l:idempotents-vs-sim-classes}
  Let $S$ be a compact semigroup, and let $(u,v)\in\mathfrak F(S)$.
  Let $e$ be an idempotent stabilizing $(u,v)$.
  If $f$ is an idempotent $\Cl J$-equivalent
  to $e$, then $f$ stabilizes an element
  of the $\sim$-class $p$ of $(u,v)$.
  Moreover, if $e$ labels a loop
  of $\mathcal K_p$, then
  $f$ also labels a loop of $\mathcal K_p$.
\end{Lemma}

\begin{proof}
  If $f\mathrel {\Cl J}e$, then there are in the $\Cl J$-class of $e$
  some elements $s,t$ such that $sts=s$, $tst=t$, $st=e$, $ts=f$.
  We have the
  four edges in~$\Cl T(S)$ which are depicted in
  Figure~\ref{fig:four-edges}.
  In particular, $f$
  stabilizes a vertex $\sim$-equivalent to $(u,v)$.
  
  Denote by $\varepsilon$, $\sigma$, $\phi$, $\tau$ the edges in
  Figure~\ref{fig:four-edges} labeled by $e$, $s$, $f$, $t$,
  respectively. Since $s=es$, $t=te$, $f=ts$, we have
  $\sigma=\varepsilon\sigma$, $\tau=\tau\varepsilon$ and
  $\phi=\tau\sigma$. Therefore, if $\varepsilon$ belongs to the ideal
  $\mathcal K_p$, then all edges in Figure~\ref{fig:four-edges} belong
  to $\mathcal K_p$, and so $f$ labels a loop of $\mathcal K_p$.
\end{proof}

\begin{figure}[h]
  \centering
  $$
  \begin{gpicture}(50,16)(0,0)
    \gasset{Nframe=n,Nadjust=w}
    \node(uv)(5,8){$(u,v)$}
    \node(xy)(45,8){$(us,tv)\,$}
    \drawedge[curvedepth=5](uv,xy){$s$}
    \drawedge[curvedepth=5](xy,uv){$t$}
    \drawloop[loopangle=180](uv){$e$}
    \drawloop[loopangle=0](xy){$f$}
  \end{gpicture}
  $$  
  \caption{Edges in $\Cl T(S)$.}
  \label{fig:four-edges}
\end{figure}

\begin{Cor}\label{c:idempotents-vs-sim-classes}
  Let $S$ be a compact semigroup, and let $p\in\mathfrak L(S)$.
  Every idempotent of $J_p$ labels a loop of $\mathcal K_p$.\qed
\end{Cor}

For $e\in J_p$, denote by $p_e$ the nonempty set of elements of $p$ stabilized by~$e$.

\begin{Prop}
  \label{p:Jp-vs-labels-in-strongly-connected-component}
    Let $S$ be a compact unambiguous semigroup.
  Let $p\in\mathfrak{L}(S)$.
  Then $J_p$ is the set of labels of edges of~$\Cl K_p$.
  Moreover, if $s\in J_p$ and $e$ and
  $f$ are idempotents such that $e\mathrel{\Cl R}s\mathrel{\Cl L}f$,
  then $s$ labels an edge from~$p_e$ to~$p_f$. Moreover,
  there is a bijection $p_e\to p_f$,
  given by $\mu_s(u,v)=(us,tv)$,
  where $t$ is the unique $t\in J_p$ such that $st=e$ and $ts=f$.
\end{Prop}

\begin{proof}
  Let $s$ be an element of~$J_p$ and let $e$ and $f$ be idempotents
  such that $e\mathrel{\Cl R}s\mathrel{\Cl L}f$. Then there exists (a
  unique) $t\in J_p$ such that $st=e$ and $ts=f$, for which we have
  $e\mathrel{\Cl L}t\mathrel{\Cl R}f$.
  Let $(u,v)\in p_e$. Note that such a pair $(u,v)$ exists by
  Corollary~\ref{c:idempotents-vs-sim-classes}. Therefore, we are in
  the same situation as in the proof of
  Lemma~\ref{l:idempotents-vs-sim-classes}, with the four edges
  depicted in Figure~\ref{fig:four-edges} belonging to~$\mathcal K_p$
  by
  Corollary~\ref{c:minimum-idempotents-labeling-loops-at-same-point}.
  If there is another edge $(u,v)\xrightarrow s(x,y)$ in $\Cl K_p$
  with $(x,y)\in p_f$, then $x=us$ and $v=sy$, thus $y=fy=tsy=tv$.
  Hence, there is for each vertex in $p_e$ exactly one edge labeled $s$
  into a vertex of~$p_f$. This defines the function $\mu_s\colon
  p_e\to p_f$ such that $\mu_s(u,v)=(us,tv)$. Finally, note that
  $\mu_s$ and $\mu_t$ are mutually inverse.
\end{proof}

We finish this section with
a couple of observations concerning aperiodic semigroups, starting with the
next lemma.

\begin{Lemma}
  \label{l:disjointness-blocks-aperiodic-case}
  Let $S$ be a compact aperiodic semigroup.
  Let $p\in \mathfrak L(S)$.
  If $(u,v),(x,y)$ are elements of $p$
  stabilized by the same idempotent $e$ of $J_p$,
  then $(u,v)=(x,y)$.
\end{Lemma}

\begin{proof}
  Since $(u,v)\sim(x,y)$, there are $t,z\in S^I$ such that there
  are edges $(u,v)\xrightarrow t(x,y)$ and $(x,y)\xrightarrow z(u,v)$
  in the category $\Cl T(S)$. Then we also have edges
  as in the following picture:
  $$
  \begin{gpicture}(50,16)(0,0)
    \gasset{Nframe=n,Nadjust=w}
    \node(uv)(5,8){$(u,v)$}
    \node(xy)(45,8){$(x,y)$}
    \drawedge[curvedepth=5](uv,xy){$ete$}
    \drawedge[curvedepth=5](xy,uv){$eze$}
    \drawloop[loopangle=180](uv){$e$}
    \drawloop[loopangle=0](xy){$e$}
  \end{gpicture}
  $$
  By Lemma~\ref{l:a-trivial-characterization-of-Jp} and stability of $S$,
  we conclude that $ete$ and $eze$ are
  \Cl H-equivalent to~$e$, thus, by aperiodicity, we get $ete=eze=e$.
  By the definition of the category~$\Cl T(S)$, we deduce that
  $u=ue=x$ and $v=ey=y$.
\end{proof}

In the following result,
we have a case in which the idempotents of
$J_p$ parameterize the elements of $p$.

\begin{Prop}
  \label{p:idempotents-in-bijection-with-sim-class}
  Let $S$ be a compact and unambiguous aperiodic semigroup.
  Let $p\in\mathfrak{L}(S)$. Then there is a bijection between the $\sim$-class $p$ and
  the set of idempotents in~$J_p$, sending each $(u,v)$ to the unique idempotent $e\in J_p$ that stabilizes~$(u,v)$.
\end{Prop}

\begin{proof}
  Let $(u,v)\in p$.
  There are in $J_p$ idempotents that stabilize $(u,v)$,
  as $J_p$ contains the minimum ideal of the monoid of stabilizers of $(u,v)$.
  If $e,f$ are idempotents of $J_p$ stabilizing $(u,v)$, then $e=f$ by
  Lemma~\ref{l:minimum-idempotents-labeling-loops-at-same-point}.
  Hence, we can consider the function $\varepsilon\colon p\to J_p$
  sending $(u,v)$ to the unique idempotent
  of $J_p$ stabilizing $(u,v)$. The function $\varepsilon$
  is injective by Lemma~\ref{l:disjointness-blocks-aperiodic-case},
  and it is surjective by Corollary~\ref{c:idempotents-vs-sim-classes}.
\end{proof}

\section{Finitely cancelable semigroups}
\label{sec:finitely-cancelable}

Consider a compact semigroup $S$ generated by a closed set $A$.
Recall that, in the context of topological semigroups, that means that
every element of $S$ is arbitrarily close to products of elements of~$A$.
Note that, since $A$ is closed, we have $S=S^IA=AS^I$. Indeed, every
element of $S$ is the limit of a net of the form $(w_ia_i)_{i\in I}$,
where the $a_i\in A$ and the $w_i$ are perhaps empty products of
elements of~$A$. By compactness, we may assume that the nets
$(w_i)_{i\in I}$ and $(a_i)_{i\in I}$ converge in~$S^I$, say to $w$
and $a$, respectively. Since $A$~is closed, we conclude that $a\in A$,
which shows that $S\subseteq S^IA$.

Say that $S$ is \emph{right finitely cancelable with respect to $A$}
when, for every $a,b\in A$ and $u,v\in S^I$, 
 the equality $ua=vb$ implies
 $a=b$ and $u=v$.
 This implies $A\cap SA=A\cap AS=\emptyset$.

 Say that $S$ is \emph{right finitely cancelable}
 if it is finitely cancelable with respect to some closed generating subset $A$.
 It turns out that the set $A$ is uniquely determined by $S$, as shown next.

 \begin{Lemma}\label{l:unicity-of-special-alphabet}
   Let $S$ be a compact semigroup generated by closed subsets $A$ and $B$
   such that $A\cap SA=B\cap SB=\emptyset$. Then we have $A=B$.   
   In particular, if $S$ is right finitely cancelable with respect to $A$ and
   to $B$, then $A=B$.
 \end{Lemma}

 \begin{proof}
   Let $a\in A$. Since $S=S^IB=S^IA$, we have $a=sb$ for some
   $s\in S^I$ and $b\in B$, and $b=tc$ for some $t\in S^I$ and $c\in A$.
   We obtain the factorization $a=stc$. Since $A\cap SA=\emptyset$,
   we must have $s=t=1$, and so $a=b\in B$, showing that
   $A\subseteq B$. By symmetry, we have $B\subseteq A$.
 \end{proof}

 Say that a pseudovariety of semigroups is \emph{right finitely
   cancelable} if $\Om AV$ is right
 finitely cancelable with respect to $A$, for every finite alphabet
 $A$.

  \begin{Prop}\label{p:a-general-example-of-finitely-cancelable}
    A pseudovariety of semigroups $\pv V$
    is right finitely cancelable if and only if $\pv V=\pv D\ast\pv V$.
 \end{Prop}

 \begin{proof}
   It is observed in~\cite{Almeida&Klima:2015a} that $\pv V$ is right
   finitely cancelable if and only if, for every finite alphabet $A$,
   and for every \pv V-recognizable language $L$ of $A^+$ and $a\in
   A$, the language $La$ is also $\pv V$-recognizable.
   In~\cite{Pin&Weil:2002b} one finds a proof that this is equivalent
   to $\pv V=\pv D\ast\pv V$.
 \end{proof}

 The above definitions have obvious duals which are obtained by replacing
 \emph{right} by \emph{left}. Note that a semigroup pseudovariety $\pv
 V$ is right finitely cancelable if and only the pseudovariety
 $\pv V^{op}$ of semigroups of $\pv V$ with reversed multiplications is left finitely cancelable.
 We say that a compact semigroup is
 \emph{finitely cancelable (with respect to~$A$)}
 if it is simultaneously right and left
 finitely cancelable (with respect to~$A$). Similarly, a pseudovariety of semigroups is \emph{finitely cancelable} if it is simultaneously right and
 left finitely cancelable.

\begin{Example}\label{eg:stabilized-by-D-are-fin-cancel}
  If $\pv V$ is a semigroup pseudovariety containing some nontrivial
  monoid and such that $\pv V=\pv V\ast\pv D$, then $\pv V$ is
  finitely cancelable (cf.~\cite[Exercise 10.2.10]{Almeida:1994a}
  and~\cite[Prop.~1.60]{ACosta:2007t}).
\end{Example}

The following proposition is~\cite[Proposition~6.3]{Almeida&ACosta:2016a}.

\begin{Prop}\label{p:equid-conca-are-finitely-cancelable}
  If $\pv V$ is an equidivisible pseudovariety of semigroups
  not contained in $\pv {CS}$, then $\pv V$
  is finitely cancelable.
\end{Prop}

The next lemma is the first of a series of results in which
the hypothesis of a semigroup being finitely cancelable enables us to
get further insight into the quasi-order of $2$-factorizations.

\begin{Lemma}
  \label{l:a-successors-vs-equivalence}
  Suppose $S$ is a compact semigroup, finitely cancelable with
  respect to $A$.
  Let $u,v\in S^I$ and $a\in A$. If
  the $\sim$-class of at least one of $(ua,v)$ and $(u,av)$ is not a
  singleton, then $(u,av)\sim(ua,v)$.
\end{Lemma}

\begin{proof}
  By duality, it suffices to consider
  the case where the $\sim$-class of $p=(ua,v)$ is not a singleton.
  Let $q$ be in $p/{\sim}$ with $p\neq q$.
  As $p\leq q$, we may consider a transition $x$ from $p$ to $q$.
  Then we have $q=(uax,y)$ for some $y\in S^I$ such that $v=xy$.
  Since $q\leq p$, there is $t$  such that $ua=uaxt$ and $y=txy$.
  Because $p\neq q$, we must have $t\neq 1$,
  whence we may take $b\in A$ and $z\in S^I$ such that $t=zb$.
  Because $S$ is finitely cancelable with respect to $A$,
  from $ua=uaxt=uaxzb$ we get $a=b$ and $u=(ua)(xz)$.
  On the other hand, we have $(xz)(av)=x(za)v=xtxy=xy=v$,
  which shows that $(ua,v)\sim(u,av)$.
\end{proof}

We now turn our attention to profinite semigroups.

\begin{Prop}
  \label{p:letter_increasing_value_in_F}
  Suppose $S$ is a profinite semigroup
  generated by a closed subset $A$.
  Let $p,q\in \mathfrak F(s)$ with $p<q$. Then,
  there are $x,y\in S^I$ and $a\in A$ such that 
  $$p\leq(x,ay)<(xa,y)\leq q.$$  
\end{Prop}

\begin{proof}
  Let $p=(u,v)$ and $q=(u',v')$.
  Since $(u,v)<(u',v')$, there exists $t\in S$ such that $u'=ut$
  and $v=tv'$, and the system
  \begin{equation}
    \label{eq:letter_increasing_value_in_F-1}
    \left\{
      \begin{array}{ll}
        utX &= u\\
        Xtv' &= v'
      \end{array}
    \right.
  \end{equation}
  has no solution $X\in S$.
  By a standard compactness argument which can be found
  in the proof of~\cite[Theorem~5.6.1]{Almeida:1994a}, there is some
   continuous onto homomorphism $\varphi_0\colon S\to R$,
   with $R$ finite,
   which may be naturally extended to
   an onto continuous homomorphism $\varphi\colon S^I\to R^I$,
   and  such that the following
  system~\eqref{eq:letter_increasing_value_in_F-2} has no solution
  $X\in R$:
  \begin{equation}
    \label{eq:letter_increasing_value_in_F-2}
    \left\{
      \begin{array}{ll}
        \varphi(u)\varphi(t)X &= \varphi(u)\\
        X\varphi(t)\varphi(v') &= \varphi(v').
      \end{array}
    \right.
  \end{equation}
  Let $(t_n)_n$ be a net of elements of
  the (discrete) subsemigroup of $S$ generated
  by $A$ such that $(t_n)_n$ converges to $t$ and
  such that $\varphi(t_n)=\varphi(t)$ for all~$n$.
  Write $t_n=a_{n,0}a_{n,1}\cdots a_{n,k_n}$, with the $a_{n,i}\in
  A$. Then the following inequalities hold for $i=0,\ldots,k_n$:
  \begin{align}
    &(\varphi(ua_{n,0}\cdots a_{n,i-1}),\varphi(a_{n,i}\cdots
    a_{n,k_n}v'))
    \nonumber\\
    &\qquad\leq
    (\varphi(ua_{n,0}\cdots a_{n,i}),\varphi(a_{n,i+1}\cdots a_{n,k_n}v')).
    \label{eq:letter_increasing_value_in_F-3}
  \end{align}
  Since $\leq$~is a transitive relation and
  the non-existence of a solution to 
  \eqref{eq:letter_increasing_value_in_F-2} guarantees that the following
  strict inequality holds
  $$
  (\varphi(u),\varphi(a_{n,0}\cdots a_{n,k_n}v'))
  <
  (\varphi(ua_{n,0}\cdots a_{n,k_n}),\varphi(v')),
  $$
  we deduce that there is $i=i_n$ such that the inequality
  \eqref{eq:letter_increasing_value_in_F-3} is also strict.
  As $A$ is closed and $S$~is compact, by taking subnets we
  may assume that the net $(a_{n,i_n})_n$
  converges to some $a\in A$,
  that $\varphi(a_{n,i_n})=\varphi(a)$ for every $n$,
  and that each of the nets $t'_n=a_{n,0}\cdots
  a_{n,i_n-1}$ and $t''_n=a_{n,i_n+1}\cdots a_{n,k_n}$ converges to
  some $t',t''\in S^I$, respectively
  (in particular, this yields $t=t'at''$). Then the strict inequality
  in~\eqref{eq:letter_increasing_value_in_F-3}, with $i=i_n$, yields
  $(\varphi(ut'),\varphi(at''v'))<(\varphi(ut'a),\varphi(t''v'))$,
  which implies that
  $$p=(u,t'at''v')
  \leq(ut',at''v')
  <(ut'a,t''v')
  \leq(ut'at'',v')=q.
  $$
  Thus, it suffices to choose $x=ut'$ and $y=t''v'$ to obtain the
  inequalities of the statement of the proposition.
\end{proof}

We close this subsection with a result
regarding the existence of a successor in the quasi-ordered set
of $2$-factorizations.

\begin{Prop}\label{p:covers-in-F}
  Suppose $S$ is a profinite semigroup, finitely cancelable with
  respect to $A$.
  Let $p,q\in\mathfrak{F}(s)$ and suppose that $p<q$.
  \begin{enumerate}
  \item\label{item:covers-in-F:1}
    Consider the unique $u,v,a$ such that
    $u,v\in S^I$, $a\in A$ and $p=(u,av)$.
    If $p\prec q$, then we have $$p=(u,av)\prec(ua,v)=q.$$ 
    Moreover, the $\sim$-classes of $p$ and $q$ are singletons.
  \item\label{item:covers-in-F:2} Conversely, if $p=(u,av)$ and
    $q=(ua,v)$, where $u,v\in S^I$ and $a\in  A$, then we have $p\prec q$.
  \end{enumerate}  
\end{Prop}

\begin{proof}
  \ref{item:covers-in-F:1}
  Notice that $u,v,a$ really exist and are unique.
  Indeed, take $p=(u,w)$, with $u,w\in S^I$.
  One has $w\neq 1$, because $p<q$, and so $w=av$ for some $a\in A$ and $v\in S^I$, which are unique because $S$ is finitely cancelable with respect to $A$.   
  By  Proposition~\ref{p:letter_increasing_value_in_F},
  there are $u',v'\in S^I$ and $a'\in  A$
  such that $p\leq (u',a'v')<(u'a',v')\leq q$.
  Since $p\prec q$,
  we must have
  $p\sim (u',a'v')<(u'a',v')\sim q$.
  It then follows from Lemma~\ref{l:a-successors-vs-equivalence}
  that
  the $\sim$-classes of $(u',a'v')$ and $(u'a',v')$ are singletons,
  thus $p=(u',a'v')$ and $q=(u'a',v')$.
  By the uniqueness of $u,v$ and $a$,
  we then have $q=(ua,v)$

  \ref{item:covers-in-F:2} Assume there exists $r=(x,y)$ with
  $p<r<q$. There are $z,t$ such that $x=ut$, $av=ty$, $ua=xz$
  and $y=zv$.  If $t=1$, then $r=p$, while if $z=1$, then $r=q$,
  hence both $t$ and $z$ are different from $1$. Since $av=ty$,
  from the fact that $S$ is finitely cancelable with respect to $A$,
  it follows that there is $t'$ such that $t=at'$ and
  $v=t'y$.
  Similarly, there is $z'$ such that $z=z'a$
  and $u=xz'$.
  Therefore, $u=xz'=utz'=ua\cdot t'z'$ and
  $v=t'y=t'zv=t'z'\cdot av$. This shows that $p\sim q$, in contradiction
  with $p<q$. Hence~$p\prec q$.
\end{proof}

\section{Step points and stationary points}
\label{sec:step-points-stat}

In this section, we continue gathering important
  properties of the linear orders induced by pseudowords. We identify
  two types of elements in such orders, that we call step points and
  stationary points. Let us start by introducing these
  notions.

Let $P$ be a partially ordered set. We call
\emph{step points} the points of~$P$
that admit either a successor or a predecessor,
or are the minimum or the maximum of~$P$, if they exist.
All other points are
said to be \emph{stationary}. The set of step points of $L$
will be denoted by $\step (L)$, and the set of stationary points of $L$
will be denoted by $\stat (L)$.

For an element $s$ of a semigroup $S$, we also say that
$p\in \mathfrak {F}(s)$ is a \emph{step point}
(respectively, a \emph{stationary point}) if $\chi(p)$ is a step point
(respectively, a stationary point) of $\mathfrak {L}(s)$.
In this section we further develop the results obtained in
Section~\ref{sec:finitely-cancelable}, using the notions
of step point and stationary point.
If $S$ is profinite and finitely cancelable, then the $\sim$-class
$p$ of a step point
$(u,v)$ of $\mathfrak F(S)$
is a singleton
(cf.~Proposition~\ref{p:covers-in-F}\ref{item:covers-in-F:1}),
for which reason, in that case, we feel free to make the abuse of
notation $p=(u,v)$.

As a preparation for the following example, recall that in
a compact semigroup $S$, if $s\in S$, then $s^\omega$
denotes
the unique idempotent in the closed subsemigroup of $S$ generated by~$s$. Later on, we
shall also make use of the notation $s^{\omega+1}$ for $s^\omega s$,
and $s^{\omega-1}$ for the inverse of $s^{\omega+1}$ in the maximal
subgroup containing~$s^{\omega+1}$.

\begin{Example}\label{eg:order-aomega-aperiodic}
  Consider the pseudoword $a^\omega$ of the free pro-aperiodic
  semigroup $\Om {\{a\}}A$.
  Then $\mathfrak {F}(a^\omega)$ has only one stationary point, namely
  $(a^\omega,a^\omega)$.
  The set $\mathfrak {F}(a^\omega)$
  is linearly ordered, whence isomorphic to $\mathfrak {L}(a^\omega)$,
  with order type $\omega+1+\omega^\ast$.
  More precisely, its elements are ordered as follows:
  {\small
\begin{equation*}
(1,a^\omega)<(a,a^\omega)<(a^2,a^\omega)
<\cdots <(a^\omega,a^\omega)<\cdots
<(a^\omega,a^2)<(a^\omega,a)<(a^\omega,1).
\end{equation*}
}
\end{Example}

Example~\ref{eg:order-aomega-aperiodic} should be compared with the
following one.

\begin{Example}\label{eg:order-aomega-S}
  Consider the pseudoword $a^\omega$ of the free profinite semigroup
  $\Om AS$.
  Like in Example~\ref{eg:order-aomega-aperiodic},
  $\mathfrak {L}(a^\omega)$ has
  order type
  $\omega+1+\omega^*$,
  and its sole stationary point is $p=(a^\omega,a^\omega)/{\sim}$:
  {\small
    \begin{equation*}
      (1,a^\omega)<(a,a^{\omega-1})<(a^2,a^{\omega-2})
      <\cdots <p<\cdots
      <(a^{\omega-2},a^2)<(a^{\omega-1},a)<(a^\omega,1).
    \end{equation*}
  }  
  But $(a^\omega,a^\omega)/{\sim}$ has infinitely many elements,
  namely, the pairs of the form $(g,g^{\omega-1})$, where $g$
  is an element in the maximal subgroup containing $a^\omega$.
\end{Example}

Examples~\ref{eg:order-aomega-aperiodic}
and~\ref{eg:order-aomega-S} fit in the following definition.

\begin{Def}[Clustered sets]\label{def:clustered-sets}
  We say that the linearly ordered set $P$ is
\emph{clustered} if the following conditions hold:
\begin{enumerate}[label=(C.\arabic*)]
\item\label{item:clustered-1} $P$ has a minimum $\min P$ and a maximum
  $\max P$;
\item\label{item:clustered-2} for every $q\in P$, if $q=\min P$
  or $q$ has a predecessor,
  then $q$ has a successor or $q=\max P$;
\item\label{item:clustered-3} for every $q\in P$,
  if $q=\max P$ or $q$ has a successor,
  then $q$ has a predecessor  or $q=\min P$;
\item\label{item:clustered-4}
  for every $p,q\in P$, if ${]}p,q{[}$ is nonempty, then
  there is a step point in the interval ${]}p,q{[}$.
\end{enumerate}
\end{Def}

Property \ref{item:clustered-4} translates
into saying that the set of step points of $P$ is dense
with respect to the order topology of $P$.

\begin{Thm}\label{t:cluster}
  Let $S$ be a profinite semigroup which is finitely cancelable,
  and let $s\in S$. Then $\mathfrak {L}(s)$ is clustered.
\end{Thm}

\begin{proof}
  Property~\ref{item:clustered-1} in Definition~\ref{def:clustered-sets}
  holds trivially, with $\min \mathfrak {L}(s)=(1,s)$
  and $\max \mathfrak {L}(s)=(s,1)$.

  Let us show~\ref{item:clustered-2}.
  Take $s\in S$ and $p,q\in\mathfrak{F}(s)$ with $p\prec q$.
  We have to show that either $q$ has a successor, or $q=\max P$.
  Let~$A$ be the generating set with respect to
  which $S$ is finitely cancelable.
  By Proposition~\ref{p:covers-in-F}, there are $u,v\in S^I$
  and $a\in A$ with $p=(u,av)$ and $q=(ua,v)$.
  If $v=1$, then $q=\max P$, so that we may assume $v\neq1$.
  Let $v=bw$ with $b\in A$, and let
  $r=(uab,w)$.  Clearly, we have $q\leq r$. We claim that $q\prec
  r$. Indeed, if $q=r$, then $ua=uab$ and $bw=w$. Therefore, $a=b$,
  $u=ua$, and $av=aaw=aw=v$, showing that $p=q$, in contradiction
  with the hypothesis. Hence, we must have $q<r$, since otherwise we
  would obtain $q\sim r$ and $q\neq r$, which entails $p\sim q$ by
  Lemma~\ref{l:a-successors-vs-equivalence}.
  Finally, by Proposition~\ref{p:covers-in-F}\ref{item:covers-in-F:2} applied to $q$
  and $r$, we get $q\prec r$.
  This establishes~\ref{item:clustered-2}, and~\ref{item:clustered-3}
  holds dually.

  Finally, let us prove~\ref{item:clustered-4}.
  If $p$ is a step point,
  then ${]}p,q{[}\neq\emptyset$
  implies that the successor of $p$ belongs to ${]}p,q{[}$. Hence,
  it suffices to consider the case where $p$ is stationary.
  Let $\hat p\in\chi^{-1}(p)$
  and $\hat q\in\chi^{-1}(q)$.
  By Proposition~\ref{p:letter_increasing_value_in_F}, there are
  $u,v\in S^I$ and $a\in A$ such that
  $\hat p\leq(u,av)<(ua,v)\leq \hat q$
  in $\mathfrak {F}(s)$.
  By Proposition \ref{p:covers-in-F}\ref{item:covers-in-F:2}, we
  have $(u,av)\prec(ua,v)$.
  Therefore, $(u,av)$ and $(ua,v)$ are step points.
  In particular, we have $(u,av)\in {]}p,q{[}\cap \step (\mathfrak {L}(s))$.
\end{proof}

In the next result, we characterize the stationary points as the vertices
with a nontrivial local monoid in the category of transitions.

\begin{Prop}
  \label{p:shifting-allowed-implies-type-2}
  Let $S$ be a profinite semigroup
  which is finitely cancelable. Let $(u,v)\in\mathfrak {F}(S)$.
  Then $(u,v)$ is a stationary point if and only if
  it is stabilized by some element of $S$.
\end{Prop}

\begin{proof}
  Suppose that $(u,v)$ is stationary.
  Let $A$ be the set with respect to which $S$ is finitely cancelable.
  Since $v\neq 1$, we may take a factorization $v=aw$ with $a\in A$ and $w\in S^I$.
  Clearly,  $(u,aw)\leq (ua,w)$ holds, and so $(u,aw)\sim (ua,w)$ by
  Proposition~\ref{p:covers-in-F}\ref{item:covers-in-F:2}.
  Hence, there is an edge $(ua,w)\xrightarrow t (u,aw)$, for some
  $t\in S^I$. This implies that $at\in S$ stabilizes $(u,v)$.
  
  Conversely, suppose there is $z\in S$
  stabilizing  $(u,v)$. There is some factorization of
  the form $z=at$, for some $a\in A$ and $t\in S^I$. The following
  equalities and inequalities
  $$(u,v)=(u,atv)\leq(ua,tv)\leq(uat,v)=(u,v)$$
  show that $(u,v)=(u,atv)\sim(ua,tv)$.
  It follows from
  Proposition~\ref{p:covers-in-F}\ref{item:covers-in-F:1} that $(u,v)$ has no successor,
  since otherwise this successor would be $(ua,tv)$, in
    contradiction with $(ua,tv)\leq(u,v)$. Since it is not the maximum
  of~$\mathfrak{F}(s)$ (because $v=zv\neq 1$), it must be a stationary point.
\end{proof}

Proposition~\ref{p:shifting-allowed-implies-type-2}
is applied in the proof of the next result.

\begin{Prop}
  \label{p:transitions-in-F}
   Let $S$ be an equidivisible profinite semigroup $S$
     which is finitely cancelable.
     Then, in $\Cl T(S)$, any two coterminal edges between
       elements of distinct strongly connected components  are
       equal. In other words, if
     $(\alpha,\beta)\xrightarrow t(\gamma,\delta)$
     and $(\alpha,\beta)\xrightarrow s(\gamma,\delta)$
     are edges of
     $\Cl T(S)$ such that $(\alpha,\beta)<(\gamma,\delta)$, then $t=s$.
 \end{Prop}

\begin{proof}
  The hypothesis translates into
  the following equalities: $\gamma=\alpha t=\alpha s$ and
  $\beta=t\delta=s\delta$. We first assume that at least one of the
  points $(\alpha,\beta)$ and $(\gamma,\delta)$ is a step point. By
  symmetry, we may as well assume that $(\alpha,\beta)$ is a step
  point. From the equality $\alpha t=\alpha s$, by equidivisibility
  and without loss of generality, we may assume that there is some
  $z\in S^I$ such that $\alpha z=\alpha$ and $t=zs$.
  Then, we have $z\beta=zs\delta=t\delta=\beta$.
  This shows that
  $(\alpha,\beta)\xrightarrow z(\alpha,\beta)$
  is a loop of $\Cl T(w)$. Since $(\alpha,\beta)$
  is assumed to be a step point,
  applying Proposition~\ref{p:shifting-allowed-implies-type-2}
  we obtain $z=1$, thus $t=s$.

  It remains to consider the case where both $(\alpha,\beta)$ and
  $(\gamma,\delta)$ are stationary. By
  Theorem~\ref{t:cluster}, there is a step point
  $(x,y)\in\mathfrak{F}(w)$ such that
  $(\alpha,\beta)<(x,y)<(\gamma,\delta)$. By the preceding paragraph,
  there are unique edges in~$\Cl T(w)$ from $(\alpha,\beta)$ to
  $(x,y)$ and from $(x,y)$ to $(\gamma,\delta)$. Let $r_1$ and $r_2$ be
  the respective labels. To prove the proposition, it suffices to show that
  $\tau=r_1r_2$ whenever $(\alpha,\beta)\xrightarrow \tau(\gamma,\delta)$ is an
  edge of~$\Cl T(w)$.

  Note that
  $(x,r_2)$ and $(\alpha,\tau)$ are elements of $\mathfrak{F}(\gamma)$.
  Since $\beta=\tau\delta$ and $y=r_2\delta$, if $(x,r_2)\leq (\alpha,\tau)$ (in $\mathfrak{F}(\gamma)$),
  then $(x,y)\leq(\alpha,\beta)$ (in $\mathfrak{F}(w)$) by
  Remark~\ref{r:trivial-preservation}, which gives a contradiction.
  Hence, by equidivisibility, there is in~$\Cl T(\gamma)$ an edge $(\alpha,\tau)\xrightarrow
  {r_0}(x,r_2)$.
    In particular $\tau=r_0r_2$.
 Remark~\ref{r:trivial-preservation} guarantees the
 existence of the edge
 $(\alpha,\beta)\xrightarrow {r_0}(x,y)$. But we defined, in the previous paragraph, the edge $(\alpha,\beta)\xrightarrow
 {r_1}(x,y)$ as the unique edge from $(\alpha,\beta)$ to $(x,y)$.
 Therefore $r_0=r_1$, and so we have the equality $\tau=r_1r_2$, which we have seen
 to be sufficient to conclude the proof.
\end{proof}

\begin{Cor}
  \label{c:transitions-in-F}
        Let $S$ be an equidivisible profinite semigroup $S$
     which is finitely cancelable.
     Let $p_1,p_2\in\mathfrak{L}(w)$.
     If $p_1<p_2$ then the set of transitions from $p_1$ to $p_2$ is
     contained in a $\Cl J$-class of $S$.
\end{Cor}

\begin{proof}
      Let $x_1,x_1'\in p_1$ and $x_2,x_2'\in p_2$,
      and consider edges $x_1\xrightarrow{t}x_2$
      and $x_1'\xrightarrow{t'}x_2'$.
      Then we have edges $x_i\xrightarrow{s_i}x_i'$
      and $x_i'\xrightarrow{r_i}x_i$,
      and also
      $x_1\xrightarrow{s_1t'r_2}x_2$
      and
      $x_1'\xrightarrow{r_1ts_2}x_2'$.
      By Proposition~\ref{p:transitions-in-F},
      we must have $t=s_1t'r_2$
      and $t'=r_1ts_2$, whence $t\mathrel{\Cl J}t'$.
\end{proof}

\begin{Remark}\label{r:cuc-reunite-condi}
  We remark that there is a large class of pseudovarieties
  whose corresponding finitely generated relatively free profinite semigroups satisfy the hypotheses of Proposition~\ref{p:transitions-in-F}.
  Let $\pv V$ be a pseudovariety
  of semigroups such that $\pv V=\pv {LI}\malcev\pv V$.
  Then every semigroup of the form $\Om AV$,
  with $A$ a finite alphabet,
  satisfies all conditions in Proposition~\ref{p:transitions-in-F}:
  they are profinite, equidivisible,
  and finitely cancelable
  (cf.~Theorem~\ref{t:characterization-of-equidivisible-pseudovarieties}
  and Proposition~\ref{p:equid-conca-are-finitely-cancelable}, where the latter
  may be applied because
  $\pv V=\pv {LI}\malcev\pv V$ implies $\pv V\supseteq \pv {LI}$
  and thus $\pv V\not\subseteq\pv{CS}$).
\end{Remark}

Proposition~\ref{p:transitions-in-F} is used in the proof of the
following lemma, establishing a sufficient condition for equality
between stationary points.

\begin{Lemma}
  \label{l:absorption}
  Let $S$ be an equidivisible profinite semigroup which is finitely
  cancelable.
  Let $p$ and $q$ be stationary points of~$\mathfrak{L}(S)$.
  If there is a transition $p\xrightarrow tq$ such that $t$
  lies \Cl J-above both $J_p$ and $J_q$ then $p=q$.
\end{Lemma}

\begin{proof}
  We have $p\leq q$, so, arguing by contradiction, suppose that $p<q$.
  Consider an edge
  $(u,v)\xrightarrow t(x,y)$
  of $\Cl T(S)$ with $(u,v)\in p$ and $(x,y)\in q$. Let $e\in J_p$ and
  $f\in J_q$ be
  idempotents respectively
  stabilizing $(u,v)$ and $(x,y)$. Then we have a transition
  $(u,v)\xrightarrow{etf}(x,y)$.
  By
  Proposition~\ref{p:transitions-in-F},
  it follows that $etf=t$,
  whence, since $S$ is stable,
  and by the hypothesis that $t$ lies \Cl J-above both $e$ and~$f$,
  we have $e\mathrel{\Cl R}t\mathrel{\Cl L}f$.
  Therefore, there exists $s\in J_p=J_q$ such that $ts=e$ and $st=f$.
  Now, the assumption that we have a transition
  $(u,v)\xrightarrow t(x,y)$
  means that the equalities $ut=x$ and $v=ty$ hold. Combining
  with the equalities $ue=u$ and $fy=y$, we deduce that $xs=uts=ue=u$
  and $sv=sty=fy=y$. Hence $(x,y)\leq(u,v)$, which
  contradicts the assumption $p<q$. To avoid the contradiction, we
  must have $p=q$.
\end{proof}

The following is an example obtained by application of Lemma~\ref{l:absorption}.

\begin{Example}\label{ex:minimal-shifts}
   Let $\pv V$ be a pseudovariety
   containing $\pv {LSl}$, and let $A$ be a finite alphabet.
   It is shown in~\cite{Almeida:2005c} (see
   also~\cite{Almeida&ACosta:2007a}),
   using Zorn's Lemma and a standard compactness
   argument, that $\Om AV$ contains regular elements that
   are $\Cl J$-maximal among the elements of $\Om AV\setminus A^+$.
   We next verify that if $\pv V$ is closed under concatenation
   and $u$ is a $\Cl J$-maximal regular element of
   $\Om AV$, then the order type of $\mathfrak {L}(u)$ is $\omega+1+\omega^*$.
   Indeed, let $p,q$ be two stationary points
  of $\mathfrak {L}(u)$ such that $p\leq q$,
  and let $e\in J_p$, $f\in J_q$.
  From the maximality assumption on $u$,
  we deduce that $e\mathrel{\Cl J}u\mathrel{\Cl J}f$.
  Hence $p=q$ by Lemma~\ref{l:absorption}, showing that
  $\mathfrak {L}(u)$ has only one stationary point.
\end{Example}

Restricting our attention to the case of profinite semigroups
that are free relatively to pseudovarieties closed under concatenation,
we obtain stronger results about step and stationary
points. The next lemma is a provisional result with this flavor, which
is improved later on, namely in
Proposition~\ref{p:approximation-type-2-by-step-points}.

\begin{Lemma}
  \label{l:approximation-type-2}
  Let $A$ be a finite alphabet, and let
  $\pv V$ be a pseudovariety closed under concatenation.
  Let $w\in \Om AV$.
  If $(u,v)$ is a stationary point
  of $\mathfrak{F}(w)$, then
  there exists a strictly increasing sequence $(u_n,v_n)_n$
  of points of\/~$\mathfrak{F}(w)$ such that $\lim u_n=u$ and $\lim v_n=v$.
\end{Lemma}

\begin{proof}
  Since $A$ is finite,
  $\Om AV$ is metrizable,
  and so $\mathfrak{L}(w)$
  is metrizable by Corollary~\ref{c:Ls-is-compact}.
    Let $p=(u,v)/{\sim}$.
  By  Theorem~\ref{t:cluster}, there is a strictly
  increasing sequence $(x_n,y_n)_n$ of step points
  converging in
  $\mathfrak{L}(w)$ to $p$.
  As $S$ is compact, taking subsequences
  we may assume that
  $(x_n,y_n)_n$ converges
  in $\mathfrak{F}(w)$ to some $(x,y)\in\mathfrak{F}(w)$. By
  Proposition~\ref{p:topology-of-Fw}, we know that $(x,y)\in p$.
  Then there are edges $(x,y)\xrightarrow s(u,v)$ and $(u,v)\xrightarrow t(x,y)$
  in~$\Cl T(w)$. Note that $\lim x_n=x=xst$. Hence, by
  Theorem~\ref{t:factorizing} and
  Lemma~\ref{l:a-condition-on-metric-semigroups},
  for every~$n$ there is a factorization $x_n=x_n's_nt_n$
  with $\lim x_n'=x$, $\lim s_n=s$, and $\lim t_n=t$. Then, we have
  $\lim(x_n's_n,t_ny_n)=(u,v)$ in $\mathfrak {F}(w)$,
  thus $\lim(x_n's_n,t_ny_n)/{\sim}=p$ in $\mathfrak {L}(w)$ by
  Proposition~\ref{p:topology-of-Fw}.
  And since $(x_n's_n,t_ny_n)\leq(x_n,y_n)<(x,y)$,
  we conclude that the sequence  $((x_n's_n,t_ny_n)/{\sim})_n$ has $p$
  as a supremum but not as a maximum, enabling us to extract from the
  sequence $(x_n's_n,t_ny_n)_n$ a strictly increasing subsequence.
  Since this subsequence also
  converges to~$(u,v)$ in $\mathfrak{F}(w)$, this completes the proof.
\end{proof}

\begin{Lemma}
  \label{l:clopen-in-F}
    Let $A$ be a finite alphabet, and let
    $\pv V$ be a pseudovariety closed under concatenation.
  Let $w\in\Om AV$ and let $K$ be a nonempty clopen subset
  of~$\mathfrak{F}(w)$. Then $K$ has a minimum and a maximum, and both
  of them are step points.
\end{Lemma}

\begin{proof}
  As $K$ is closed, whence compact, we know from
  Proposition~\ref{p:topology-of-Fw} that its projection under
  $\chi$ is also compact,
  whence closed. Let $p=\inf{\chi(K)}$ and $q=\sup\chi(K)$. Since
  $\chi(K)$ is closed, we have $p=\min\chi(K)$ and $q=\max\chi(K)$.
  Hence, there exist $\hat p\in K\cap\chi^{-1}(p)$ and $\hat q\in
  K\cap\chi^{-1}(q)$. Although these points are perhaps not unique,
  all choices are $\sim$-equivalent. Moreover, for every other point
  $r\in K$, from $\chi(r)\in[p,q]$ we deduce that $\hat p\leq r\leq
  \hat q$.

  It remains to show that $\hat p$ and $\hat q$ are step points. Suppose on
  the contrary that $\hat p$ is a stationary point. By
  Lemma~\ref{l:approximation-type-2}, there is a strictly
  increasing sequence $(\hat p_n)_n$ in~$\mathfrak{F}(w)$ converging
  to~$\hat p$. Since $K$~is open, it must contain points of the form
  $\hat p_n$. As $\hat p_n<\hat p$, this contradicts the fact that
  $\hat p$ is
  a minimum of~$K$. Hence, $\hat p$~is a step point and, similarly, so
  is~$\hat q$.
\end{proof}

Lemma~\ref{l:approximation-type-2} was used to prove
Lemma~\ref{l:clopen-in-F}.
We next use
Lemma~\ref{l:clopen-in-F}
to show that the sequence in Lemma~\ref{l:approximation-type-2}
may be formed only by step points.

\begin{Prop}
  \label{p:approximation-type-2-by-step-points}
    Let $A$ be a finite alphabet, and let
  $\pv V$ be a pseudovariety closed under concatenation.
  Let $w\in \Om AV$.
  If $(u,v)$ is a stationary point
  of $\mathfrak{F}(w)$, then
  there exists a strictly increasing sequence $(u_n,v_n)_n$
  of step points of\/~$\mathfrak{F}(w)$ such that $\lim u_n=u$ and
  $\lim v_n=v$.
\end{Prop}

\begin{proof}
  Since
  $\mathfrak{F}(w)$ is a zero-dimensional space
  and the set of step points of~$\mathfrak{F}(w)$ is linearly
  ordered by~$\leq$, it suffices to show that, for every clopen subset
  $K\subseteq\mathfrak{F}(w)$ containing $p=(u,v)/{\sim}$ and every interval
  $[q,p]$, with $q<p$, there is some step point in the intersection
  $K\cap[q,p]$.
  Because $\step (\mathfrak {L}(w))$ is
  topologically dense in $\mathfrak {L}(w)$
  (cf.~Theorem~\ref{t:cluster}),
  we are reduced to the case where $q$ is a step point.
  Let $r$ be a step point such that $r>p$, for instance
  $r=(w,1)$.
  The closed intervals in $\mathfrak {L}(w)$ whose extremities
  are step points are clopen, hence, by Proposition~\ref{p:topology-of-Fw},
  the interval $[q,r]$ in~$\mathfrak{F}(w)$ is clopen.
  Therefore, the subset $L=K\cap[q,r]$ is also clopen.
    Note that $L$ is nonempty, as $p\in L$.
  Lemma~\ref{l:clopen-in-F} guarantees that $L$
  has a minimum $s$ which is a step point.
  It then follows from $p\in L$ that $s\in K\cap[q,p]$,
  thereby completing the proof of the lemma.
\end{proof}

\section{Cluster words}
\label{sec:cluster-words}

In this section, we use
  the knowledge obtained about the labeled linear orders induced by
  pseudowords over \pv A to obtain a representation theorem: every
  pseudoword over \pv A may be represented by a (partially) labeled
  linear order having specific properties, which we introduce now.

By a \emph{partially labeled ordered set}, we mean a pair $(P,f)$ such
that $P$ is an ordered set and $f$ is a function (the \emph{labeling})
with domain contained in~$P$. An isomorphism between partially labeled
ordered sets $(P,f)$ and $(Q,g)$ is an isomorphism $\varphi\colon P\to
Q$ of ordered sets such that $\Dom g=\varphi(\Dom f)$ and
$f(p)=g\circ\varphi(p)$ for every $p\in \Dom f$.

Consider a profinite equidivisible semigroup $S$ which is finitely
cancelable with respect to $A$.
We may then consider the mappings $\beu$ and $\teu$ from $S^I$ to
$A\uplus \{1\}$
such that $\beu (1)=\teu (1)=1$ and, for $s\in S$, the images $\beu (s)$
and $\teu (s)$ are respectively the unique prefix and the unique
suffix of $s$ in $A$.
For $s\in S$, denote by~$\mathfrak {L}_c(s)$ the
partially labeled linearly ordered set $(\mathfrak {L}(s),\ell)$ defined by the mapping
$\ell\colon \step (L)\to A\uplus \{1\}$ such that $\ell (u,v)=\beu(v)$,
for $(u,v)\in\step (\mathfrak {L}(s))$.
Note that $\ell (p)=1$ if and only if $p=(s,1)$.
Recall from Theorem~\ref{t:cluster} that $\mathfrak {L}(s)$
is clustered. By a \emph{cluster word over $A$}
we mean a partially labeled linearly ordered set $(L,\ell)$ such that $L$ is
clustered and $\ell$ is a function $\step (L)\to A\uplus \{1\}$.

\begin{Example}\label{ex:examples-of-cluster-words}
   For the pseudoword $w=(a^\omega b)^\omega$ of $\Om {\{a,b\}}A$, the
    cluster word
  $\mathfrak{L}_c(w)$ is
  described in Diagram~\eqref{eq:a-cluster-word},
    \begin{align}\label{eq:a-cluster-word}
      \begin{split}
      &\underbrace{ \underbrace{aa\cdots}_\omega \;\bullet\;
        \underbrace{\cdots aa}_{\omega^*}b
        \underbrace{aa\cdots}_\omega \;\bullet\; \underbrace{\cdots
          aa}_{\omega^*}b\underbrace{aa\cdots}_\omega \;\bullet\;
        \underbrace{\cdots aa}_{\omega^*}b
        \cdots}_\omega\\
      &\qquad{\mathlarger{\mathlarger{\mathlarger{\bullet}}}}
      \underbrace{ \cdots \underbrace{aa\cdots}_\omega \;\bullet\;
        \underbrace{\cdots aa}_{\omega^*}b
        \underbrace{aa\cdots}_\omega \;\bullet\; \underbrace{\cdots
          aa}_{\omega^*}b \underbrace{aa\cdots}_\omega \;\bullet\;
        \underbrace{\cdots aa}_{\omega^*}b }_{\omega^*}
      \end{split}      
    \end{align}
  where each of the small bullets
  $\bullet$ represent a stationary point $q$ such that
  $a^\omega\in J_q$ and
  the bigger bullet
  ${\displaystyle\mathlarger{\mathlarger{\mathlarger{\bullet}}}}$
  represents the unique stationary point $r$ such that
  $(a^\omega b a^\omega)^\omega\in J_r$.
  Note that, in particular, the
  order type of $\mathfrak{L}(w)$ is
  $(\omega+1+\omega^*)\omega+1+(\omega+1+\omega^*)\omega^*$.
\end{Example}

We next introduce, in Definition~\ref{def:successful-run}, a notion of
algebraic recognition of cluster words inspired by the definition of
automata recognizing words indexed by linear orders, introduced
in~\cite{Bruyere&Carton:2007}. A notorious similarity resides in the
role played by \emph{cofinal} sets, whose definition we next recall.
In a linearly ordered set $L$, a subset $X$ of $L$ is \emph{left
  cofinal at $p$} if $X\cap {]}q,p{[}\neq \emptyset$ for every $q<p$,
\emph{right cofinal at $p$} if $X\cap {]}p,q{[}\neq \emptyset$ for
every~$q>p$, and \emph{cofinal} if it is right or left cofinal at $p$.

\begin{Def}
    \label{def:successful-run}
    Let $\varphi\colon A\to S$ be a generating mapping of a
    semigroup~$S$,  and let $s\in S$.
    We say that the cluster word $(L,\ell)$ over $A$
    is \emph{recognized} by the pair $(\varphi,s)$
    if there is a mapping $g\colon \step(L)\to \mathfrak {F}(s)$
    satisfying:
\begin{enumerate}
  [label=(R.\arabic*)]
\item $g(\min L)=(1,s)$;\label{item:run-1}
\item $g(\max L)=(s,1)$;\label{item:run-2}
\item if $p_1\prec p_2$ in $L$, then
  $g(p_1)\xrightarrow{\varphi(\ell(p_1))}g(p_2)$
  is an edge of $\Cl T(s)$;\label{item:run-3}
\item if $p$ is a stationary point of $L$\label{item:run-4}
  then, for every $q\in \mathfrak {F}(s)$,
  the set $g^{-1}(q)$ is left cofinal at $p$ if and only if
  it is right cofinal at $p$.
\end{enumerate}
If such conditions are satisfied,
then we say that
$(L,\ell)$ is \emph{$g$-recognized} by the pair $(\varphi,s)$.
We also say that
$(L,\ell)$ is \emph{recognized} by the pair $(\varphi,s)$
if it is $g$-recognized, for some $g$.
Finally, when $(L,\ell)$ is $g$-recognized by $(\varphi,s)$, we define
$F_g$ to be the function
from $\stat(L)$ to the power set $\mathcal P(\mathfrak {F}(s))$
  such that
  \begin{align*}
    F_g(p)&=\{q\in\mathfrak {F}(s):\text{$g^{-1}(q)$ is left cofinal at
      $p$}\}\label{eq:Fp}\\
    &=\{q\in\mathfrak {F}(s):\text{$g^{-1}(q)$ is right cofinal at
      $p$}\},
  \end{align*}
  for every stationary point $p$ of $L$ (cf.~\ref{item:run-4}).
\end{Def}

\begin{Remark}
  \label{r:running-is-isomorphic-invariant}
If $f\colon (L',\ell')\to (L,\ell)$ is an isomorphism
of cluster words over~$A$ and if $(L,\ell)$ is $g$-recognized by $(\varphi,s)$,
then $(L',\ell')$ is $(g\circ f)$-recognized by $(\varphi,s)$.
Hence, the property of
$(L,\ell)$ being recognized by $(\varphi,s)$ is invariant under
isomorphism of cluster words.
\end{Remark}

The next lemma, concerning the function $F_g$, will be
applied in Section~\ref{sec:proof-theorem-reft:a}.

\begin{Lemma}
  \label{l:if-P-is-maximal-then-inverse-image-by-g-is-closed}
 Consider a cluster word $(L,\ell)$ over $A$,
 recognized by the pair $(\varphi,s)$, where
 $\varphi\colon A\to S$ is a generating mapping of a finite semigroup~$S$
 and $s\in S$.
 Suppose that the order topology of $L$ is metrizable.
 Let $p$ be a stationary point of $L$.
 Then the following properties hold:
 \begin{enumerate}
 \item if $P$ is a subset of $\mathfrak {F}(s)$
  such that $F_g^{-1}(P)$ is cofinal at $p$, then $P\subseteq
  F_g(p)$;\label{item:if-P-is-maximal-1}
\item   there is an open interval $I$ containing $p$ such that
  $F_g(q)\subseteq F_g(p)$ for every stationary point $q$
  in $I$.\label{item:if-P-is-maximal-2}
 \end{enumerate}
\end{Lemma}

\begin{proof}
  \ref{item:if-P-is-maximal-1}
  Without loss of generality, suppose that
  $F_g^{-1}(P)$ is left cofinal at $p$.
  Then, taking into account that $L$ is metrizable,
  there is a strictly increasing sequence $(p_n)_n$ of
  stationary points of $L$ converging to~$p$
  such that $F_g(p_n)=P$ for all~$n$.
  Consider a metric $d$ on $L$
  inducing the order topology of $L$. 
  For each $k\geq 1$, let $]q_k,p[$ be an open interval of
  $L$ such that $d(q_k,p)<\frac{1}{k}$
  Then, there is 
  $n_k$ such that $p_{n_k}\in {]}q_k,p{[}$.
  Let $(s_1,s_2)\in P$.
  By the definition of $F_g$, we know that 
  $g^{-1}(s_1,s_2)$ is left cofinal at $p_{n_k}$. Therefore, there is
  in the interval $]q_k,p_{n_k}[$ a step point belonging to $g^{-1}(s_1,s_2)$. Since
  ${]}q_k,p_{n_k}{[}\subseteq {]}q_k,p{[}$, this proves that
  $g^{-1}(s_1,s_2)$ is left cofinal at $p$, thus
  $(s_1,s_2)\in F_g(p)$. We have therefore proved that the inclusion
  $P\subseteq F_g(p)$ holds.

  \ref{item:if-P-is-maximal-2} Let $\Cl P$ be the set of subsets $P$
  of $\mathfrak {F}(s)$ such that $F_g^{-1}(P)$ is not cofinal at
  $p$.  For each $P\in \Cl P$, there is an open interval $I_p$
  containing $p$ such that $F_g^{-1}(P)\cap I_p=\emptyset$.
  Because $S$ is a finite semigroup,
  the set $\Cl P$~is finite, and so $I=\bigcap_{P\in \Cl P}I_p$ is an open interval of~$L$ containing $p$.  Let $q$ be a stationary point in $I$.  Let
  $R=F_g(q)$.  Suppose $R\in\Cl P$. Then $q\in I_R$. But we also have
  $q\in F_g^{-1}(R)$, which leads to
  $F_g^{-1}(R)\cap I_R\neq\emptyset$, a contradiction with the
  definition of $I_R$.  This shows that $R\notin\Cl P$.  We then
  deduce from part~\ref{item:if-P-is-maximal-1} of the lemma that
  $R\subseteq F_g(p)$.
\end{proof}

The main results in this section
are about cluster words defined by elements of $\Om AA$,
but in the next proposition we embrace
without additional effort all pseudovarieties closed under concatenation.

\begin{Prop}
  \label{p:direct-implication-of-theorem-we-search}
  Consider a pseudovariety $\pv V$ closed under
  concatenation.
Given $w\in\Om AV$, and a generating mapping
$\varphi\colon A\to S$ of a semigroup $S$ of~$\pv V$,
let $s=\varphi_\pv V(w)$.
Consider the mapping
$g_{w,\varphi}\colon \step(\mathfrak{L}(w))\to \mathfrak {F}(s)$
such that
$g_{w,\varphi}(u,v)=(\varphi_\pv V(u),\varphi_\pv V(v))$
for every 
step point $(u,v)$ of $\mathfrak{L}(w)$.
  Then the cluster word $\mathfrak{L}_c(w)$
  is $g_{w,\varphi}$-recognized by $(\varphi,s)$.
\end{Prop}

\begin{proof}
  The conditions \ref{item:run-1}-\ref{item:run-2}
  in Definition~\ref{def:successful-run}
  for $g_{w,\varphi}$-recognition by  $(\varphi,s)$ are
  clearly satisfied, and \ref{item:run-3}
  follows directly from
  Proposition~\ref{p:covers-in-F}\ref{item:covers-in-F:1}.

  Let us verify condition~\ref{item:run-4}.
    Let $(x,y)\in \mathfrak {F}(s)$ be such that
    $g_{w,\varphi}^{-1}(x,y)$ is left
  cofinal at $p$. Then, there is a strictly increasing sequence
  $(u_n,v_n)_{n\geq 1}$ of step points
  belonging to $g_{w,\varphi}^{-1}(x,y)$
  converging in $\mathfrak {F}(w)$ to an element $(u,v)$ of~$p$.
    Since $(\varphi_\pv V(u_n),\varphi_\pv V(v_n))=(x,y)$ for every $n\geq 1$,
  we also have $(\varphi_\pv V(u),\varphi_\pv V(v))=(x,y)$,
  by continuity of $\varphi_\pv V$.
  By the dual of Proposition~\ref{p:approximation-type-2-by-step-points},
  there is a strictly decreasing sequence
  $(u'_n,v'_n)_{n\geq 1}$ of step points converging to $(u,v)$ in
  $\mathfrak {F}(w)$.
  In particular, by continuity of $\varphi_\pv V$,
  for all sufficiently large~$n$, we have
  $(\varphi_\pv V(u_n'), \varphi_\pv V(v_n')) = (x,y)$,
  and so $g_{w,\varphi}^{-1}(x,y)$ is right cofinal at $p$.
  Dually, if
  $g_{w,\varphi}^{-1}(x,y)$ is right cofinal at $p$ then
  $g_{w,\varphi}^{-1}(x,y)$ is left cofinal at $p$. This establishes 
  condition~\ref{item:run-4}.
\end{proof}

In the case of unambiguous aperiodic semigroups,
we have a converse for
Proposition~\ref{p:direct-implication-of-theorem-we-search}, as stated
in the next theorem.

\begin{Thm}\label{t:a-description-of-the-representation-theorem}
    Let $w\in\Om AA$,
    and consider
    a generating mapping  $\varphi\colon A\to S$ of a finite aperiodic
  unambiguous semigroup $S$.
  Then $\varphi_{\pv A}(w)=s$ if and only if
  $\mathfrak{L}_c(w)$ is recognized by $(\varphi,s)$.
\end{Thm}

We defer the proof of
Theorem~\ref{t:a-description-of-the-representation-theorem}
to Section~\ref{sec:proof-theorem-reft:a} (but note that
the direct implication in
  Theorem~\ref{t:a-description-of-the-representation-theorem}
  is an immediate application of
  Proposition~\ref{p:direct-implication-of-theorem-we-search}).
  Meanwhile, we use it to prove the following main result.

\begin{Thm}\label{t:equal-cluster-words}
  Let $u,v\in\Om AA$. Then
  $u=v$ if and only if the cluster words $\mathfrak {L}_c(u)$
  and $\mathfrak {L}_c(v)$ are isomorphic.
\end{Thm}

\begin{proof}
  The isomorphism of cluster words is clearly necessary to have
  $u=v$. Conversely, suppose $\mathfrak {L}_c(u)$ and
  $\mathfrak {L}_c(v)$ are isomorphic.
  Let $\varphi\colon A\to S$ be a generating mapping of
  a finite unambiguous aperiodic semigroup~$S$. Take $s=\varphi_{\pv  A}(u)$.
  Then $\mathfrak {L}_c(u)$ is recognized by
  $(\varphi,s)$, according to the direct implication in
  Theorem~\ref{t:a-description-of-the-representation-theorem}.
  But then
  $\mathfrak {L}_c(v)$ is also recognized by
  $(\varphi,s)$~(cf.~Remark~\ref{r:running-is-isomorphic-invariant}).
  Hence, we have $s=\varphi_{\pv  A}(v)$ by the converse implication
  in Theorem~\ref{t:a-description-of-the-representation-theorem}.
  By Proposition~\ref{p:residual-B-semigroups}, this establishes $u=v$.
\end{proof}

\section{Stabilizers}

Given a semigroup $S$ and $s\in S$, we say that an element $x$
of~$S^I$ \emph{stabilizes $s$ on the right} if $sx=s$; the set
$\rstab(s)$ of all such $x$ constitutes a submonoid of $S^I$ and is
called the \emph{right stabilizer} of~$s$. One defines dually the
elements that \emph{stabilize $s$ on the left}, which form a submonoid
$\lstab(s)$ of $S^I$, called the \emph{left stabilizer} of~$s$.

An application of the following result will be required in the sequel. 

\begin{Thm}
  \label{t:cancellation}
  Let $S$ be an equidivisible profinite semigroup which
  is finitely cancelable.
  Let $u\in S$.
  Let $g$ be an element of the monoid $\rstab(u)$
  or of the monoid  $\lstab(u)$.
  If $g$ is regular within that monoid, then $g=g^2$.
\end{Thm}

In Theorem~\ref{t:cancellation}, the hypothesis that $S$ is finitely
cancelable is not superfluous: if $G$ is a group, then the semigroup
$G^0$ obtained from $G$ by adjoining a zero is equidivisible, and
$G^0$ is the left and right stabilizer of zero.

We should mention that we do not know of any other examples of
equidivisible profinite semigroups that are finitely cancelable other
than free pro-\pv V semigroups, where \pv V is a pseudovariety with
suitable closure properties. For such semigroups,
Theorem~\ref{t:cancellation} follows from more general results
in~\cite[Theorem~13.1]{Rhodes&Steinberg:2002}.
Nevertheless, since our results apply to all equidivisible profinite
semigroups that are finitely cancelable, we present a proof of
Theorem~\ref{t:cancellation} which may be of independent interest. 

As a first step we have the following simple
statement.

\begin{Lemma}
  \label{l:cancellation-in-R-class}
  Let $x$ and $y$ be $\Cl D$-equivalent
  elements of a stable semigroup~$S$.
  If $yx=x$ then $y=y^2$.
\end{Lemma}

\begin{proof}
  Since $yx=x$ and $y$ are in the same $\Cl D$-class and $S$~is
  stable, we must have $y\mathrel{\Cl R}x$. Therefore, $y=xu$ for
  some $u\in S^{I}$, so $y^{2}=yxu=xu=y$.
\end{proof}

We proceed with an auxiliary lemma. 

\begin{Lemma}
  \label{l:type1-vs-left-group-action}
  Let $S$ be an equidivisible profinite semigroup
  which is finitely cancelable.
  Let $g,w\in S$ be such that $gw=w$. If $(x,y)\in\mathfrak{F}(w)$ is
  a step point  satisfying $(g^\omega,w)\leq(x,y)$, then $gx=x$.
\end{Lemma}

\begin{proof}
  Note that $(g^\omega,w)\leq(x,y)$ implies
  $x\leq_{\Cl R} g^\omega$, thus $x=g^\omega x$, a fact that we shall
  use along the proof.

  By equidivisibility,
  $(x,y)$ and $(gx,y)$ are comparable in $\mathfrak{F}(w)$.
  Suppose first that $(gx,y)\leq(x,y)$. Then there is
  $t\in S^I$ such that $gxt=x$ and $ty=y$.
  It suffices to show that $t=1$. The equality $gxt=x$
  entails $x=g^nxt^n$ for every $n\geq1$, thus $x=g^\omega
  xt^\omega=xt^\omega$.
  And since $ty=y$,
  we conclude that $t^\omega$ stabilizes $(x,y)$
  in $\Cl T(S)$.
  Because $(x,y)$
  is a step point,  Proposition~\ref{p:shifting-allowed-implies-type-2}
  implies that $t=1$.

  Suppose next that $(x,y)\leq(gx,y)$. Then
  $(g^{\omega-1}x,y)\leq(g^\omega x,y)=(x,y)$
  (cf.~Remark~ \ref{r:trivial-preservation})
  and, by the preceding case, we deduce that $g^{\omega-1}x=x$.
  With a left multiplication  by~$g$ on both sides of the latter equality, we
  obtain $x=g^\omega x=gx$, as desired.
\end{proof}

\begin{proof}[Proof of Theorem~\ref{t:cancellation}]
  It suffices to consider the case where $g$ is an element of
  $\lstab(u)$, as the other case is dual.

  We first establish the theorem when $g$ is a group element, that is, $g=g^{\omega+1}$.
  
  Let $R$ be the set of elements $(\alpha,\beta)$
  of $\mathfrak {F}(u)$ such that
  $\alpha\mathrel{\Cl R}g$.
  Note that $R$ is nonempty, indeed $(g,u)\in R$.   
  The set $R$ is closed, whence compact, and so by continuity of
  $\chi$,
  the image $\chi(R)$ is also compact, whence closed.
  Therefore, by completeness of $\mathfrak {L}(u)$
  (cf.~Proposition~\ref{p:complete-quasi-ordering}),
  the closed set $\chi(R)$ has a maximum $p=(x,y)/{\sim}$.
  Let us observe that, since
  two $\sim$-equivalent elements
  must have $\Cl R$-equivalent first components,
  the inclusion $\chi^{-1}(p)\subseteq R$ holds. 
  Moreover, since $g=g^{\omega+1}$,
  we have $(g^{\omega},u)\in R$, thus $(g^\omega,u)\leq (x,y)$
  by definition of $p$.
  
  Suppose first that $p$ is a step point.  
  Then $gx=x$
  by Lemma~\ref{l:type1-vs-left-group-action}. As $x\mathrel {\Cl R}g$,
  we conclude that $g=g^2$ by Lemma~\ref{l:cancellation-in-R-class}.

  If $p$ is stationary, then,
  by Theorem~\ref{t:cluster},
  there is a
  net $(x_i,y_i)_{i\in I}$
  of step points converging in $\mathfrak{L}(u)$ to~$p$
  and such that $p<(x_i,y_i)$ for all $i\in I$.
  By compactness, taking a subnet, we may assume that
  $(x_i,y_i)_{i\in I}$
  converges in $\mathfrak{F}(u)$ to
  some element $(x',y')$ of $p$.
    As $(g^\omega,u)< (x_i,y_i)$,
    we deduce from Lemma~\ref{l:type1-vs-left-group-action}
    that $gx_i=x_i$ for every~$i\in I$.
   Taking limits, it follows that $gx'=x'$.
   Since  $(x',y')\in p$, we have  $x\mathrel{\Cl R}x'$, whence
   $g\mathrel{\Cl R}x'$, and we again deduce that $g=g^2$ by
   Lemma~\ref{l:cancellation-in-R-class}.

   We have thus concluded the proof for the case where $g$ is a group
   element of $S$. Let us now suppose
   that  $g$ is regular within $\lstab(u)$. Then there is
    $h\in \lstab(u)$
   such that $g=ghg$ and $h=hgh$. Since $hu=gu$,
   by equidivisibility we know that $h$ and $g$
   are $\mathcal R$-comparable in~$S$. We actually have
   $h\mathrel{\Cl R}g$ in~$S$, by stability of~$S$.
   Let $z\in S^I$ be such that $h=gz$. Then $g=g^2zg$,
   and therefore $g\mathrel {\Cl R}g^2$.
   This shows that $g$ is a group element of $S$, and so,
   as we are in the case already proved, we get $g=g^2$.   
\end{proof}

A semigroupoid $S$ is \emph{trivial}
if, for any two vertices $p,q\in S$, there is at most one edge $p\to q$.

\begin{Cor}
  \label{c:trivial-groupoid-at-type-2-point}
      Let $S$ be an equidivisible profinite semigroup
      which is finitely cancelable.
      For every $p\in\mathfrak {L}(S)$,
      the ideal $\Cl K_p$ is a trivial category.
\end{Cor}

\begin{proof}
  Let $(u,v)\xrightarrow s(x,y)$ and
  $(u,v)\xrightarrow t(x,y)$ be edges of
  $\mathcal K_p$.
  The proof amounts to showing that $s=t$.
  There is an edge $(x,y)\xrightarrow z(u,v)$ in $\mathcal K_p$.
  In particular,
  $sz$ and $(sz)^2$
  label loops in $\mathcal K_p$ at vertex $(u,v)$.
  Hence, $sz$ and $(sz)^2$ belong to~$J_p$,
  and so $sz$ is a group element of $S$ stabilizing $u$ on the right.
  We then deduce from Theorem~\ref{t:cancellation} that $sz$
  is an idempotent of $J_p$, denoted by $e$, which
  stabilizes $(u,v)$.
  Similarly, $tz$ is an idempotent of $J_p$
  stabilizing $(u,v)$. It follows
  from Lemma~\ref{l:minimum-idempotents-labeling-loops-at-same-point}
  that $sz=tz=e$.
  Symmetrically, we have $zs=zt=f$, with $f$ an idempotent of $J_p$.
  As $s,t,z$ belong to $J_p$, this shows that $s=t$.
\end{proof}

For a compact semigroup $S$
and $p\in\mathfrak L(S)$, let $U_p$
be the union of the maximal subgroups of $J_p$.
The next proposition, whose proof relies on
Theorem~\ref{t:cancellation}, should be compared with
Proposition~\ref{p:idempotents-in-bijection-with-sim-class}.
We show
that $U_p$ parameterizes in a natural way the class~$p$,
without assuming aperiodicity, but assuming equidivisibility and
finite cancelability.

\begin{Prop}
  \label{p:nontrivial-strongly-connected-components}
  Let $S$ be an equidivisible profinite semigroup which is
  finitely cancelable.
  Let
  $p\in\mathfrak{L}(S)$ be a stationary
  point.
  Then we have a bijection $\nu_p\colon U_p\to p$, defined as follows:
  \begin{enumerate}
  \item   for each idempotent $e\in J_p$, fix an element 
  $(u_e,v_e)$ of $p$ stabilized by~$e$;
  \item to each $g$ in the maximal subgroup $H_e$ of $S$ containing
    $e$, associate the element $\nu_p(g)=(u_eg,g^{\omega-1}v_e)$ of $p$.
  \end{enumerate} 
\end{Prop}

\begin{proof}
  We first verify that the function $\nu_p$ is well defined.
  Corollary~\ref{c:idempotents-vs-sim-classes}
  guarantees that every idempotent $e\in J_p$ stabilizes
  some $(u_e,v_e)\in p$.
  If $g\in H_e$, then
  $(u_eg,g^{\omega-1}v_e)
  \xrightarrow {g^{\omega-1}}(u_e,v_e)$
  and
    $(u_e,v_e)
  \xrightarrow {g}(u_eg,g^{\omega-1}v_e)$
  are edges of $\Cl T(S)$, whence
  $(u_eg,g^{\omega-1}v_e)\sim (u_e,v_e)$.
  It follows that $\nu_p$ is indeed well defined.

  We next show that $\nu_p$~is injective.
  If $g\in U_p$ then $g^\omega$ is an idempotent of $J_p$ stabilizing
  $\nu_p(g)=(u_{g^\omega}g,g^{\omega-1}v_{g^\omega})$.
  Hence, by
  Lemma~\ref{l:minimum-idempotents-labeling-loops-at-same-point}, if
  $\nu_p(g)=\nu_p(h)$, then $g^\omega=h^\omega$,
  thus $g,h$
  belong to the same maximal subgroup
  and $u_{g^\omega}g=u_{g^\omega}h$.
  The latter is equivalent to
  $u_{g^\omega}=u_{g^\omega}hg^{\omega-1}$,
  and so, by
  Theorem~\ref{t:cancellation}, we have $hg^{\omega-1}=g^\omega$.
  This means that $g=h$,
  thereby showing that $\nu_p$ is injective.

  It remains to show that $\nu_p$ is surjective. Let $(u,v)$ be an
  arbitrary element of~$p$ and let $e$ be an idempotent in $J_p$
  stabilizing~$(u,v)$. Since $p$ is the set of vertices of a
  strongly connected component of~$\Cl T(S)$, there is some edge
  $(u_e,v_e)\xrightarrow t(u,v)$, and whence also an edge
  $(u_e,v_e)\xrightarrow{ete}(u,v)$.
  By Lemma~\ref{l:a-trivial-characterization-of-Jp}, it
  follows that $ete\in H_e$. Hence, we must have $\nu_p(ete)=(u,v)$.
\end{proof}

Note that, under the conditions of
Proposition~\ref{p:nontrivial-strongly-connected-components},
the set $p_e$ of elements of $p$ stabilized by $e$ is precisely
$\nu_p(H_e)$.

\section{A characterization of the $\Cl J$-class
  associated to a $\sim$-class}

Let $S$ be an equidivisible compact semigroup, $w\in S$ and
$p\in\mathfrak{L}(w)$.
We define a subset $L_p$ of $S^I$, depending only on $p$,
as follows.
Take an arbitrary strictly increasing sequence
$p_1<p_2<\cdots$ 
converging to $p$
in~$\mathfrak{L}(w)$ --- if such a sequence does not exist, for instance,
if $p$ has a predecessor, then take $L_p=\emptyset$.
For each $m\geq1$ and $n> m$,
let $t_{m,n}$ be a transition from $p_m$ to~$p_n$.
For fixed $m\geq 1$, let $t_m$ be an accumulation point
of the sequence $(t_{m,n})_{n> m}$.
If $t$ is an accumulation point of the sequence $(t_m)_m$
then $t\in L_p$, and every element of $L_p$ is obtained by this
process, the sequence
$p_1<p_2<\cdots$ being allowed to change.
Dually, taking strictly decreasing sequences converging to~$p$, we
define a subset $R_p$ of $S$ associated with~$p$.

\begin{Thm}
  \label{t:J-class-type-2}
  Let $S$ be an equidivisible compact semigroup.
  For every $p\in\mathfrak{L}(S)$, the
  sets $L_p$ and $R_p$ are contained in~$J_p$.
\end{Thm}

\begin{proof}
  Let $w\in S$ be such that $p\in\mathfrak L(w)$, and
  suppose that $(p_n)_n$ is a strictly increasing sequence of
  elements of $\mathfrak{L}(w)$ converging to~$p$.
  For each $m\geq 1$ and $n>m$, let $t_{m,n}$ be a transition from $p_m$
  to~$p_n$, and let $t_m$ be an accumulation point of
  the sequence $(t_{m,n})_n$. Finally, let $t$ be an accumulation point of the
  sequence $(t_m)_m$. The proof that $L_p$ is contained in $J_p$ is concluded once we show that $t\in J_p$.

  We first claim that, for every fixed $m\geq1$, $t_m$~is a transition
  from~$p_m$ to~$p$. For each $n>m$, choose $(u_n,v_n)\in p_m$ and
  $(x_n,y_n)\in p_n$ such that $u_nt_{m,n}=x_n$ and $v_n=t_{m,n}y_n$.
  By compactness,
  the sequence $(t_{m,n},u_n,v_n,y_n)_n$
  has some subnet
  $(t_{m,n_k},u_{n_k},v_{n_k},y_{n_k})_{n_k}$
  such that $t_m=\lim_{k} t_{m,n_k}$ and the nets $(u_{n_k})_k$,
  $(v_{n_k})_k$, and $(y_{n_k})_k$ converge, respectively to some $u$, $v$,
  and~$y$.
  Since
  $(u_{n_k},t_{m,n_k}y_{n_k})=(u_{n_k},v_{n_k})\in p_m$
  and
  $(u_{n_k}t_{m,n_k},y_{n_k})=(x_{n_k},y_{n_k})\in p_{n_k}$, it follows
  from Proposition~\ref{p:topology-of-Fw}
  that $(u,v)=(u,t_my)\in p_m$ and $(ut_m,y)\in p$, which proves the claim.

  Since $(p_m)_m$ converges to~$p$ and $t$ is an accumulation point of the
  sequence $(t_m)_m$, it follows
  again from Proposition~\ref{p:topology-of-Fw}
  that $t$~is a transition from $p$
  to~$p$.

  Choose $z\in J_p$
  such that there is a loop $(ut_m,y)\xrightarrow z(ut_m,y)$.
  Since $(u,v)<(ut_m,y)$ and $(u,v)\xrightarrow{t_m}(ut_m,y)$
  is an edge of $\Cl T(S)$, Lemma~\ref{l:suffix-transition}
  yields that $z$ is a suffix of $t_m$. This holds
  for every~$m\geq1$, whence $z$ is a suffix of $t$. As we
  have already shown that $t$ is a transition from $p$ to~$p$, we
  deduce that $t\in J_p$ by Lemma~\ref{l:a-trivial-characterization-of-Jp}.
  Hence we have $L_p\subseteq J_p$ and
  dually $R_p\subseteq J_p$.
\end{proof}

\begin{Cor}
  \label{c:convergence-of-transitions}
    Let $S$ be an equidivisible profinite semigroup which is
  finitely cancelable.
  Let $w\in S$. Suppose that $(u_n,v_n)_n$ is a strictly
  increasing sequence in~$\mathfrak{F}(w)$ converging to a stationary
  point $(u,v)$. For each pair $m<n$, let $t_{m,n}\in S$ be a
  transition $(u_m,v_m)\to(u_n,v_n)$.
  Then, for each $m$,
  the sequence $(t_{m,n})_n$ converges to
  the unique transition from
  $(u_m,v_m)$ to $(u,v)$.
  Moreover,
  the sequence $(t_m)_m$ converges to the label of the only loop at the
  vertex~$(u,v)$ in the trivial category $\Cl K_p$, where $p=(u,v)/{\sim}$.
\end{Cor}

\begin{proof}
  Every accumulation point of the sequence
  $(t_{m,n})_n$ labels an edge from
  $(u_m,v_m)$ to $(u,v)$. But there is only one such edge, by
  Proposition~\ref{p:transitions-in-F}. Since
  we are dealing with a compact space, this implies
  that $(t_{m,n})_n$ converges to some element $t_m$.
  To conclude the proof, it suffices to
  show that every accumulation point $t$ of the
  sequence $(t_m)_m$ labels the same loop at vertex~$(u,v)$.
  By the definition of $L_p$, we have $t\in L_p$,
  whence $t\in J_p$ by Theorem~\ref{t:J-class-type-2}.
  Moreover, in $\Cl T(w)$ the sequence
  of edges
  $(u_m,v_m)\xrightarrow{t_m}(u,v)$
  admits the loop
  $(u,v)\xrightarrow{t}(u,v)$
  as an accumulation point.
  Since $t\in J_p$, this loop belongs
  to~$\Cl K_p$
  by Corollary~\ref{c:minimum-idempotents-labeling-loops-at-same-point}.
  By Corollary~\ref{c:trivial-groupoid-at-type-2-point}, there is only
  one loop of $\Cl K_p$ at~$(u,v)$. Therefore,
  every accumulation point of $(t_m)_m$ is the label of
  that loop.
\end{proof}

In some cases, Theorem~\ref{t:J-class-type-2} can be strengthened, as
seen next.

\begin{Thm}
  \label{t:Jp-equals-Rp-equals-Lp}
    Let $A$ be a finite alphabet, and let
  $\pv V$ be a pseudovariety closed under concatenation.
  Then $L_p=J_p=R_p$ for every stationary point $p\in \mathfrak{L}(\Om AV)$.
\end{Thm}

\begin{proof}
  According to Theorem \ref{t:J-class-type-2},
  it remains to prove
  that $J_p$ is contained in $L_p$ and in $R_p$.
  By symmetry, it suffices to
  prove that $J_p\subseteq L_p$.

  Let $\tau$ be an element of $J_p$. Then by
  Proposition~\ref{p:Jp-vs-labels-in-strongly-connected-component}
  there are elements $(u,v)$
  and $(u',v')$ of $p$ such that $(u,v)\xrightarrow \tau(u',v')$ is an
  edge of $\Cl K_p$. According to
  Proposition~\ref{p:approximation-type-2-by-step-points}, there are
  strictly increasing sequences $(q_n)_n$ and $(q'_n)_n$ of step
  points of $\mathfrak F(w)$ converging to $(u,v)$ and $(u',v')$,
  respectively. We may define recursively
  a strictly ascending sequence of step points $p_n$ as follows:
  \begin{itemize}
    \item $p_1=q_1$;
    \item if $n>1$ is even, then $p_n$
      is the smallest term of the sequence $(q'_n)_n$
      belonging to $]p_{n-1},p[$;         
    \item if $n>1$ is odd, then $p_n$
      is the smallest term of the sequence $(q_n)_n$ belonging to $]p_{n-1},p[$;      
  \end{itemize}
  For each $m\geq1$
  and each $n\geq m$
  let $t_{m,n}$
  be the unique transition from $p_{2m-1}$ to $p_{2n}$.
  Let $t_m$ be an accumulation point of the sequence
  $(t_{m,n})_{n\geq m}$.
  Since $(p_{2n})_n$ converges to
  $(u',v')$, the pseudoword $t_m$ labels an edge
  from $p_{2m-1}$ to $(u',v')$.
  Let $t$ be an accumulation point of the sequence $(t_m)_m$.
  Since $(p_{2m-1})_m$ converges to
  $(u,v)$, the pseudoword $t$ labels an edge
  from $(u,v)$ to $(u',v')$.
  By the definition of $L_p$, we have $t\in L_p$.
  Therefore
  $t\in J_p$ by Theorem~\ref{t:J-class-type-2}.
  By Corollary~\ref{c:trivial-groupoid-at-type-2-point},
  the category $\Cl K_p$ is trivial,
  whence $t=\tau$,
  thus proving that $\tau\in L_p$.
\end{proof}

\section{Proof of Theorem~\ref{t:a-description-of-the-representation-theorem}}
\label{sec:proof-theorem-reft:a}

Throughout this section, when not explicitly stated, we consider $w$
to be an element of $\Om AA$, where $A$ is a finite alphabet, and
$\varphi\colon A\to S$ to be a generating mapping of a finite
aperiodic semigroup~$S$. We also take $s=\varphi_{\pv A}(w)$.

Consider a mapping $g\colon \step(\mathfrak {L}(w))\to \mathfrak {F}(s)$.
From hereon, we assume that $\mathfrak {L}_c(w)$ is $g$-recognized by the
pair $(\varphi,s)$.  In particular, all properties of
  Definition~\ref{def:successful-run} are fulfilled when
  $(L,\ell)=\mathfrak{L}_c(w)$.

\begin{Def}[$g$-projection]
Let $x$ and $y$ be two step points of $\mathfrak {L}(w)$
such that $x\leq y$.
By Propositions~\ref{p:transitions-in-F} (in case $x<y$)
and~\ref{p:shifting-allowed-implies-type-2} (in case $x=y$),
there is a unique edge
in $\Cl T(w)$
from $x$ to $y$. Let $t$ be its label, which is $1$ if $x=y$.
If $g(x)\xrightarrow{\varphi_{\pv A}(t)} g(y)$ is
an edge of $\Cl T(s)$,
then we say that the edge $x\xrightarrow{t}y$ is \emph{$g$-projected (to
$g(x)\xrightarrow{\varphi_{\pv A}(t)} g(y)$)}.
\end{Def}

The unique edge from $x$ to $y$ will sometimes be denoted simply
by $x\xrightarrow{}y$, without reference to the label.

\begin{Remark}
  \label{r:a-sort-of-homomorphism}
  Let $x$, $y$ and $z$ be step points such that
  $x\leq y\leq z$. If
  $x \xrightarrow{}y$
  and $y \xrightarrow{}z$
  are $g$-projected, then $x \xrightarrow{}z$
  is $g$-projected.
\end{Remark}

  Given two step points $x$ and $y$, write $x\prec\prec y$
if $x\leq y$ and the interval
$[x,y]$ is finite.

\begin{Remark}
  \label{r:projection-of-finite-paths}
  If $x\prec\prec y$, then $x \xrightarrow{}y$ is $g$-projected.
\end{Remark}

Let $\approx$ be the equivalence relation on $\step(\mathfrak{L}(w))$ generated by
$\prec\prec$,
that is, $x\approx y$ if and only
$x\prec\prec y$
or $y\prec\prec x$. The $\approx$-class of $x$ is denoted $[x]_\approx$.
Note that $w\in A^+$
if and only if $(1,w)\approx (w,1)$.
Let $\Cl O_w$ be a subset of
$\step(\mathfrak {L}(w))$ such that each
$\approx$-class contains exactly one element of 
$\Cl O_w$, with  the additional restriction
that if $w\notin A^+$ then we have
\begin{equation*}
\Cl O_w\cap [(1,w)]_{\approx}=\{(1,w)\}\quad
\text{and}\quad \Cl O_w\cap [(w,1)]_{\approx}=\{(w,1)\}.  
\end{equation*}

\begin{Def}[Bridges]
  A \emph{bridge} in $\mathfrak{L}(w)$, with respect to
  the mapping~$g$,
  is a nonempty open interval $I$ of $\mathfrak{L}(w)$
  such that, for every pair of step points
  $x,y$ of $I$, with $x<y$,
  the edge $x\rightarrow {}y$ is $g$-projected.
  A \emph{special bridge} in $\mathfrak{L}(w)$
  is a bridge of the form ${[}X,Y{[}$,
    with $X,Y\in \Cl O_w$ such that $X<Y$.
\end{Def}

Notice that every nonempty interval contained in a bridge is also a bridge, and that special bridges are clopen intervals.

\begin{Lemma}\label{l:cover-of-special-bridges}
  Let $\Cl F$ be a nonempty family of special bridges
  of $\mathfrak {L}(w)$. 
  If $\bigcup\Cl F$ is a closed interval,
   then $\bigcup\Cl F$ is a special bridge.
\end{Lemma}

Before proving Lemma~\ref{l:cover-of-special-bridges}, we remark that
the hypothesis that the union $\bigcup\Cl F$ is closed cannot be removed.
Indeed, consider a case where we have a strictly increasing sequence $(q_n)_n$ of stationary points converging to a stationary point $q$
(cf.~Example~\ref{ex:examples-of-cluster-words}).
Between $q_n$ and $q_{n+1}$ pick an element $p_n$ of $\Cl O_w$.
Then the union of the special bridges
${[}\min \mathfrak {L}(w),p_n{[}$ is ${[}\min \mathfrak {L}(w),q{[}$,
which is not a special bridge.

\begin{proof}[Proof of Lemma~\ref{l:cover-of-special-bridges}]
  Since $\bigcup\Cl F$ is a compact set having an open cover by the
  elements of $\Cl F$,
    we have $\bigcup\Cl F=\bigcup\Cl F'$
    for some finite subfamily $\Cl F'$ of~\Cl F.
    Then, for some $n\geq 1$, we may assume that
    there are elements  $X_1<X_2<\cdots <X_n$ and
    $Y_1<Y_2<\cdots <Y_n$ of $\Cl O_w$
    such that
    \begin{equation*}
    \Cl F'=\{{[}X_k,Y_k{[}:1\leq k\leq n\}.      
    \end{equation*}
    Consider the set
    \begin{equation*}
    \Cl Z=\{X_i:1\leq k\leq n\}\cup \{Y_i:1\leq k\leq n\}.  
    \end{equation*}
    and let $Z_1<Z_2<\cdots <Z_m$ be the elements of $\Cl Z$. Notice that $m\geq 2$.
    For each $k\in \{1,\ldots,m-1\}$, denote
    by $I_k$ the interval ${[}Z_k,Z_{k+1}{[}$.
    In case $Z_k=X_j$ for some $j$,
    we have $Z_{k+1}\leq Y_j$, and so $I_k$ is a special bridge contained
    in the special bridge ${[}X_j,Y_j{[}$.
    If $Z_k=Y_j$ for some $j$,
    then, since $Z_k<Z_m$, we must have $Z_k\in\bigcup\Cl F$, whence
    $Z_k\in {[}X_i,Y_i{[}$ for some $i$.
    As $Z_k<Y_i$, we have $Z_{k+1}\leq Y_i$, and so  
    $I_k$ is a special bridge contained
    in the special bridge ${[}X_i,Y_i{[}$.
    Hence,  the set
    \begin{equation*}
    \Cl F''=\{I_k:1\leq k\leq m-1\}      
    \end{equation*}
    is a family of special bridges.
    
    Let $U_k=\bigcup_{j=1}^k I_j$.
    Note that $U_k={[}Z_1,Z_{k+1}{[}$.
    We prove by induction on $k\in\{1,\ldots,m-1\}$ that
    $U_k$ is a special bridge.
    The initial step is trivial. Suppose
    that $U_k$ is a special bridge,
    for some $k\in\{1,\ldots,m-2\}$.
    To prove that $U_{k+1}$ is
    a special bridge, consider two step points $x,y$ in $U_{k+1}$ with $x<y$. We have to show that the edge $x\xrightarrow{} y$ is $g$-projected. Using
    induction, it suffices to consider the case where  $x\in U_k$ and $y\in I_{k+1}$. Let $Z'_{k+1}$ be the
    predecessor of the step point $Z_{k+1}$ in~$\mathfrak{L}(w)$.
    Since $U_k={[}Z_1,Z_{k+1}{[}$ and $I_{k+1}={[}Z_{k+1},Z_{k+2}{[}$,
    the edges $x\xrightarrow{} Z'_{k+1}$ and $Z_{k+1}\xrightarrow{} y$
    are $g$-projected, and the same is true obviously for the edge
    $Z'_{k+1}\xrightarrow{} Z_{k+1}$;
    hence, by Remark~\ref{r:a-sort-of-homomorphism},
  $x\xrightarrow{}y$ is $g$-projected. This proves that
    $U_{k+1}$ is a special bridge, concluding the induction.
    
    The result now follows because $U_{m-1}=\bigcup \Cl F''=\bigcup \Cl F'=\bigcup\Cl F$.
  \end{proof}

  In the proof of the next lemma and in the sequel,
  we use the notation $\lambda(x)=\inf [x]_\approx$ and
     $\rho(x)=\sup[x]_\approx$, for a step point $x$.
     Note that $\lambda(x)$ is stationary unless
     $(1,w)\prec\prec x$, and $\rho(x)$ is stationary unless
     $x\prec\prec (w,1)$. Also, when $r\in \step(\mathfrak {L}(w))$,
the unique element of $\Cl O_w\cap [r]_{\approx}$
is denoted~$\Cl O_w(r)$.

  \begin{Lemma}\label{l:cover-of-special-bridges-2}
    Let $X,Y\in\Cl O_w$ be such that $X<Y$. Suppose that for every
    stationary point $p$ in ${[}X,Y{[}$ there is a special bridge
    containing $p$. Then ${[}X,Y{[}$ is a special bridge.
  \end{Lemma}

  \begin{proof}
    Let $P={[}X,Y{[}\cap \stat(\mathfrak {L}(w))$.
    By hypothesis, for each $p\in P$
    there is a special bridge $I_p$ containing~$p$. Note that
    $I_p\cap {[}X,Y{[}$ is also a special bridge containing~$p$. Therefore, we may as well assume that $I_p\subseteq
    {[}X,Y{[}$. Let $U=\bigcup_{p\in P}I_p$.
  Then $U\subseteq {[}X,Y{[}$.

  We claim that $U={[}X,Y{[}$.
    As we trivially have $P\subseteq U$,
    we are reduced to showing that,
    for every $z\in {[}X,Y{[}\cap \step(\mathfrak{L}(w))$,
    we have $z\in U$.
    
    If $\rho(X)<z<\lambda (Y)$ holds, then
    the stationary points $\rho(z)$ and $\lambda(z)$ belong to ${[}X,Y{[}$.
    Either $z\in {]}\lambda(z),\Cl O_w(z){[}$ or $z\in {[}\Cl
    O_w(z),\rho(z){[}$.
    As ${]}\lambda(z),\Cl O_w(z){[}\subseteq I_{\lambda(z)}$
    and ${[}\Cl O_w(z),\rho(z){[}\subseteq I_{\rho(z)}$,
    it follows that $z\in U$.

    Suppose that $\rho(X)\geq z$.
    Then $z\in {[}X,\rho(X){[}$, thus $z\in I_{\rho (X)}$.
    Similarly, if $\lambda (Y)\leq z$  then
    $z\in I_{\lambda (Y)}$. In both cases $z\in U$.

    We proved that $U={[}X,Y{[}$. By
    Lemma~\ref{l:cover-of-special-bridges}, the interval ${[}X,Y{[}$
    is a special bridge.
  \end{proof}

  Let $p_1$ and $p_2$ be two elements of $\mathfrak {L}(w)$
  such that $p_1\leq p_2$.
  If $p_1<p_2$, then, as written in Corollary~\ref{c:transitions-in-F},
  the set of transitions from $p_1$ to $p_2$ is contained
  in a $\Cl J$-class $J_{p_1,p_2}$. In case $p_1=p_2$,
  let $J_{p_1,p_2}=J_{p_1}$.
  By a \emph{$\Cl J$-minimum transition} from $p_1$ to $p_2$,
  we mean a transition from $p_1$
  to $p_2$ belonging to $J_{p_1,p_2}$; this terminology is useful to
  unify both cases $p_1<p_2$ and $p_1=p_2$ in some of our arguments.

  \begin{Lemma}\label{l:the-j-minimum-transitions}
    Let $p_1,p_2\in\mathfrak{L}(w)$, with $p_1\leq p_2$.
    For every $(x_1,y_1)\in p_1$ and $(x_2,y_2)\in p_2$,
    the intersection $I$ of $J_{p_1,p_2}$
    with the set of labels of edges from $(x_1,y_1)$ to
    $(x_2,y_2)$ is a singleton. If
    $t $ is the unique element of $I$,
    and if $e_i$ is the unique idempotent of $J_{p_i}$
    that stabilizes $(x_i,y_i)$ (cf.~Lemma~\ref{l:minimum-idempotents-labeling-loops-at-same-point}), then $t=e_1te_2$.
  \end{Lemma}

  \begin{proof}
    If $p_1<p_2$, then the lemma follows
    straightforwardly from Proposition~\ref{p:transitions-in-F}.
    If $p_1=p_2=p$, then 
    $J_{p_1,p_2}=J_{p}$.
    Since edges in $\mathcal T_p$ labeled by elements of
    $J_p$ are edges of $\Cl K_p$
    by Corollary~\ref{c:minimum-idempotents-labeling-loops-at-same-point},
    and since $\Cl K_p$ is trivial by
    Corollary~\ref{c:trivial-groupoid-at-type-2-point}, 
    we have also in this case that there is only one
    edge from $(x_1,y_1)$ to $(x_2,y_2)$
    with label in $J_{p_1,p_2}$.
    Let $e_i$ and $t$ be as in the statement of the lemma.
    Because $e_i$ stabilizes $(x_i,y_i)$,
    we may consider the edge
    $(x_1,y_1)\xrightarrow{e_1te_2}(x_1,y_2)$, 
    and so $t\mathrel{\Cl J}e_1te_2$.
    By the already proved uniqueness, we must have
    $t=e_1te_2$.
  \end{proof}

  \begin{Def}[{$\Cl J$-bridge}]
    A pair of elements $(p_1,p_2)$ of $\stat(\mathfrak{L}(w))$ is a
    \emph{\Cl J-bridge with respect to $g$}, if
    \begin{enumerate}
    \item $p_1\leq p_2$;
    \item $F_g(p_1)\cap F_g(p_2)\neq\emptyset$;
    \item the elements of $\varphi_{\pv A}(J_{p_1})$
      are \Cl J-equivalent to the elements of $\varphi_{\pv A}(J_{p_2})$;
    \item if $\tau$ is a \Cl J-minimum
      transition from an element
      of $p_1$ to an element of $p_2$, then $\varphi_{\pv A}(\tau)$
      is \Cl J-equivalent to the elements of $\varphi_{\pv A}(J_{p_1})$.    
    \end{enumerate}
\end{Def}

Note that $(p,p)$ is a \Cl J-bridge, whenever $p$ is a stationary point.

 In the proof of the following
 proposition, the hypothesis that $S$ is an unambiguous aperiodic semigroup
 is essential.
 
\begin{Prop}
  \label{p:towards-big-jumps}
  Suppose that $S$~is unambiguous and let $(p_1,p_2)$ be a \Cl J-bridge with respect to~$g$. Suppose there are step points $z_1$ and $z_2$ such that
  ${[}z_1,p_1{[}$ and ${]}p_2,z_2{]}$
  are bridges with respect to $g$.
  Let $(s_1,s_2)\in F_g(p_1)\cap F_g(p_2)$.  
  Then there are step points $x_1\in [z_1,p_1[$
  and $x_2\in {]}p_2,z_2]$
  such that
  $x_1\xrightarrow{} x_2$
  is $g$-projected to an idempotent loop
  of $\Cl T(s)$ of the form $(s_1,s_2)\xrightarrow{} (s_1,s_2)$.
\end{Prop}

\begin{proof}
  Let $I_1=[z_1,p_1[$ and $I_2={]}p_2,z_2]$.
  By the definition of $F_g(p_i)$, there are
  strictly monotone sequences $(r_{i,m})_{m\geq 1}$ (increasing if
  $i=1$, decreasing if $i=2$)  of step points in $I_i\cap g^{-1}(s_1,s_2)$,
  converging in $\mathfrak {F}(w)$ to elements $r_i$ of $p_i$ ($i=1,2$).
    Denote by $e_i$
    the unique idempotent of $J_{p_i}$
    that stabilizes~$r_i$.

    For each $n\geq m$, let $t_{1,m,n}$ be the label of the unique
    edge from $r_{1,m}$ to $r_{1,n}$,
    and let
    $t_{2,m,n}$ be the label of the unique
    edge from $r_{2,n}$ to $r_{2,m}$.
    By Corollary~\ref{c:convergence-of-transitions}
    and its dual, for each $i\in\{1,2\}$ the sequence $(t_{i,m,n})_n$
    converges to an element $t_{i,m}$, and
    in turn the sequence $(t_{i,m})_m$ converges to~$e_i$.
    Therefore, by continuity of $\varphi_{\pv A}$,
    for each $i\in\{1,2\}$ we may
    take $m_i\geq 1$ and $n_i\geq m_i$ for which we have
    \begin{equation}\label{eq:approximation-jumps}
    \varphi_{\pv A}(e_i)=\varphi_{\pv A}(t_{i,m_i})=\varphi_{\pv A}(t_{i,m_i,n_i}).  
    \end{equation}

   Our hypothesis yields that the edge
  $r_{1,m_1}\xrightarrow{t_{1,m_1,n_1}} r_{1,n_1}$
  is $g$-projected
  to the edge
  $(s_1,s_2)\xrightarrow{\varphi_{\pv A}(e_1)}(s_1,s_2)$.
  In particular, $(s_1,s_2)$
  is stabilized by $\varphi_{\pv A}(e_1)$. Similarly,
  $(s_1,s_2)$
  is stabilized by $\varphi_{\pv A}(e_2)$.
  Therefore, $\varphi_{\pv A}(e_1)=\varphi_{\pv A}(e_2)$, by
  Lemma~\ref{l:minimum-idempotents-labeling-loops-at-same-point},
  since the finite semigroup $S$ is unambiguous.

  By Lemma~\ref{l:the-j-minimum-transitions},
   the unique minimum $\Cl J$-transition $\tau$
    from $r_1$ to $r_2$ is such that
    $\tau=e_1\tau e_2$.
  By the definition
  of \Cl J-bridge, we
  have $\varphi_{\pv A}(e_1)\mathrel{\Cl J}\varphi_{\pv A}(\tau)$.
  On the other hand, $\tau\leq_{\Cl R}e_1$ and $\tau\leq_{\Cl L}e_2$.
  Hence $\varphi_{\pv A}(e_1)\mathrel{\Cl H}\varphi_{\pv A}(\tau)$ because
  $\varphi_{\pv A}(e_1)=\varphi_{\pv A}(e_2)$ and $S$ is stable. Therefore,
  \begin{equation}
    \label{eq:e1-equals-tau-equals-e2}
  \varphi_{\pv A}(e_1)=\varphi_{\pv A}(\tau)=\varphi_{\pv A}(e_2),    
  \end{equation}
  since $S$ is aperiodic.       

  The unique edge (cf.~Proposition~\ref{p:transitions-in-F})
  from $r_{1,m_1}$ to $r_{2,m_2}$,
  with label $\zeta$,
  is factorized by the edges
  $r_{1,m_1}\xrightarrow{t_{1,m_1}} r_{1}$,
  $r_{1}\xrightarrow{\tau} r_{2}$
  and
  $r_{2}\xrightarrow{t_{2,m_2}} r_{2,m_2}$.
  Therefore $\zeta=t_{1,m_1}\tau t_{2,m_2}$.
  From~\eqref{eq:e1-equals-tau-equals-e2}
  and~\eqref{eq:approximation-jumps} we get
  $\varphi_{\pv A}(\zeta)=\varphi_{\pv A}(e_1)$.
  Therefore, $r_{1,m_1}\xrightarrow{} r_{2,m_2}$
  is $g$-projected to the idempotent loop
  $(s_1,s_2)\xrightarrow{\varphi_{\pv A}(\zeta)} (s_1,s_2)$.
  Since $r_{i,m_i}$ is a step point in $I_i$, this concludes the proof.
\end{proof}

We next show a pair of lemmas where the function~$g$ does not appear,
but which will be used later in this section in the proof of
Theorem~\ref{t:a-description-of-the-representation-theorem}.

\begin{Lemma}
  \label{l:transitions-near-a-stationary-point-at-finite-semigroup-level}
  Consider a pseudovariety $\pv V$ closed under concatenation.
  Let $w$ be an element of $\Om AV$.
  Consider a generating mapping $\varphi\colon A\to S$ of a semigroup in $\pv V$.
  For every stationary point $p$ of $\mathfrak{L}(w)$,
  there is a clopen interval $I$ containing $p$
  such that, if $t$ is a transition between elements $q,r$ of $I$
  with $q\leq r$, then
  $\varphi_{\pv V}(t)$ is
  a factor of the elements of~$\varphi_{\pv V}(J_p)$.
\end{Lemma}

\begin{proof}
  Let $(u,v)\in p$. By
  Proposition~\ref{p:approximation-type-2-by-step-points}
  and its dual, there is in $\mathfrak F(w)$ an increasing sequence
  $(p_n)_n$
  of step points
  converging to $(u,v)$
  and a decreasing sequence
  $(p'_n)_n$
  of step points converging to $(u,v)$.
  Denote by $\tau_n$
  the unique edge
  from $p_n$ to $(u,v)$,
  and by
  $\tau'_n$
  the unique edge
  from $(u,v)$ to $p'_n$ (uniqueness of these edges is
  from Proposition~\ref{p:transitions-in-F}).
  Let $\varepsilon$ be a loop in $\Cl T_p$ at vertex $(u,v)$.
  Then $\tau_n\varepsilon\tau'_n$ converges to a loop $\tilde\varepsilon$ of
  $\Cl T(w)$ at vertex $(u,v)$. Since
  $\varepsilon$ is a factor of $\tilde\varepsilon$,
  we must have $\Lambda(\tilde\varepsilon)\in J_p$
  (thus we have $\tilde\varepsilon=\varepsilon$ by
  Corollary~\ref{c:trivial-groupoid-at-type-2-point}, although we
  shall not use that fact).
  Therefore, since $\varphi_\pv A$ and the labeling mapping $\Lambda$
  are continuous,
  there is $N$ such that
  $\varphi_\pv V(\Lambda(\tau_n\varepsilon\tau'_n))\in\varphi_\pv V(J_p)$
  for all $n\geq N$.
  Consider the clopen interval $I=[p_N,p'_N]$.
  Let $q$ and $r$ be elements of $I$ such that $q\leq r$,
  and let $t$ be a transition from $q$ to $r$. Then
  $t$ is a factor of a transition from $p_N$ to $p'_N$.
  But $\Lambda(\tau_N\varepsilon\tau'_N)$ is the unique transition
  from $p_N$ to $p'_N$, by Proposition~\ref{p:transitions-in-F}. Hence,
  $\varphi_\pv V(t)$ is a factor of the elements of $\varphi_\pv V(J_p)$.
\end{proof}

\begin{Lemma}\label{l:continuity-of-Jp}
  Let $(q_n)_n$ be a sequence of stationary points of
  $\mathfrak {L}(w)$  converging
  to $q$. For each $n$, let $u_n$ be an element of
  $J_{q_n}$. Suppose that $(u_n)_n$ converges to $u$.
  Then $u$ is \Cl J-above $J_q$.
\end{Lemma}

\begin{proof}
  The sequence of edges
  $(q_n\xrightarrow{u_n}q_{n})_n$ of~$\Cl T(w)$
  converges to
  $q\xrightarrow{u}q$.
  Since~$u$ labels a transition from $q$ to itself,
  $u$ is \Cl J-above $J_q$.
\end{proof}

In Lemma~\ref{l:continuity-of-Jp} it is not true
in general that $u\in J_q$.
In Example~\ref{ex:examples-of-cluster-words},
the stationary point $r$ such that $(a^\omega ba^\omega)^\omega$
belongs to $J_r$ is the limit of a sequence $(q_n)_n$ of stationary points 
such that $a^\omega\in J_{q_n}$, but $a^\omega\notin J_{r}$.

\begin{Def}[Mapping $\Gamma$]
  Let $J_S$ be the set of regular \Cl J-classes of~$S$, and
  denote by $\Upsilon$ the mapping
  $\stat(\mathfrak {L}(w))\to J_S$ sending a stationary
  point $p$ to the \Cl J-class containing $\varphi_{\pv A}(J_p)$.
  Denote by $\Gamma$ the mapping
  $\Upsilon\times F_g\colon \stat(\mathfrak {L}(w))\to J_S\times
  \mathrm{Im}(F_g)$
  sending $p$ to $(\Upsilon(p),F_g(p))$.
\end{Def}

In the next lemma, we continue to assume that $\mathfrak {L}_c(w)$
is $g$-recognized by~$(\varphi,s)$.

\begin{Lemma}\label{l:a-sufficient-condition-for-being-special}
  Suppose that the semigroup $S$ is unambiguous.
  Let $X,Y$ be elements of $\Cl O_w$ such that $X<Y$.
  Suppose that the restriction of $\Gamma$ to
  $\stat(\mathfrak{L}(w))\cap {[}X,Y{[}$ is constant.
  Then ${[}X,Y{[}$ is a special bridge.
\end{Lemma}

\begin{proof}
  Let $p\in \stat(\mathfrak{L}(w))\cap {[}X,Y{[}$.
  By Lemma~\ref{l:cover-of-special-bridges-2},
  it suffices to prove that there is a special bridge
  containing $p$.
  Let $I$ be a clopen interval containing $p$
  satisfying the properties described in
  Lemma~\ref{l:transitions-near-a-stationary-point-at-finite-semigroup-level},
  and let $[x_0,y_0]$
  be the clopen interval $I\cap {[}X,Y{[}$, which is indeed nonempty as it
  contains~$p$.
  Note that $x_0$ and $y_0$ are step points.
  Let $X'=\Cl O_w(x_0)$
  and $Y'=\Cl O_w(y_0)$.
  By the definition of $\Cl O_w$, $x_0\geq X$ implies
  $X'\geq X$, and $y_0\leq Y$ implies
  $Y'\leq Y$, thus ${[}X',Y'{[}\subseteq {[}X,Y{[}$.
  Therefore, we have $p\in {[}X',Y'{[}$ and
  ${[}\rho(X'),\lambda(Y'){]}
  \subseteq I\cap {[}X,Y{[}$.
  
  Let $x,y$ be step points in ${[}X',Y'{[}$ such that
  $x<y$.
  We want to show that $x\xrightarrow{} y$ is $g$-projected.
  If
  $x\prec\prec y$,  then we apply
  Remark~\ref{r:projection-of-finite-paths},
  so we suppose that is not the case,
  which implies $\rho (x)\leq\lambda(y)$.
  Note that ${[}\rho(x),\lambda(y){]}
  \subseteq {[}\rho(X'),\lambda(Y'){]}$.
  If $t$ is a $\Cl J$-minimum
  transition from
  $\rho (x)$ to $\lambda(y)$
  then $t\leq_{\Cl J}J_{\rho(x)}$
  and $t\leq_{\Cl J} J_{\lambda(y)}$ by Lemma~\ref{l:the-j-minimum-transitions},
  thus
  $\varphi_{\pv A}(t)\leq_{\Cl J}\Upsilon(\rho(x))$
  and $\varphi_{\pv A}(t)\leq_{\Cl J}\Upsilon(\lambda(y))$.
  Hence
  $\varphi_{\pv A}(t)\leq_{\Cl J}\Upsilon(p)$,
  because the restriction of $\Upsilon$
  to $\stat(\mathfrak{L}(w))\cap {[}X,Y{[}$ is constant.
  By the choice of~$I$
    (cf.~Lemma~\ref{l:transitions-near-a-stationary-point-at-finite-semigroup-level}),
    this implies $\varphi_{\pv A}(t)\in \Upsilon(p)$.
  Therefore, since the restriction of $\Gamma$
  to $\stat(\mathfrak{L}(w))\cap {[}X,Y{[}$ is constant,
  we have shown that the pair $(\rho (x),\lambda(y))$ is a \Cl
  J-bridge.
     By Remark~\ref{r:projection-of-finite-paths},
     ${[}x,\rho(x){[}$
     and ${]}\lambda(y),y{]}$
     are bridges, whence it follows from
     Proposition~\ref{p:towards-big-jumps}
     that there are step points $x'\in {[}x,\rho(x){[}$
     and $y'\in {]}\lambda(y),y{]}$ such that
     $x'\xrightarrow{} y'$ is $g$-projected.
     Therefore $x\xrightarrow{} y$ is $g$-projected
     by Remarks~\ref{r:projection-of-finite-paths}
     and \ref{r:a-sort-of-homomorphism}.     
     This concludes the proof that ${[}X',Y'{[}$ is a special bridge
     containing $p$.
\end{proof}

\begin{Prop}\label{p:the-really-big-jump}
  Suppose that the finite aperiodic semigroup $S$ is unambiguous.
  Then every stationary point $p$
  of $\mathfrak{L}(w)$ is contained in some special bridge.
\end{Prop}

\begin{proof}
  Endow $J_S\times \mathrm{Im}(F_g)$ with the following partial order:
\begin{equation*}
  (J_1,P_1)\leq (J_2,P_2)\iff (J_2<_{\Cl J}J_1\text{ or
  }(J_1=J_2\text{ and }P_1\subseteq P_2)).
\end{equation*}  
For each $p\in\stat(\mathfrak{L}(w))$, let
  \begin{equation*}
    G(p)=\{\Gamma(q)\,:\,q\in \stat(\mathfrak{L}(w))\text{ and }\Gamma(q)<\Gamma(p)\}.
  \end{equation*}
  We shall prove by induction on the cardinal of $G(p)$
  that $p$ is contained in some special bridge.
  
  We start with some preliminary remarks.  
  Let $I=[\alpha,\beta]$ be a clopen interval containing $p$
  satisfying the properties described in
  Lemmas~\ref{l:transitions-near-a-stationary-point-at-finite-semigroup-level}
  and~\ref{l:if-P-is-maximal-then-inverse-image-by-g-is-closed}\ref{item:if-P-is-maximal-2}.
  Let $X=\Cl O_w(\alpha)$ and $Y=\Cl O_w(\beta)$.
  Then the following holds:
  $$  \stat(\mathfrak{L}(w))\cap I=
  \stat(\mathfrak{L}(w))\cap {[}X,Y{[}.$$
  Let $q\in \stat(\mathfrak{L}(w))\cap {[}X,Y{[}$.
  Then we have $\Upsilon(p)\leq_{\Cl J}\Upsilon(q)$
  by the choice
  of $I$
  (cf.~Lemma~\ref{l:transitions-near-a-stationary-point-at-finite-semigroup-level}).
  It also follows from the choice of $I$
  (cf.~Lemma \ref{l:if-P-is-maximal-then-inverse-image-by-g-is-closed}\ref{item:if-P-is-maximal-2})
  that $F_g(q)\subseteq F_g(p)$.
  Hence we have
  \begin{equation}
    \label{eq:a-sort-of-maximality-of-p}
    \Gamma(q)\leq\Gamma(p)\quad \text{for all $q\in \stat(\mathfrak{L}(w))\cap {[}X,Y{[}$}.
  \end{equation}
  
  {\it Initial step.}
  Suppose that $G(p)=\emptyset$. It follows
  from~\eqref{eq:a-sort-of-maximality-of-p}
  that $\Gamma(q)=\Gamma(p)$ for every $q\in \stat(\mathfrak{L}(w))\cap {[}X,Y{[}$.
  Therefore, ${[}X,Y{[}$ is a special bridge, by
  Lemma~\ref{l:a-sufficient-condition-for-being-special}.
  
  {\it Inductive step.}
  Suppose that for every
  stationary point $q$ such that
  $|G(q)|<|G(p)|$, there is a special bridge containing $q$.
  Consider the sets
  \begin{align*}
    M&=\{q\in \stat(\mathfrak{L}(w))\cap {[}X,Y{[}:\Gamma(q)<\Gamma(p)\},\\
    N&=\{q\in \stat(\mathfrak{L}(w))\cap {[}X,Y{[}:\Gamma(q)=\Gamma(p)\}.\
  \end{align*}
  Note that~$M\cup N=\stat(\mathfrak{L}(w))\cap {[}X,Y{[}$,
  by~\eqref{eq:a-sort-of-maximality-of-p}.
  Note also that, by the induction hypothesis,
  every element of $M$ is contained in some special bridge,
  as $\Gamma(q)<\Gamma(p)$ implies $G(q)\subsetneq G(p)$
  with $\Gamma(q)\in G(p)\setminus G(q)$.
  Our goal is to apply Lemma~\ref{l:cover-of-special-bridges-2}
  to the interval ${[}X,Y{[}$. So,
  it remains to show that every element of $N$ is contained in some
  special bridge.
  For that purpose, we prove the following lemma.  
  
  \begin{Lemma}\label{l:Hp-is-closed}
    The set $N$ is closed.
  \end{Lemma}
  
  \begin{proof}
    Consider a sequence $(q_n)_n$ of elements of $N$
    converging to~$q$. Since
    $[X,Y[$ and
    $\stat(\mathfrak{L}(w))$ are closed (the latter in view of
    Proposition~\ref{p:shifting-allowed-implies-type-2}), we only need to prove
    that $\Gamma(q)=\Gamma(p)$.
    
    By
    Lemma~\ref{l:if-P-is-maximal-then-inverse-image-by-g-is-closed}\ref{item:if-P-is-maximal-2},
    we have $F_g(q_n)\subseteq F_g(q)$
    for sufficiently large $n$.
    By the definition of $N$,
    we have $F_g(q_n)= F_g(p)$ for all $n$,
    whence $F_g(p)\subseteq F_g(q)$.
    On the other hand, since
    $q\in {[}X,Y{[}$
    and thus $q\in I$, by the choice of $I$
    (cf.~Lemma~\ref{l:if-P-is-maximal-then-inverse-image-by-g-is-closed}\ref{item:if-P-is-maximal-2})
    we have $F_g(q)\subseteq F_g(p)$.
    This concludes the proof of the equality $F_g(q)=F_g(p)$.
    
    By Lemma~\ref{l:continuity-of-Jp},
    we have $\Upsilon (q)\leq_{\Cl J}\Upsilon (q_n)$
    for some sufficiently large $n$.
    By the definition of $N$,
    we have $\Upsilon (p)=\Upsilon (q_n)$
    for all $n$, whence $\Upsilon (q)\leq_{\Cl J}\Upsilon (p)$.
    On the other hand,
    since $q\in [X,Y[$ we have $\Upsilon (p)\leq_{\Cl J}\Upsilon (q)$
    by~\eqref{eq:a-sort-of-maximality-of-p}.
    Therefore we have $\Gamma(p)=\Gamma(q)$.
  \end{proof}

  Let us proceed with the proof of Proposition~\ref{p:the-really-big-jump}.
    Suppose that $q\in N$.
  Let $I'$ be a clopen interval containing $q$,
  and contained in ${[}X,Y{[}$,
  satisfying the properties described in
  Lemmas~\ref{l:transitions-near-a-stationary-point-at-finite-semigroup-level}
  and~\ref{l:if-P-is-maximal-then-inverse-image-by-g-is-closed}\ref{item:if-P-is-maximal-2}.
  In a similar way as we have built $X,Y$ from $I$,
  we may build from $I'$ elements $X',Y'\in \Cl O_w$ such that
  $q\in {[}X',Y'{[}$ and
  $\stat(\mathfrak{L}(w))\cap {[}X',Y'{[}\subseteq I'$.
  
  Let $x$ and $y$ be two step points of ${[}X',Y'{[}$ such that
  $x<y$. We want to prove that the edge $x\xrightarrow{} y$
  is $g$-projected.
  By Remark~\ref{r:projection-of-finite-paths},
  we may as well assume that $x$ and $y$ are not $\approx$-equivalent.
  
  Suppose $N\cap [x,y]=\emptyset$. Then every element of
  the nonempty intersection $\stat(\mathfrak{L}(w))\cap {[}x,y{[}$
  belongs to $M$ and is therefore
  contained in a special bridge, as observed earlier.
  Hence ${[}\Cl O_w(x),\Cl O_w(y){[}$ is a special bridge
  by Lemma~\ref{l:cover-of-special-bridges-2}.
  It follows that the edge $x\xrightarrow{} y$ is
  $g$-projected (cf.~Remarks~\ref{r:projection-of-finite-paths}
  and~\ref{r:a-sort-of-homomorphism}).
  
  Suppose that $N\cap [x,y]\neq\emptyset$.
  Then
  $N\cap [x,y]$ has a minimum $r_x$ and a maximum $r_y$,
  because $N\cap [x,y]$ is closed by Lemma~\ref{l:Hp-is-closed}.

  Let $t$ be a \Cl J-minimum transition from
  $r_x$ to~$r_y$.
  Then we have $t\leq_{\Cl J} J_{r_x}$ by Lemma~\ref{l:the-j-minimum-transitions},
  thus $\varphi_{\pv A}(t)\leq_{\Cl J}\Upsilon(r_x)$.
  On the other hand, $r_x,r_y\in I'$, and so, by the choice of $I'$
  (cf.~Lemma~\ref{l:transitions-near-a-stationary-point-at-finite-semigroup-level}),
  we know that $\varphi_{\pv A}(t)\geq_{\Cl J}\Upsilon(q)$.
  Since by the definition of $N$, the equalities
  $\Upsilon(r_x)=\Upsilon(r_y)=\Upsilon(q)$ hold,
  we obtain $\varphi_{\pv A}(t)\in \Upsilon(q)$.
  Also by the definition of $N$,
  we conclude that $F_g(r_x)=F_g(r_y)$.
  This shows that $(r_x,r_y)$ is a $\Cl J$-bridge.

  We claim  that ${[}x,r_x{[}$ is a bridge.
  Let $x_1$ and $x_2$ be step points such that $x<x_1\leq x_2<r_x$.
  We want to show that $x_1\xrightarrow{} x_2$ is $g$-projected,
  for which we may as well assume that $x_1$ and
  $x_2$ are not $\approx$-equivalent.
  Recall that $M\cup N=\stat(\mathfrak{L}(w))\cap {[}X,Y{[}$
  by~\eqref{eq:a-sort-of-maximality-of-p}, and so,
  by the definition of $r_x$,
  the set
  $\stat(\mathfrak{L}(w))\cap {[}\Cl O_w(x_1),\Cl O_w(x_2){[}$
  is contained in $M$.
  But we already observed that every element of $M$ is contained in some
  special bridge.
  Hence, ${[}\Cl O_w(x_1),\Cl O_w(x_2){[}$ is a special bridge
  by Lemma~\ref{l:cover-of-special-bridges-2}.
  and so $x_1\xrightarrow{} x_2$ is $g$-projected ((cf.~Remarks~\ref{r:projection-of-finite-paths}
  and~\ref{r:a-sort-of-homomorphism})), thus
  proving the claim. Similarly, ${]}r_y,y{]}$ is a bridge.

  By Proposition~\ref{p:towards-big-jumps},
  there are step points $x_0$ and $y_0$
  satisfying $x<x_0<r_x$ and $r_y<y_0<y$
  such that  $x_0\xrightarrow{} y_0$ is $g$-projected.
  Again
  by the definition of $r_x$ and by~\eqref{eq:a-sort-of-maximality-of-p},
  every element of
  $\stat(\mathfrak{L}(w))\cap {[}\Cl O_w(x),\Cl O_w(x_0){[}$
  belongs to $M$, and so ${[}\Cl O_w(x),\Cl O_w(x_0){[}$
  is a special bridge
  by Lemma~\ref{l:cover-of-special-bridges-2}.
  Hence, the edge 
  $x\xrightarrow{} x_0$ is $g$-projected.
  Similarly, $y_0\xrightarrow{} y$ is $g$-projected.
  It follows that $x\xrightarrow{} y$ is $g$-projected by
  Remark~\ref{r:a-sort-of-homomorphism}.
  
  We showed in all cases that
  $x\xrightarrow{} y$ is $g$-projected. Hence,
  ${[}X',Y'{[}$ is a special bridge containing~$q$.

  We have established that every element of $\stat(\mathfrak{L}(w))\cap {[}X,Y{[}$
  is contained in some special bridge.
  From Lemma~\ref{l:cover-of-special-bridges-2},
  we then deduce that ${[}X,Y{[}$ is a special bridge containing~$p$.
  This concludes the inductive step in our proof, thereby showing
  that the proposition holds.
\end{proof}

\begin{Thm}\label{t:g-is-a-category-homomorphism}
  Let $w$ be an element of $\Om AV$.
  Let $\varphi\colon A\to S$ be a generating mapping of a finite aperiodic
  unambiguous semigroup $S$.
  Suppose that $\mathfrak{L}_c(w)$
  is $g$-recognized by $(\varphi,s)$.
  Then $\mathfrak{L}(w)$ is a bridge.
\end{Thm}

\begin{proof}
  Clearly, if $w\in A^+$ then
  $\mathfrak{L}(w)$ is a bridge by
  Remark~\ref{r:projection-of-finite-paths}.
  Suppose that $w\in \Om AA\setminus A^+$.
  By Proposition~\ref{p:the-really-big-jump},
  for each stationary point $p$ of $\mathfrak{L}(w)$,
  there is a special bridge ${[}X_p,Y_p{[}$ containing it.
  The union of the nonempty family $({[}X_p,Y_p{[})_{p\in
    \stat(\mathfrak{L}(w))}$
  is $\mathfrak{L}(w)\setminus \{(w,1)\}$, therefore $\mathfrak{L}(w)$
  is a bridge, in view of Lemma~\ref{l:cover-of-special-bridges}.
\end{proof}

\begin{proof}[Conclusion of the proof of Theorem~\ref{t:a-description-of-the-representation-theorem}]
  The direct implication in Theorem
  \ref{t:a-description-of-the-representation-theorem} follows directly
  from Proposition~\ref{p:direct-implication-of-theorem-we-search}.
  
  Conversely, suppose that
  $\mathfrak{L}(w)$ is $g$-recognized by $(\varphi,s)$,
  for some $g$. Then $(1,w)\xrightarrow{w} (w,1)$ is $g$-projected
  to $g(1,w)\xrightarrow{\varphi_{\pv A}(w)} g(w,1)$
  by Theorem~\ref{t:g-is-a-category-homomorphism}.
  But $g(1,w)=(1,s)$
  and $g(w,1)=(s,1)$. It follows that $s=\varphi_{\pv A}(w)$.
\end{proof}

\section{The effect of multiplication on the quasi-order}

\label{sec:sort-repr-theor}

In this section, under suitable conditions, we relate
$\mathfrak {F}(uv)$ on one hand, with
$\mathfrak {F}(u)$ and $\mathfrak {F}(v)$ on the other hand.

The quasi-orders considered in this section are all total, as they
stem from the quasi-order of $2$-factorizations of equidivisible
semigroups. We want to compare different intervals of
quasi-ordered sets of $2$-factorizations. This leads us to introduce
the following definitions.
Let $(P,\leq)$ and $(Q,\leq)$
be two quasi-ordered sets,
and let $\varphi$ be a function from $P$ to $Q$.
Recall that $\varphi$ is \emph{monotone} if
$p\leq q$ implies $\varphi(p)\leq \varphi(q)$, for every $p,q\in P$.
Suppose moreover that the quasi-order on $P$ is total.
Then we say that $\varphi$ is
a \emph{quasi-isomorphism} if $\varphi$
is a surjective monotone mapping such that,
for all $p,q\in P$, we have
$p< q\Rightarrow \varphi(p)<\varphi(q)$.
Because the quasi-order on $P$ is total,
we have $\varphi(p)<\varphi(q)\Rightarrow p<q$
and
$\varphi(p)\sim \varphi(q)\Rightarrow p\sim q$, for all $p,q\in P$.
Therefore, $\varphi$ induces the isomorphism
of linearly ordered sets
$\tilde\varphi\colon P/{\sim}\to Q/{\sim}$ sending $p/{\sim}$ to $\varphi(p)/{\sim}$.
In particular, the quasi-order on $Q$ is also total.

Let $I_u$ be an interval of $\mathfrak {L}(u)$
and $I_v$ an interval of $\mathfrak {L}(v)$, for some $u,v$
in a compact semigroup $S$.
A mapping $\theta\colon I_u\to I_v$
is said to be \emph{$J$-preserving}
if $J_p=J_{\theta(p)}$
for every $p\in I_u$.

\begin{Prop}
  \label{p:cut-at-type-1}
  Consider an equidivisible profinite semigroup $S$ which is finitely
  cancelable, and take $w\in S$.
  Let $(u,v)\in\mathfrak{F}(w)$. For $p=(u,v)/{\sim}$,
  let $e$ be the unique
  idempotent of $J_p$ stabilizing~$(u,v)$. Then the following mappings
  are quasi-isomorphisms of intervals of the respective
  totally quasi-ordered sets $\mathfrak{F}(w)$, $\mathfrak{F}(u)$, and $\mathfrak{F}(v)$:
  \begin{align*}
    \lambda_{(u,v)}\colon[(1,u),(u,e)]&\to[(1,w),(u,v)],\\
    (x,y)&\mapsto(x,yv)\\
    \rho_{(u,v)}\colon[(e,v),(v,1)]&\to[(u,v),(w,1)].\\
    (x,y)&\mapsto(ux,y)
  \end{align*}  
  Moreover, the induced isomorphisms
  $\tilde\lambda_{(u,v)}$
  and
  $\tilde\rho_{(u,v)}$
  are $J$-preserving.
\end{Prop}

Before proceeding with the proof of Proposition~\ref{p:cut-at-type-1},
let us recall that the uniqueness of $e$ mentioned in its statement
is guaranteed by Lemma~\ref{l:minimum-idempotents-labeling-loops-at-same-point}.
The next technical lemma will be used
in the proof of Proposition~\ref{p:cut-at-type-1}.

\begin{Lemma}
  \label{l:stabilizer-of-product}
  Let $S$ be an equidivisible compact semigroup,
  and let $x,y,z\in S$.
  If $xy=xyz$,
  then there exists some idempotent $e\in S^I$ such that
  $yz^\omega=ey$ and $xe=x$.
  Dually, if $yz=xyz$, then there exists some idempotent
  $e\in S^I$ such that $x^\omega y=ye$ and $ez=z$.
\end{Lemma}

\begin{proof}
  We deal only with the case $xy=xyz$, as the other case is dual.
  Since $S$ is equidivisible, the
  pairs $(x,y)$ and $(x,yz)$ are comparable
  elements of the quasi-ordered set $\mathfrak{F}(xy)$.
  If $(x,yz)\leq(x,y)$, then there exists
  $t\in S^I$ such that $x=xt$ and $yz=ty$, and so $yz^k=t^ky$
  for every $k\geq1$, whence $yz^\omega=t^\omega y$ and we choose
  $e=t^\omega$. Otherwise,
  $(x,y)<(x,yz)$, and so there exists $u\in S$ such that $x=xu$ and
  $y=uyz$, whence $x=xu^\omega$ and $y=u^\omega yz^\omega$;
  since $yz^\omega=y=u^\omega y$, we may choose $e=u^\omega$.
\end{proof}

\begin{proof}[Proof of Proposition~\ref{p:cut-at-type-1}] 
  By symmetry, it suffices to consider the mapping
  $\lambda=\lambda_{(u,v)}$. 
  By Remark~\ref{r:trivial-preservation},
  and since $ev=v$, the mapping $\lambda$
  indeed takes its values in the interval $[(1,w),(u,v)]$
  and it is monotone.

  Let $(x,z)\in \mathfrak F(w)$ be such that $(x,z)\leq(u,v)$.
  Then there is $t\in S^I$ such that
  $xt=u$ and $z=tv$.
  And, since $u=ue$,
  we deduce that $(x,te)\xrightarrow {t}(u,e)$ is and edge
  of $\mathcal T(u)$, whence $(x,te)$ belongs to
  the interval $[(1,u),(u,e)]$.
  As $\lambda(x,te)=(x,tev)=(x,z)$, we conclude
  that $\lambda$ is surjective.

  To prove that $\lambda$ is a quasi-isomorphism,
  it remains to show that
  if $(x_1,y_1)$, $(x_2,y_2)$ are elements of $[(1,u),(u,e)]$, then
  \begin{equation}
    \label{eq:cut-at-type-1.1}
    (x_1,y_1)<(x_2,y_2) \implies (x_1,y_1v)<(x_2,y_2v).
  \end{equation}
  
  Reasoning by \emph{reductio ad absurdum},
  suppose that the implication fails, that is,
  that $(x_1,y_1)<(x_2,y_2)$ and
  $(x_2,y_2v)\leq (x_1,y_1v)$. 
  We may then consider $s,t\in S^I$ such that
  $x_1t=x_2$, $y_1=ty_2$, $x_2s=x_1$, and $y_2v=sy_1v$.
  The latter equality can be written as $y_2v=st\cdot y_2v$,
  and applying the second case of Lemma~\ref{l:stabilizer-of-product}
  to it, we conclude that there exists an idempotent
  $f\in S^I$ such that $(st)^\omega y_2=y_2f$ and $fv=v$. The
  calculations
  $$x_2\cdot(st)^{\omega-1}s=x_1(ts)^\omega=x_1
  \text{ and } (st)^{\omega-1}s\cdot y_1=(st)^\omega y_2=y_2f$$
  show that $x_2\cdot y_2f=x_1y_1=u$ and
  \begin{equation}
    \label{eq:cut-at-type-1.2}
    (x_2,y_2f)\leq(x_1,y_1)
  \end{equation}
  in $[(1,u),(u,e)]$.
  Since $(x_1,y_1)<(x_2,y_2)$, we reach a contradiction provided
  we prove that $y_2f=y_2$.

  Recall that $v=fv$, and note that $u=x_2y_2f$ implies $u=uf$.
  Hence, as $(u,v)\xrightarrow e(u,v)$ belongs to
  the ideal $\mathcal K_p$, we conclude
  that $(u,v)\xrightarrow {efe}(u,v)$ also belongs to
  $\mathcal K_p$.
  From Corollary~\ref{c:trivial-groupoid-at-type-2-point},
  we get $efe=e$.
  On the other hand, since $(x_2,y_2f)\leq (u,e)$
  (cf.~\eqref{eq:cut-at-type-1.2}) and
  $(x_2,y_2)\leq(u,e)$, we have $y_2fe=y_2f$ and
  $y_2e=y_2$. Hence, $y_2f=y_2 fe=y_2efe=y_2e=y_2$, and so we reach
  the desired contradiction.
  The contradiction was originated by the assumption that
  the implication~\eqref{eq:cut-at-type-1.1} fails in the interval
  $[(1,u),(u,e)]$. Hence, the implication holds, which concludes the
  proof that $\lambda$ is a quasi-isomorphism.

    It remains to show that, for
    $(x,y)\in [(1,u),(u,e)]$, we have $J_q=J_{\tilde\lambda(q)}$, where
    $q=(x,y)/{\sim}$. Let $\varepsilon\in S^I$ be an idempotent.
    Observe that if $\varepsilon$ stabilizes~$(x,y)$
  then it also stabilizes~$(x,y v)$,
  which shows that $J_q$ lies \Cl J-above $J_{\tilde\lambda(q)}$.

  Conversely, suppose that $\varepsilon$ stabilizes $(x,yv)$.
  As $\varepsilon yv=yv$,
  it follows from Lemma~\ref{l:stabilizer-of-product} that
  there is some idempotent
  $f\in S^I$ with  $\varepsilon y=yf$ and $fv=v$.
  
  We claim that $(x,yf)\leq (u,e)$. Suppose
  on the contrary that
  \begin{equation}
    \label{eq:contrad-1}
  (u,e)<(x,yf).  
  \end{equation}
  Then $e\leq_{\Cl L}yf$, and so in particular we have $e=ef$.
  Similarly, from
  \begin{equation}
    \label{eq:contrad-2}
  (x,y)\leq (u,e)
  \end{equation}
 we get $y=ye$. Therefore, $yf=yef=ye=y$,
 yielding a contradiction
 between~\eqref{eq:contrad-1}
 and~\eqref{eq:contrad-2}.
 This shows the claim that $(x,yf)\leq (u,e)$.

 We are now assured that $(x,yf)$ belongs to the domain of
 $\lambda$. From $v=fv$, we get $\lambda(x,y)=(x,yv)=\lambda(x,yf)$.
  We have
  already proved that $\lambda$ is a quasi-isomorphism, and so we
  conclude that $(x,y)\sim(x,yf)$.
  On the other hand, $\varepsilon$ clearly
  stabilizes~$(x,yf)=(x,\varepsilon y)$. Hence, $J_{\tilde\lambda(q)}$
  lies \Cl J-above $J_q$, which proves that the two \Cl J-classes
  coincide.
\end{proof}

\begin{Remark}\label{r:quasi-step-point-case}
  Notice that in Proposition~\ref{p:cut-at-type-1}, in the special case
where $(u,v)$ is a step point, we have $e=1$ and so the domains of
$\lambda_{(u,v)}$ and $\rho_{(u,v)}$ are, respectively,
$\mathfrak {F}(u)$ and $\mathfrak {F}(v)$.
\end{Remark}

\begin{Remark}\label{r:quasi-stationary-point-case}
  Suppose that in
  Proposition~\ref{p:cut-at-type-1}
  we have $e\neq 1$. Let $q$ be the stationary point $(e,e)/{\sim}$ of $\mathfrak {L}(e)$.
Since $e$ stabilizes $(e,e)$ and $J_q$ is $\Cl J$-above $e$,
we have $e\in J_q$. Therefore,
applying Proposition~\ref{p:cut-at-type-1}, we may consider the
quasi-isomorphisms
$\lambda_{(e,v)}\colon[(1,e),(e,e)]\to[(1,v),(e,v)]$
and
$\rho_{(u,e)}\colon[(e,e),(e,1)]\to[(u,e),(u,1)]$.
\end{Remark}

The diagram in Figure~\ref{fig:quasi-iso-stationary-point} may
facilitate the understanding of the applications of
Proposition~\ref{p:cut-at-type-1} in the case in which $(u,v)$ is a stationary point.
\begin{figure}[ht]
  \begin{center}
    \unitlength=1mm
    \begin{picture}(105,36)(0,-1)
      \small
      \put(0,31){\line(1,0){100}}
      \put(0,30){\line(0,1){2}}
      \put(100,30){\line(0,1){2}}
      \put(102,30){$w$}
      \put(60,30){\line(0,1){2}}
      \put(56,27){$(u,v)$}
      \put(0,22){\line(1,0){75}}
      \put(0,21){\line(0,1){2}}
      \put(75,21){\line(0,1){2}}
      \put(77,21){$u$}
      \put(60,21){\line(0,1){2}}
      \put(56,18){$(u,e)$}
      \put(45,13){\line(1,0){30}}
      \put(45,12){\line(0,1){2}}
      \put(75,12){\line(0,1){2}}
      \put(77,12){$e$}
      \put(60,12){\line(0,1){2}}
      \put(56,9){$(e,e)$}
      \put(45,4){\line(1,0){55}}
      \put(45,3){\line(0,1){2}}
      \put(100,3){\line(0,1){2}}
      \put(102,3){$v$}
      \put(60,3){\line(0,1){2}}
      \put(56,0){$(e,v)$}
      \put(30,23){\vector(0,1){7}}
      \put(21,25.5){$\lambda_{(u,v)}$}
      \put(67.5,14){\vector(0,1){7}}
      \put(69.5,17){$\rho_{(u,e)}$}
      \put(52.5,12){\vector(0,-1){7}}
      \put(43.5,7.5){$\lambda_{(e,v)}$}
      \put(88,5){\vector(0,1){25}}
      \put(90,17){$\rho_{(u,v)}$}
    \end{picture}
  \end{center}
  \caption{Quasi-isomorphisms associated with a stationary point $(u,v)$}
  \label{fig:quasi-iso-stationary-point}
\end{figure}
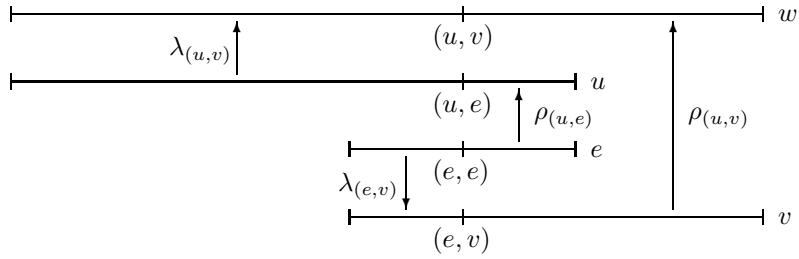
The arrows indicate quasi-isomorphisms between various intervals of
the quasi-ordered sets $\mathfrak{F}(w)$, $\mathfrak{F}(u)$,
$\mathfrak{F}(v)$, and $\mathfrak{F}(e)$. Those quasi-isomorphisms induce
isomorphisms between the corresponding intervals of the linearly
ordered sets $\mathfrak{L}(w)$, $\mathfrak{L}(u)$, $\mathfrak{L}(v)$,
and $\mathfrak{L}(e)$. The picture is perhaps clearer if interpreted
in this context, in which case, the points $(u,v)$, $(u,e)$, $(e,v)$,
and $(e,e)$ should be replaced by their respective $\sim$-classes.

 Let $S$ be an equidivisible profinite semigroup $S$ which is finitely
 cancelable, and let $w\in S$.
 We endow every ordered subset $Q$ of $\mathfrak {L}(w)$
 with the following labeling $\ev$:
   for a step point $p=(u,v)$ of $\mathfrak {L}(w)$ belonging to $Q$,
   let $\ev(p)=\beu(v)$ (and so
   if $Q=\mathfrak {L}(w)$ then the labeling on step points is the one defining the cluster word
   $\mathfrak {L}_c(w)$);
   for a stationary point~$p$ of $\mathfrak {L}(w)$ belonging to $Q$,
   let~$\ev(p)=J_p$. The resulting labeled ordered set is denoted $Q_\ev$.
   In the next result, $P_\ev+Q_\ev$ denotes
   the labeled ordered set with underlying
   ordered set $P+Q$ and labeling whose restriction to $P$
   and $Q$ is respectively the labeling of $P_\ev$ and $Q_\ev$.
   The symbol $\cong$ stands for isomorphism of labeled ordered sets.  

\begin{Prop}
  \label{p:cut}
  Let $S$ be an equidivisible profinite semigroup $S$ which is finitely
cancelable. Take $w\in S$. Let $u,v\in S$ be such that $w=uv$.
If $e$ is the unique idempotent of $J_{(u,v)/{\sim}}$ stabilizing $(u,v)$, then
\begin{equation*}
  \mathfrak {L}(w)_\ev\cong {[}(1,u)/{\sim}, (u,e)/{\sim} {[}_\ev+
  {[}(e,v)/{\sim},(v,1)/{\sim}{]}_\ev.
\end{equation*}
 In particular, if $(u,v)$ is a step point, then
 $\mathfrak {L}(w)_{\ev}\cong(\mathfrak {L}(u)\setminus\{(u,1)\})_{\ev}+\mathfrak {L}(v)_{\ev}$.
\end{Prop}

\begin{proof}
  Consider the quasi-isomorphism $\lambda_{(u,v)}$ and respective
  isomorphism $\tilde\lambda_{(u,v)}$ as in
  Proposition~\ref{p:cut-at-type-1}. Then, the pair $(x,y)\in
  \mathfrak {F}(u)$ is a step point of the interval
  ${[}(1,u),(u,e){[}$ if and only if its image
  $\lambda_{(u,v)}(x,y)=(x,yv)$ is a step point of
  ${[}(1,w),(u,v){[}$, and $\beu(y)=\beu(yv)$. This fact, together
  with $\tilde \lambda_{(u,v)}$ being $J$-preserving, enables us to
  conclude that
  \begin{equation*}
  {[}(1,u)/{\sim},(u,e)/{\sim}{[}_{\ev}\cong
  {[}(1,w)/{\sim},(u,v)/{\sim}{[}_{\ev}.    
  \end{equation*}
  Similarly, we obtain
  \begin{equation*}
  {[}(e,v)/{\sim},(v,1)/{\sim}{]}_{\ev}\cong {[}(u,v)/{\sim},(w,1)/{\sim}{]}_{\ev}.    
  \end{equation*}
   Since we clearly have
  \begin{equation*}
    \mathfrak
    {L}(w)_{\ev}={[}(1,w)/{\sim},(u,v)/{\sim}{)[}_{\ev}+{[}(u,v)/{\sim},(w,1)/{\sim}{]}_{\ev},
  \end{equation*}
  this concludes the proof.  
\end{proof}

\section{The image of the representation
  \texorpdfstring{$w\mapsto\mathfrak{L}_c(w)$}{w|->Lc(w)} in the
  aperiodic case}

Consider a cluster word $(L,\ell)$ over $A$.
Let $\varphi\colon A\to S$ be a generating mapping
of a semigroup $S$.
Let $s\in S$ and $g\colon\step(L)\to \mathfrak{F}(s)$ be such that
$(L,\ell)$ is $g$-recognized by $(\varphi,s)$. We say that
$g$ is a \emph{$(\varphi,s)$-recognizer} of~$(L,\ell)$.

\begin{Lemma}\label{l:onto-homomorphism-preserves-recognition}
  Let $\varphi\colon A\to S$ be a generating mapping of a finite semigroup~$S$,
  and let $\pi\colon S\to T$ be an onto homomorphism of semigroups.
  Suppose that $(L,\ell)$ is recognized by $(\varphi,s)$.
  Then $(L,\ell)$ is recognized by $(\pi\circ\varphi,\pi(s))$.
\end{Lemma}

\begin{proof}
  Let $g\colon\step(L)\to\mathfrak{F}(s)$
  be a $(\varphi,s)$-recognizer of $(L,\ell)$.
  For each $p\in \step(L)$, let $g(p)=(u_p,v_p)$.
  Consider the mapping $h\colon\step(L)\to\mathfrak{F}(\pi(s))$
  defined by $h(p)=(\pi(u_p),\pi(v_p))$. We claim
  that $(L,\ell)$ is $h$-recognized by $(\pi\circ\varphi,\pi(s))$.
  The conditions \ref{item:run-1}-\ref{item:run-3}
  in Definition~\ref{def:successful-run}
  for $h$-recognition by  $(\pi\circ\varphi,\pi(s))$ are
  clearly satisfied.
  It remains to show that condition~\ref{item:run-4} holds.
  Let $p$ be a stationary point of $(L,\ell)$.
  Take an element $q$ of $\mathfrak{F}(\pi(s))$ such that
  $h^{-1}(q)$ is left cofinal at $p$.
  Consider the set
  \begin{equation*}
    X=\{(u,v)\in \mathfrak{F}(s):(\pi(u),\pi(v))=q\}.
  \end{equation*}
  Then we have $h^{-1}(q)=g^{-1}(X)=\bigcup_{x\in X} g^{-1}(x)$.
  Since $h^{-1}(q)$ is left cofinal at $p$ and
  $X$ is finite, there is at least one element $x_0$ of
  $X$ such that $g^{-1}(x_0)$ is left cofinal
  at $p$. But then $g^{-1}(x_0)$ is also right cofinal at $p$, because
  $(L,\ell)$ is $g$-recognized by $(\varphi,s)$.
  Therefore, $h^{-1}(q)$ is right cofinal at $p$.
  Symmetrically, if
  $h^{-1}(q)$ is right cofinal at $p$,
  then it is left cofinal at $p$.
  This concludes our proof.  
\end{proof}

For a cluster word $(L,\ell)$ over $A$, if $p$ and $q$ are step points of $L$ such that $p\leq q$, then $([p,q],\ell)$ is the
cluster word obtained from $(L,\ell)$ by restricting $\ell$ to ${[}p,q{[}$, and letting
$\ell(q)=1$.

We wish to study cluster words $(L,\ell)$ satisfying the following conditions:

\begin{enumerate}
  [label=(W.\arabic*),series=axioms]
\item For every
  finite aperiodic unambiguous $A$-generated semigroup $S$,
  and every generating mapping $\varphi\colon A\to S$, there is
  a unique $s\in S$ such that $(L,\ell)$
  is recognized by $(\varphi,s)$.\label{item:worthy-1}
\item If $p$ and $q$ are step points of $L$
  such that $p<q$,
  then~$([p,q],\ell)$ satisfies~\ref{item:worthy-1}.\label{item:worthy-2}
\item Consider an arbitrary finite aperiodic unambiguous
  $A$-generated
  semigroup $S$
  and a generating mapping $\varphi\colon A\to S$.
  Let $s$ be such that $(L,\ell)$
  if recognized by $(\varphi,s)$.
  Take a $(\varphi,s)$-recognizer $g$.
  Suppose $p$ and $q$ are step points such that $p<q$.
  If $t\in S$ is such that $(\varphi,t)$ recognizes $([p,q],\ell)$,
  then $g(p)\xrightarrow{t}g(q)$ is an edge of~$\Cl T(s)$.  
  \label{item:worthy-3}

\end{enumerate}

\begin{Remark}\label{r:g-is-determined}
  In the setting of condition~\ref{item:worthy-3}, there is only one such $(\varphi,s)$-recognizer, assuming that \ref{item:worthy-1} and \ref{item:worthy-2} also hold.
Indeed, if $g$ is a $(\varphi,s)$-recognizer,
and $p$ is a step point of $L$, and if $s_1$ and $s_2$
   are (the unique) elements
   of $S$ such that $([\min L,p],\ell)$ and
   $([p,\max L],\ell)$ are respectively
   recognized by $(\varphi,s_1)$ and $(\varphi,s_2)$,
   then $g(\min L)\xrightarrow{s_1} g(p)$
   and
   $g(p)\xrightarrow{s_2} g(\max L)$
   are edges of $\Cl T(s)$, and so $g(p)=(s_1,s_2)$.
\end{Remark}

 Finally we consider a fourth condition, assuming 
  \ref{item:worthy-1}-\ref{item:worthy-3} hold:
   
\begin{enumerate}[resume*=axioms]
  
\item For every step point $p$ of $L$,
  there is a finite aperiodic unambiguous semigroup $S$
  and a generating mapping $\varphi\colon A\to S$
  such that, for the unique $(\varphi,s)$-recognizer $g$ of $(L,\ell)$,
  there are no elements of $S$
  that stabilize $g(p)$ in $\Cl T(S)$.\label{item:worthy-4}  
\end{enumerate}

A cluster word satisfying conditions
\ref{item:worthy-1}-\ref{item:worthy-4}
is called a \emph{worthy} cluster word.

\begin{Thm}\label{t:the-image-are-the-worthys}
  A cluster word $(L,\ell)$ over $A$ is
  isomorphic to a cluster word of
  the form $\mathfrak{L}_c(w)$, $w\in\Om AA$,
  if and only if it is a worthy cluster word.
\end{Thm}

\begin{proof}
  Let $w\in \Om AA$. By Theorem~\ref{t:g-is-a-category-homomorphism},
  the cluster word $\mathfrak{L}_c(w)$ satisfies
  condition \ref{item:worthy-1}.

  Take two step points $p$ and $q$ of $\mathfrak{L}_c(w)$ such that $p<q$.
  Let $t\in\Om AA$ be the unique transition from $p$ to $q$.
  Applying twice Proposition~\ref{p:cut-at-type-1}, we conclude that
  $([p,q],\ell)$ is isomorphic with $\mathfrak{L}_c(t)$. 
  Therefore, by Theorem~\ref{t:g-is-a-category-homomorphism},
  $([p,q],\ell)$ satisfies condition~\ref{item:worthy-1}.

  By the previous paragraph and by Theorem~\ref{t:g-is-a-category-homomorphism},  the cluster word $\mathfrak{L}_c(w)$ satisfies
  condition \ref{item:worthy-3}.
  
  Suppose condition~\ref{item:worthy-4} does not hold for
  $\mathfrak{L}_c(w)$. Then, there is a step point $p$ such that, for
  every finite aperiodic unambiguous $A$-generated semigroup $S$ and
  every generating mapping $\varphi\colon A\to S$, the vertex $g(p)$
  of~$\mathcal T(s)$ is stabilized by some element of $S$. Let
  $p=(u,v)$. We then have $g(p)=(\varphi_{\pv A}(u),\varphi_{\pv
    A}(v))$ (cf.~Theorem~\ref{t:g-is-a-category-homomorphism}). By a
  standard compactness argument, this implies that $(u,v)$ is
  stabilized by some element of $\Om AA$. In view of
  Proposition~\ref{p:shifting-allowed-implies-type-2}, this is impossible
  since $p$ is a step point.

  Conversely, suppose that $(L,\ell)$ is a worthy cluster word over $A$.
  Let $(\pi_{i})_{i\in I}$ be an inverse system of continuous 
  homomorphisms $\pi_i\colon \Om AA\to S_i$ onto finite aperiodic
  unambiguous $A$-generated semigroups, with connecting homomorphisms
  $\pi_{j,i}\colon S_j\to S_i$, such that $\Om AA=\varprojlim_{i\in I} S_i$.
  According to condition~\ref{item:worthy-1},
  for each $i\in I$, we may consider the unique element $s_i$ of $S_i$ such that
  $(L,\ell)$ is recognized by $(\pi_i,s_i)$.
  Applying Lemma~\ref{l:onto-homomorphism-preserves-recognition},
  we then conclude that $\pi_{j,i}(s_j)=s_i$, whenever $i,j\in I$ are such that
  $i\leq j$.
  Hence, we may consider the unique element $w_L$ of $\Om AA$
  such that
  \begin{equation}\label{eq:definition-of-wl}
  \pi_i(w_L)=s_i,  
  \end{equation}
  for every $i\in I$.
  
    Consider the mapping
  $\lambda\colon\step(L)\to \step(\mathfrak{L}(w_L))$ defined by
  \begin{equation*}
  \lambda(p)=(w_{[\min L,p]},w_{[p,\max L]}).  
  \end{equation*}
  Note that condition~\ref{item:worthy-2}
  ensures that $w_{[\min L,p]}$ and $w_{[p,\max L]}$ are well defined.
  We claim that $\lambda$ is an isomorphism
  between the cluster words $(L,\ell)$ and $\mathfrak{L}_c(w_L)$.
  In the process of proving this
  we show that $(w_{[\min L,p]},w_{[p,\max L]})$
  is indeed a step point of $\mathfrak{F}(w_L)$
  (and thus of $\mathfrak{L}(w_L)$).
  
  We begin by observing
  that formula~\eqref{eq:definition-of-wl} generalizes to
  every generating mapping $\varphi\colon A\to S$ of
  a finite aperiodic unambiguous $A$-generated semigroup.
  Indeed, take $s\in S$ such that $(L,\ell)$ is recognized by $(\varphi,s)$.
  There is some $i\in I$ for which
  there is an onto homomorphism
  $\rho\colon S_i\to S$ satisfying $\varphi_{\pv A}=\rho\circ\pi_i$.
  By Lemma~\ref{l:onto-homomorphism-preserves-recognition}, we know that
  $\rho(s_i)=s$. Therefore, we have
  \begin{equation}\label{eq:second-formula-for-wl}
      \varphi_{\pv A}(w_L)=s. 
  \end{equation}
  For such a pair $(\varphi,s)$, let $g_{\varphi}\colon \step(L)\to \mathfrak{F}(s)$
  be the unique $(\varphi,s)$-recognizer of~$(L,\ell)$.
  If $p$ is a step point of $L$, then applying formula
  \eqref{eq:second-formula-for-wl}
  to $[\min L,p]$ and to
  $[p,\max L]$, and taking into account
  Remark~\ref{r:g-is-determined}, we conclude that
  \begin{equation}\label{eq:who-is-g-fi}
  g_\varphi(p)=(\varphi_{\pv A}(w_{[\min L,p]}),\varphi_{\pv A}(w_{[p,\max L]})).  
\end{equation}
In particular,
we have
\begin{equation*}
  \varphi_{\pv A}(w_{[\min L,p]}w_{[p,\max L]})=s=\varphi_{\pv A}(w_L).
\end{equation*}
Because $\varphi$ was arbitrarily chosen among generating mappings of finite
aperiodic unambiguous $A$-generated semigroups, this shows
that the pair $\lambda(p)=(w_{[\min L,p]},w_{[p,\max L]})$ indeed belongs to $\mathfrak{F}(w_L)$.
  
  Consider step points $q$ and $r$
  of $L$ such that $q\prec r$. Let $a=\ell(q)$.
  Take a generating mapping $\varphi\colon A\to S$
  of a finite aperiodic unambiguous semigroup~$S$.
  By the definition of $(\varphi,s)$-recognizer,
  we can consider in $\Cl T(s)$ the edge
  $g_\varphi(q)\xrightarrow{\varphi(a)}g_\varphi(r)$.
  In view of formula~\eqref{eq:who-is-g-fi}, applied to $q$ and $r$,
  we then have
  \begin{equation*}
    \varphi_{\pv A}(w_{[\min L,q]}a)=\varphi_{\pv A}(w_{[\min L,r]})\quad
    \text{and}
    \quad
    \varphi_{\pv A}(w_{[q,\max L]})=\varphi_{\pv A}(aw_{[r,\max L]}).
  \end{equation*}
  Since $\varphi$ was arbitrarily chosen among generating mappings of finite
  aperiodic unambiguous $A$-generated semigroups, we conclude that
  \begin{equation*}
  w_{[\min L,q]}a=w_{[\min L,r]}\quad
    \text{and}
    \quad
  w_{[q,\max L]}=aw_{[r,\max L]},  
  \end{equation*}
  that is, $\lambda(q)\xrightarrow{a}\lambda(r)$
  is an edge of $\Cl T(w_L)$.
  By Proposition~\ref{p:covers-in-F},
  we either have $\lambda(q)\sim\lambda(r)$ or $\lambda(q)\prec\lambda(r)$.
  
  If $\lambda(q)\sim\lambda(r)$, then
  there is $z\in a(\Om AA)^I$
  such that
  $\lambda(q)\xrightarrow{z}\lambda(q)$
  is an edge of $\Cl T(w_L)$.
  Therefore, in view of formula~\eqref{eq:who-is-g-fi},
  we conclude that $\varphi(z)$ labels a loop of $\Cl T(s)$
  rooted at $g_\varphi(q)$.
  This contradicts the assumption
  that condition~\ref{item:worthy-4} holds.

  We then conclude that, for step points $q,r$ of $L$, we have
  \begin{equation}\label{eq:succ-is-preserved}
  q\prec r\implies \lambda(q)\prec\lambda(r).  
\end{equation}
  We also showed that $\ell(\lambda(q))=\ell(q)$,
  thus establishing that
  the mapping $\lambda\colon\step(L)\to \step(\mathfrak{L}(w_L))$
  has a well-defined codomain and that it preserves labels.

 Notice that formula~\eqref{eq:who-is-g-fi}
 can now be seen as follows: for
 the $(\varphi,s)$-recognizer
 $g_{w_L,\varphi}\colon\step(\mathfrak{L}(w))\to \mathfrak{F}(s)$
 of $\mathfrak{L}_c(w_L)$,
 as in Proposition~\ref{p:direct-implication-of-theorem-we-search},
 we have
 \begin{equation}\label{eq:relating-recognizers}
   g_\varphi(p)=g_{w_L,\varphi}(\lambda(p)),
 \end{equation}
 for every step point $p$ of $L$.
  
  Let $q$ and $r$ be step points such that $q<r$.  
  Suppose that $\lambda(q)\geq \lambda(r)$.
  Then, for the pair $(\varphi,s)$ considered so far,
  and in view of~\eqref{eq:relating-recognizers},
  we may consider in $\Cl T(s)$
  an edge $g_{\varphi}(r)\xrightarrow{t}g_{\varphi}(q)$ labeled by some
  $t\in S^I$.
On the other hand, according to condition~\ref{item:worthy-3},
there is in $\Cl T(s)$ an edge
$g_{\varphi}(q)\xrightarrow{z}g_{\varphi}(r)$ labeled by some $z\in S$.
   It follows that there is a loop in $\Cl T(s)$
   at $g_\varphi(q)$
   labeled by $zt\in S$.
   This contradicts~\ref{item:worthy-4}. Hence, we have
   $\lambda(q)<\lambda(r)$.
   
   It remains to show that $\lambda$ is onto.
   Let $(u,v)$ be a step point of $\mathfrak L(w_L)$.
   Consider the set 
   \begin{equation*}
     X=\{q\in\step(L):\lambda(q)\leq (u,v)\}.
   \end{equation*}
   Notice that $X$ is nonempty: indeed, one clearly has $\min L\in X$.

   We claim that $p=\sup X$ is a step point.
   Suppose not.
   Let $\varphi\colon A\to S$
   be the generating mapping of a finite aperiodic unambiguous $A$-generated
   semigroup. 
   Since $\Im g_{\varphi}$ is finite
   and $\{q\in\step(L):q>p\}$ is right cofinal at~$p$, there is
   $(s_1,s_2)\in \Im g_{\varphi}$
   such that
   \begin{equation*}
      R=\{q\in\step(L):q>p\text{ and }g_{\varphi}(q)=(s_1,s_2)\}
   \end{equation*}
   is right cofinal at $p$.
   In particular, $g_\varphi^{-1}(s_1,s_2)$ is right cofinal at $p$.
   Taking into account condition \ref{item:run-4}
   in Definition~\ref{def:successful-run},
   we know that $g_\varphi^{-1}(s_1,s_2)$ is also left cofinal at $p$.
   Therefore, there is a step point $q$ such that $q<p$ and
   $g_{\varphi}(q)=(s_1,s_2)$. Since $p=\sup X$, there is a step point
   $q'$ such that $q<q'<p$ and $\lambda(q')\leq (u,v)$.
   We have already shown that $\lambda$ is injective and respects the order,
   so we actually have $\lambda(q)<(u,v)$.
   Let $r$ be an element of the nonempty set $R$. Since $r>p$,
   we have $(u,v)<\lambda (r)$. Let $t_1$ and $t_2$ be (the unique)
   transitions from $\lambda(q)$ to $(u,v)$ and from $(u,v)$ to
   $\lambda(r)$, respectively. Then, in $\Cl T(s)$, we have the
   following edges
   \begin{equation}\label{eq:edges-forming-a-loop}
     g_{w_L,\varphi}(\lambda(q))\xrightarrow{\varphi_{\pv A}(t_1)}
     g_{w_L,\varphi}(u,v)\xrightarrow{\varphi_{\pv A}(t_2)}g_{w_L,\varphi}(\lambda(r)).  
   \end{equation}
   But we have $g_{w_L,\varphi}(\lambda(q))=g_\varphi(q)=(s_1,s_2)=g_\varphi(r)
    =g_{w_L,\varphi}(\lambda(r))$.
   Hence, we can multiply the second edge in~\eqref{eq:edges-forming-a-loop}
   with the first edge, obtaining a loop at
   $g_{w_L,\varphi}(u,v)=(\varphi_{\pv A}(u),\varphi_{\pv A}(v))$
   labeled by $\varphi_{\pv A}(t_2t_1)\in S$,
   leading to a contradiction, since $\mathfrak L_c(w_L)$ satisfies \ref{item:worthy-4}. 
   This establishes the claim that $p$ is a step point, thus $p\in X$.

   Suppose that $\lambda(p)<(u,v)$.
   Let $p'$ be the step point such that $p\prec p'$. Then,
   applying~\eqref{eq:succ-is-preserved},
   we get $\lambda(p)\prec \lambda (p')$. 
   Since $\lambda(p)< (u,v)$, we
   obtain $\lambda(p')\leq (u,v)$, and so $p'\in X$.
   But then $p'\leq \sup X=p$, a contradiction with $p<p'$.
   As $p\in X$, to avoid the contradiction, we must have $\lambda(p)=(u,v)$.
   This concludes the proof that $\lambda$ is onto. 
\end{proof}

It would also be interesting to characterize the worthy clustered
linear orders that arise as the images $\mathfrak{L}_c(w)$ of
$\omega$-words $w$. We leave this as an open problem.

\section{On the cardinality of the set of stationary points}
\label{sec:card-prot-l_2w}

Let $\pv V$ be an equidivisible pseudovariety of semigroups not contained
in~$\pv {CS}$. Then $\pv V$ is finitely cancelable
(cf.~Proposition~\ref{p:equid-conca-are-finitely-cancelable}),
and so, by Theorem~\ref{t:cluster}, for a finite alphabet $A$ and for $w\in\Om AV$,
the set $\step(\mathfrak {L}(w))$ of step points is
the set of isolated points of
$\mathfrak {L}(w)$, with respect to the order topology.
Therefore, $\step(\mathfrak {L}(w))$ is at most countable
by Corollary~\ref{c:Ls-is-compact}.
The aim of this section
is to show that when $A$ has at least two elements,
there are elements $w$ in $\Om AV$
for which the set  $\stat(\mathfrak {L}(w))$ of stationary points
has cardinal $2^{\aleph_0}$. This will be done
using some tools originated from
symbolic dynamics, following an approach that
has been successfully used in recent years
to elucidate structural aspects of
relatively free profinite semigroups~\cite{Almeida:2005c,Almeida&Volkov:2006,Almeida&ACosta:2007a,ACosta&Steinberg:2011,Almeida&ACosta:2013}.

 \subsection{Subshifts}
 Consider a finite alphabet $A$,
 and endow $A^\mathbb{Z}$
 with the product topology, where $A$ is endowed with
 the discrete topology.
 The \emph{shift map} of $A^\mathbb{Z}$ is the
 homeomorphism $\sigma\colon A^\mathbb{Z}\to A^\mathbb{Z}$,
 defined by $\sigma((x_i)_{i\in \mathbb Z})=(x_{i+1})_{i\in \mathbb Z}$.
  A \emph{symbolic dynamical system of  $A^\mathbb{Z}$},
  also called \emph{subshift of $A^\mathbb{Z}$},
  is a nonempty closed subset $\Cl X$ of~$A^\mathbb{Z}$
  such that $\sigma(\Cl X)=\Cl X$.
  The books~\cite{Lind&Marcus:1996,Kitchens:1998}
  are good references on symbolic dynamical systems.

   We say that a subset $L$ of a semigroup $S$ is
\begin{itemize}
\item \emph{factorial} if it is closed under taking factors;
\item \emph{prolongable} if, for every $s\in L$, there are $t,u\in
  S$ such that $ts,su\in L$;
\item \emph{irreducible} if, for all $s,t\in L$, there is $u\in S$
  such that $sut\in L$.
\end{itemize}

If $\Cl X$ is a subshift of $A^\mathbb{Z}$, then $L(\Cl X)$ denotes
the language of the words of $A^+$ of the form
$x_kx_{k+1}\ldots x_{k+n}$, where
$k\in\mathbb Z$, $n\geq 0$ and $(x_i)_{i\in \mathbb Z}\in \Cl X$.
The set $L(\Cl X)$
is a factorial and prolongable language of $A^+$,
and in fact all nonempty factorial and prolongable languages of $A^+$ are of
  this form; moreover, $\Cl Y\subseteq \Cl X$
  if and only if $L(\Cl Y)\subseteq L(\Cl X)$,
  whenever $\Cl X$
  and $\Cl Y$ are subshifts of
  $A^{\mathbb{Z}}$~\cite[Proposition 1.3.4]{Lind&Marcus:1996}.
  Finally, \Cl X is said to be \emph{irreducible} if  $L(\Cl X)$
 is an irreducible subset of $A^+$.

   If $\Cl X$ is a subshift of $A^{\mathbb Z}$ then
   the sequence $(\frac{1}{n}\log_2 |L(\Cl X)\cap A^n|)_n$
   converges to its infimum, which is called the
   \emph{entropy of $\Cl X$} and denoted $h(\Cl X)$~\cite{Lind&Marcus:1996}.
   Note that $\Cl X\subseteq \Cl Y$     implies $h(\Cl X)\leq h(\Cl Y)$,
   whenever $\Cl X$ and $\Cl Y$ are subshifts.

\begin{Remark}\label{r:entropy-full-shift}
  If $\Cl X$ is a subshift of $A^{\mathbb Z}$
  then $h(\Cl X)\leq \log_2|A|=h(A^{\mathbb Z})$.
  Moreover, from the fact that
  $(\frac{1}{n}\log_2 |L(\Cl X)\cap A^n|)_n$ converges to its infimum
  one easily deduces that the subshift $\Cl X$ of $A^{\mathbb Z}$ satisfies
  $h(\Cl X)=\log_2|A|$ if and only if~$\Cl X=A^{\mathbb Z}$
  (this a special case
  of \cite[Corollary 4.4.9]{Lind&Marcus:1996}).
\end{Remark}

\subsection{A special \texorpdfstring{$\Cl J$}{J}-class}    

Consider a subshift $\Cl X$ of $A^{\mathbb Z}$, and suppose that \pv V
is a pseudovariety containing $\pv {LSl}$. Let $M_\pv V(\Cl X)$ be the
set of pseudowords $w\in\Om AV$ such that all finite factors of~$w$
belong to~$L(\Cl X)$. The set $M_\pv V(\Cl X)$ is a factorial subset of~\Om
AV. Because, as it is well known, the languages of the form $A^\ast
uA^\ast$, with $u\in A^+$, are $\pv {LSl}$-recognizable, the
hypothesis that $\pv V$ contains $\pv {LSl}$ ensures that $M_\pv V(\Cl
X)$ is a closed subset of $\Om AV$ (cf.~\cite{ACosta:2007t}).

    \begin{Lemma}\label{l:bridge-between-j-minimum-elements-LSl-case}
      Let $\Cl X$ be an irreducible subshift of
      $A^{\mathbb Z}$.
      Consider a pseudovariety \pv V containing $\pv {LSl}$.
      For every $u,v\in M_\pv V(\Cl X)$ there is $w\in \Om AV$,
      depending only on the finite suffixes of $u$ and
      on the finite prefixes of $v$,
      such that $uwv\in M_\pv V(\Cl X)$.      
    \end{Lemma}

    \begin{proof}
      If $\pv V$ is a pseudovariety containing $\pv D$ and its dual, then
      every infinite element of $\Om AV$ has a unique prefix (suffix)
      in $A^+$ with length $n$, for every $n\geq 1$.
      Let $s_n$ be the suffix
      of length $n$ of $u$
      and let $p_n$ be the prefix of length $n$ of $v$.
      Since $\Cl X$ is irreducible, 
      for each $n\in \mathbb N$, there is
      $w_n\in L(\Cl X)$ such that
      $s_nw_np_n\in L(\Cl X)$.
      Let $w$ be an accumulation point
      of $(w_n)_n$.
      Then $w$ has the desired property.
    \end{proof}

\begin{Prop}[{\cite[Proposition 3.6]{ACosta&Steinberg:2011}}]
  \label{p:closed-factorial-prolongable-irreducible-in-compact}
  Let $S$ be a compact semigroup and let $X\subseteq S$. Then $X$ is a
  closed, factorial, irreducible subset of $S$ if and only if $X$ consists of
  all factors of some regular element of~$S$.
\end{Prop}

By Lemma~\ref{l:bridge-between-j-minimum-elements-LSl-case}
and
Proposition~\ref{p:closed-factorial-prolongable-irreducible-in-compact},
if $\Cl X$ is an irreducible subshift, then
there is a unique regular $\Cl J$-class $J_\pv V(\Cl X)$
such that the elements of $M_\pv V(\Cl X)$ are the factors
of elements of $J_\pv V(\Cl X)$.

\begin{Remark}\label{r:j-order-shifts}
It also follows from
Proposition~\ref{p:closed-factorial-prolongable-irreducible-in-compact}
that $J_{\pv V}(\Cl X)\leq_{\Cl J} J_{\pv V}(\Cl Y)$
if and only if $M_\pv V(\Cl Y)\subseteq M_\pv V(\Cl X)$.
Since we clearly have $M_\pv V(\Cl Y)\subseteq M_\pv V(\Cl X)$
if and only if $L(\Cl Y)\subseteq L(\Cl X)$,
we conclude that
    \begin{equation}
      \label{eq:subshifts-j-classes}
    \Cl Y\subseteq \Cl X\iff
    J_\pv V(\Cl X)\leq_{\mathrel{\Cl J}} J_\pv V(\Cl Y).
    \end{equation}  
\end{Remark}

 \begin{Remark}\label{r:the-case-of-minimum-ideal}
   If $\Cl X=A^{\mathbb Z}$ then $M_\pv V(\Cl X)=\Om AV$, and so
   $J_{\pv V}(\Cl X)$ is the minimum ideal of $\Om AV$.
 \end{Remark}

 \subsection{Uncountable \texorpdfstring{$<_{\mathcal R}$}{<R}-chains
   and uncountable sets of stationary points}

We use the standard notation $\lceil \alpha\rceil$ for the least
integer greater than or equal to the real number $\alpha$.

\begin{Thm}\label{t:theorem-beta-shifts-that-provides-examples}
  There is a family $(\mathcal S_\beta)_{\beta\in {]}1,+\infty{[}}$ of
  symbolic dynamical systems, parameterized by the set of real numbers
  greater than one, such that:
  \begin{enumerate}
  \item $\cal S_\beta$ is an irreducible subshift of
    $\{0,\ldots,\lceil \beta\rceil-1\}^{\mathbb Z}$;
    \label{item:theorem-beta-shifts-that-provides-examples-1}
  \item $h(\cal
    S_\beta)=\log_2\beta$;\label{item:theorem-beta-shifts-that-provides-examples-2}
  \item for every $\alpha,\beta\in {]}1,+\infty{[}$, we have
    $\alpha<\beta$ if and only if $\cal S_{\alpha}\subsetneq\cal
    S_{\beta}$.\label{item:theorem-beta-shifts-that-provides-examples-3}
  \end{enumerate}      
\end{Thm}

A concrete family of symbolic dynamical systems satisfying the
conditions of
Theorem~\ref{t:theorem-beta-shifts-that-provides-examples}
is the family of \emph{$\beta$-shifts}. A comprehensive exposition
about this family can be found
in~\cite[Chapter 7]{Lothaire:2001}
and \cite{Frougny:2000}.
That these subshifts are irreducible follows from them being \emph{coded}~\cite{Bertrand-Mathis:1986} ---
a subshift $\Cl X$ of $A^{\mathbb Z}$ is coded if there is a prefix
code $Y$ contained in 
$A^+$ such that $L(\cal X)$ is the set of factors of elements of $Y^+$.
The entropy of $\beta$-shifts was computed
in~\cite{Renyi:1957,Parry:1960}.
The fact that this class fits into
Property~\ref{item:theorem-beta-shifts-that-provides-examples-3}
of Theorem~\ref{t:theorem-beta-shifts-that-provides-examples}
appears at the beginning of~\cite[Section 4]{Ito&Takahashi:1974}
(actually, only the implication
$\alpha\leq \beta\Rightarrow\cal S_{\alpha}\subseteq\cal S_{\beta}$
is explicit there, but
from Property~\ref{item:theorem-beta-shifts-that-provides-examples-2}
one gets $\cal S_{\alpha}\subsetneq\cal S_{\beta}\Rightarrow\alpha <\beta$).

As usual, the notation $<_{\Cl J}$ stands for the irreflexive relation
originated by $\leq_{\Cl J}$, and similarly for $<_{\Cl R}$ and $<_{\Cl L}$.
From Theorem~\ref{t:theorem-beta-shifts-that-provides-examples}
   and equivalence~\eqref{eq:subshifts-j-classes}
   in Remark~\ref{r:j-order-shifts}
   one immediately deduces the existence of a $<_{\Cl J}$-chain
   in $\Om AV$ formed by $2^{\aleph_0}$ regular elements,
whenever~$\pv V$
contains $\pv {LSl}$
and $A$ has at least two letters.
The next theorem gives a refinement of this, as it shows in particular
the existence in $\Om AV$ of a $<_{\Cl R}$-chain
formed by $2^{\aleph_0}$ regular elements.
We remark that
in~\cite{Costa:2001a} an example is given of a $<_{\Cl R}$-chain
of $2^{\aleph_0}$ non-regular elements in $\Om A{LSl}$, when $|A|>1$.

\begin{Thm}\label{t:pseudword-minimum-ideal-r-chain-limit}
  Let \pv V be a finitely cancelable pseudovariety of semigroups
  containing $\pv {LSl}$ and let $(\mathcal S_\beta)_{\beta\in
    {]}1,+\infty{[}}$ be a family of subshifts as in
  Theorem~\ref{t:theorem-beta-shifts-that-provides-examples}. Fix an
  integer $n>1$ and let $A$ be the alphabet $\{0,\ldots,n-1\}$. There
  is a family $(w^{(\beta)})_{\beta\in {]}1,n{]}}$ of pseudowords of
  $\Om {A}V$ satisfying the following conditions:
  \begin{enumerate}
  \item\label{item:pseudword-minimum-ideal-r-chain-limit-1}
    $w^{(\beta)}\in J_\pv V(\Cl S_{\beta})\subseteq \Om {A}V$,
    for every
    $\beta\in {]}1,n{]}$;
  \item\label{item:pseudword-minimum-ideal-r-chain-limit-2}
    $\alpha<\beta \Leftrightarrow w^{(\beta)}<_{\Cl R}w^{(\alpha)}$,
        for every
    $\alpha,\beta\in {]}1,n{]}$;
  \item\label{item:pseudword-minimum-ideal-r-chain-limit-2-e-meio}
    $w^{(n)}$ is an element of the minimum ideal of $\Om AV$;
  \item\label{item:pseudword-minimum-ideal-r-chain-limit-3} there is a subnet of
    $(w^{(\beta)})_{\beta\in {{]}1,n{[}}}$ converging to $w^{(n)}$, where
    ${{]}1,n{[}}$ is endowed with the usual order;
  \item\label{item:pseudword-minimum-ideal-r-chain-limit-4} for each
    $\beta\in {]}1,n{]}$, there are $v^{(\beta)}$,
    $f^{(\beta)}$
    such that $(w^{(\beta)},v^{(\beta)})$ is a stationary point of
    $\mathfrak F(w^{(n)})$, 
    and, for $q_\beta=(w^{(\beta)},v^{(\beta)})/{\sim}$,
    the pseudoword 
    $f^{(\beta)}$ is an idempotent in
    $J_{q_\beta}$ stabilizing $(w^{(\beta)},v^{(\beta)})$
    and satisfying $f^{(\beta)}\mathrel{\Cl L} w^{(\beta)}$;
  \item\label{item:pseudword-minimum-ideal-r-chain-limit-5}
    we have  $\alpha <\beta \Rightarrow q_{\alpha}<q_{\beta}$,
    and if  moreover $\Om AV$ is equidivisible,
    then the equivalence $\alpha <\beta \Leftrightarrow
    q_{\alpha}<q_{\beta}$ holds, for every $\alpha,\beta\in {]}1,n{]}$.
  \end{enumerate}
\end{Thm}

The proof of Theorem~\ref{t:pseudword-minimum-ideal-r-chain-limit}
will
be done in several steps.
But first
we highlight the following corollary, which is our main motivation for
the theorem.

\begin{Cor}
  \label{c:continuum-many-type-2-elements}
  Let \pv V be a finitely cancelable pseudovariety of semigroups
  containing $\pv {LSl}$ and
  let $A$ be a finite alphabet with at least two elements.
  Then there are pseudowords $w$ in the minimum ideal of $\Om AV$
  such that $\stat(\mathfrak L(w))$ has~$2^{\aleph_0}$ elements.\qed
\end{Cor}

Note also that Theorem~\ref{t:pseudword-minimum-ideal-r-chain-limit}
gives an example of a pseudoword in the minimum ideal of $\Om AV$
whose set of stationary points contains a subset with the same order
type as the set of real numbers. In contrast, the following example
exhibits a pseudoword also in the minimum ideal of $\Om AV$ with only
one stationary point.

\begin{Example}
  Let $u_1,u_2,u_3,\ldots$ be an enumeration of the elements
  of~$A^+$, and let $\pv V$ be a pseudovariety
  containing $\pv {LSl}$ such that $\Om AV$ is equidivisible.
  For each $k\geq 1$, consider in $\Om AV$ an accumulation
  point~$v_k$
  of the sequence $(u_ku_{k+1}\cdots u_{n-1}u_n)_{n\geq k}$
  and an accumulation point $w_k$ of the sequence
  $(u_nu_{n-1}\cdots u_{k+1}u_k)_{n\geq k}$.
  As every element of $A^+$ is a factor of $v_k$ and $w_k$,
  we know that $v_k$ and $w_k$ belong to the minimum ideal $K_A$ of $\Om AV$.
  Therefore, if $p$ and $q$ are respectively the first
  and last stationary point of $\mathfrak{L}(v_1w_1)$, then $J_p=J_q=K_A$ by
  Theorem~\ref{t:J-class-type-2},
  and so $p=q$ by Lemma~\ref{l:absorption}.
\end{Example}

\subsection{About the proof of Theorem~\ref{t:pseudword-minimum-ideal-r-chain-limit}}

Let $S$ be a compact semigroup and $I$ an ordered set. Suppose that
$F=(F_i)_{i\in I}$ is a nonempty family of compact subsets of $S$.
Denote by $R_F$ the set of partial functions $f$ from $I$ to $\bigcup
F$ such that $f(i)\in F_i$ for all $i\in\Dom f$, and such that $i\leq
j \Rightarrow f(i)\leq_{\Cl R} f(j)$ whenever $i,j\in \Dom f$. We
endow $R_F$ with the partial order $\leq$ defined by
\begin{equation*}
  f\leq g\iff 
  (f=g\lor
  \Dom f\subsetneq \Dom g).
\end{equation*}
    
\begin{Lemma}\label{l:maximal-r-chain-general-situation}
  The ordered set $R_F$ has a maximal element.
\end{Lemma}

\begin{proof}
  Let $C$ be a chain of elements of $R_F$. We want to show that $C$
  has an upper bound in~$R_F$. For each $f\in R_F$, let $f'$ be an
  element of $\prod_{i\in I} F_i$ whose restriction to $\Dom f$ equals
  $f$. Since $\prod_{i\in I}F_i$ is compact, the net $(f')_{f\in C}$
  has a subnet converging to some $\varphi\in \prod_{i\in I}F_i$. For
  achieving our goal, we may as well assume that $(f')_{f\in C}$
  converges. Let us fix an element $f_0$ of $C$, and take $i,j\in\Dom
  f_0$ such that $i\leq j$. For all $f\in C$ such that $f_0\leq f$,
  one has $f(i)\leq_{\Cl R} f(j)$. As the net $(f')_{f\in C\land
    f_0\leq f}$ converges to $\varphi$ and $\leq_{\Cl R}$ is a closed
  relation, we deduce that $\varphi(i)\leq_{\Cl R}\varphi(j)$.
  Moreover, since $F_i$ is closed, we also have $\varphi(k)\in F_k$
  for all $k\in\Dom f_0$. As $f_0$ was chosen arbitrarily from $C$, we
  conclude that the restriction of $\varphi$ to $\bigcup_{f\in C}\Dom
  f$ belongs to $R_F$ and is an upper bound for $C$. Hence, by
  Zorn's Lemma, $R_F$ has a maximal element.
\end{proof}
    
\begin{Prop}
  \label{p:continuum-chain-of-R-classes}
  For the relation $\supseteq$, let $\Cl C$ be a nonempty chain of
  irreducible subshifts of \z A. Consider a pseudovariety of
  semigroups \pv V that contains \pv {LSl}. Let $\mathscr J$ be the family of
  \Cl J-classes $(J_\pv V(\Cl X))_{\Cl X\in \Cl C}$. Then there is an
  element of $R_{\mathscr J}$ with domain $\Cl C$.
\end{Prop}

\begin{proof}
  By Lemma~\ref{l:maximal-r-chain-general-situation}, we know there is
  in $R_{\mathscr J}$ a maximal element $f$. We claim that $\Dom f=\Cl C$.
  Suppose this is false. Let $\Cl Z\in \Cl C\setminus \Dom f$.
  Supposing that $I=\{\Cl X\in\Dom f: \Cl X\subseteq \Cl Z\}$ is
  nonempty, let $u$ be an accumulation point of the net $(f(\Cl
  X))_{\Cl X\in (I,\subseteq)}$; in case $I=\emptyset$, we let $u$ be any
  element of $J_\pv V(\Cl Z)$. Since $\Cl X\subseteq \Cl Z$ implies
  $M_\pv V(\Cl X)\subseteq M_\pv V(\Cl Z)$, we have $f(\Cl X)\in M_\pv
  V(\Cl Z)$ for all $\Cl X\in I$. And since $M_\pv V(\Cl Z)$ is
  closed, we conclude that $u\in M_\pv V(\Cl Z)$. Moreover, fixed $\Cl
  X\in I$, then, as $f\in R_{\mathscr J}$, we have $f(\Cl Y) \leq_{\Cl R}f(\Cl
  X)$ for all $\Cl Y\in I$ such that $\Cl X\subseteq \Cl Y$, whence
  \begin{equation}\label{eq:upper-chain-of-r-classes}
    \Cl X\in I\implies u\leq_{\Cl R} f(\Cl X).
  \end{equation}
  Let $v\in J_\pv V(\Cl Z) $. By the irreducibility of $M_\pv V(\Cl
  Z)$, there is $w\in\Om AV$ such that $uwv\in M_\pv V(\Cl Z)$.
  Since $uwv$ is a factor of $v$ and $v$ is a \Cl J-minimum element of
  $M_\pv V(\Cl Z)$, we have $uwv\in J_\pv V(\Cl Z)$. As $f\in R_{\mathscr J}$,
  every two elements in the image of $f$ are~$\Cl R$-comparable,
  and so the elements in the image of $f$ have all the same set $P$ of
  finite prefixes. By
  Lemma~\ref{l:bridge-between-j-minimum-elements-LSl-case}, for each
  $\Cl X\in\Dom f$ such that $\Cl Z\subseteq \Cl X$, there is a
  pseudoword $w'$, depending only on $uwv$ and $P$, such that
  $uwvw'f(\Cl X)\in M_\pv V(\Cl X)$. More precisely, we have
  $uwvw'f(\Cl X)\in J_\pv V(\Cl X)$, as $f(\Cl X)\in J_\pv V(\Cl X)$.
  The partial function
  \begin{equation*}
    f'\colon \Cl X\in\Dom f
    \cup\{\Cl Z\}\mapsto
    \begin{cases}
      f(\Cl X)&\text{ if }\Cl X\subsetneq \Cl Z,\\
      uwv&\text{ if }\Cl X=\Cl Z,\\
      uwvw'f(\Cl X)&\text{ if }\Cl Z\subsetneq \Cl X,
    \end{cases}
  \end{equation*}
  belongs to $R_{\mathscr J}$ (cf.
  implication~\eqref{eq:upper-chain-of-r-classes}) and $\Dom
  f\subsetneq \Dom f'$. This contradicts the fact that $f$ is a
  maximal element of $R_{\mathscr J}$. The absurdity comes from the hypothesis
  $\Cl C\setminus \Dom f\neq\emptyset$.
\end{proof}
  
We recall the concept of entropy of a pseudoword, first introduced
in~\cite{Almeida&Volkov:2006}, and applied there in the study of
relatively free profinite semigroups. Some further applications were
given in~\cite{ACosta&Steinberg:2011}. Let \pv V be a pseudovariety
containing $\pv {LSl}$ and $A$ an alphabet with at least two letters.
For $w\in \Om AV$, let $q_w(n)$ denote the number of factors of length
$n$ of $w$. If $w$ is an infinite pseudoword then the sequence
$\frac{1}{n}\log_2 q_w(n)$ converges to its infimum, which is denoted
by~$h(w)$ and called the \emph{entropy of~$w$}.\footnote{This is the
  definition used in~\cite{ACosta&Steinberg:2011}.
  In~\cite{Almeida&Volkov:2006} the entropy of $w$ is defined as
  $\frac{1}{n}\log_{|A|} q_w(n)$, which equals $h(w)\log_{|A|}2$ for
  $h(w)$ as defined here.} This definition extends to finite words,
by letting $h(w)=0$ when $w$ is finite. If $\Cl X$ is a subshift, then
$h(\cal X)=h(w)$ for every pseudoword $w$ whose set of finite factors
is equal to $L(\Cl X)$. For instance, if $\Cl X$ is irreducible then
$h(\cal X)=h(w)$ when $w\in J_\pv V(\Cl X)$.

Note that $h(w)\in[0,\log_2|A|]$, for all $w\in\Om AV$. Moreover, we
have the following fact from~\cite{Almeida&Volkov:2006}.
    
\begin{Prop}\label{p:maximal-entropy-is-minimum-ideal}
  Let $w\in\Om AV$. Then $h(w)=\log_2|A|$ if and only if $w$ belongs
  to the minimum ideal of $\Om AV$.
\end{Prop}
    
In particular, the entropy of pseudowords of $\Om AV$ is not
continuous, since every finite word has entropy zero and the set of
finite words is dense. However, it is upper semi-continuous, as proved
next.

\begin{Lemma}\label{l:entropy-is-upper-semi-continuous}
  Let $\pv V$ be a pseudovariety containing $\pv {LSl}$. If $(w_n)_n$
  is a sequence of elements of $\Om AV$ converging to $w$ then
  $\limsup h(w_n)\leq h(w)$.
\end{Lemma}

\begin{proof}
  Since $\limsup h(w_n)$ is the greatest accumulation point of the
  sequence $(h(w_n))_n$, the proof is reduced to the case where
  $(h(w_n))_n$ converges.
    
  Since $\lim w_n=w$ and $\pv V$ contains $\pv{LSl}$, for each $k$
  there is $p_k$ such that for all $n\geq p_k$ the pseudowords $w_n$
  and $w$ have the same factors of length~$k$. Let $(n_k)_k$ be the
  sequence recursively defined by $n_1=p_1$ and
  $n_{k+1}=\max\{n_k,p_{k+1}\}$. Given $\varepsilon>0$, consider the
  set $K=\{k:h(w_{n_k})\geq h(w)+\varepsilon\}$. For every $k\in K$,
  one has
  \begin{equation}\label{eq:leads-to-contradiction-entropy}
    \frac{1}{k}\log_2 q_w(k)=\frac{1}{k}\log_2 q_{w_{n_k}}(k)
    \geq h(w_{n_k})\geq h(w)+\varepsilon.
  \end{equation}
  As $\lim\frac{1}{k}\log_2 q_w(k)=h(w)$, if $K$ is infinite
  then~\eqref{eq:leads-to-contradiction-entropy} leads to the
  contradiction $h(w)\geq h(w)+\varepsilon$. Hence $K$ is finite, and
  so $\lim h(w_{n_k})\leq h(w)+\varepsilon$. Since $\varepsilon$ is
  arbitrary and $(h(w_n))_n$ converges, we conclude that $\lim
  h(w_n)\leq h(w)$.
\end{proof}

We now have all the tools to achieve the proof of
Theorem~\ref{t:pseudword-minimum-ideal-r-chain-limit}.

\begin{proof}[Proof of
  Theorem~\ref{t:pseudword-minimum-ideal-r-chain-limit}]
  Let $\Cl C$ be the chain $(\mathcal S_\beta)_{{]}1,n{[}}$, ordered
  by $\supseteq$. Applying
  Proposition~\ref{p:continuum-chain-of-R-classes} to the family of
  $\Cl J$-classes $(J_{\pv V}(\Cl X))_{\Cl X\in\Cl C}$, we conclude
  that there is a function $f\colon \Cl C\to \Om AV$ such that $f(\Cl
  S_\beta)\in J_{\pv V}(\Cl S_\beta)$ and
  \begin{equation*}
    \Cl S_\beta\supseteq \Cl S_\alpha\implies
    f(\Cl S_\beta) \leq_{\Cl R}f(\Cl S_\alpha),
    \qquad\forall \alpha,\beta\in {]}1,n{[}.
  \end{equation*}
  On the other hand, if $f(\Cl S_\beta) \leq_{\Cl R}f(\Cl S_\alpha)$,
  then $\Cl S_\beta\supseteq \Cl S_\alpha$ by
  equivalence~\eqref{eq:subshifts-j-classes} in
  Remark~\ref{r:j-order-shifts}. Therefore, we actually have
  \begin{equation}
    \label{eq:more-precise}
    \Cl S_\beta\supsetneq \Cl S_\alpha\iff 
    f(\Cl S_\beta) <_{\Cl R}f(\Cl S_\alpha),
    \qquad\forall \alpha,\beta\in {]}1,n{[}.
  \end{equation}
  For each $\beta\in {]}1,n{[}$, let $w^{(\beta)}=f(\Cl S_\beta)$. By
  the given characterization of $(\Cl S_\beta)_{\beta\in
    {]}1,+\infty{[}}$
  (cf.~Theorem~\ref{t:theorem-beta-shifts-that-provides-examples}), we
  know that $\Cl S_\beta\supsetneq \Cl S_\alpha$ if and only if
  $\alpha<\beta$, whence~\eqref{eq:more-precise} translates to
  \begin{equation}\label{eq:translation-for-w}
    \alpha<\beta \iff w^{(\beta)}<_{\Cl R}w^{(\alpha)},
    \qquad\forall \alpha,\beta\in {]}1,n{[}.
  \end{equation}
  Let $(\alpha_k)_k$ be an increasing sequence of elements of the open
  interval ${]}1,n{[}$ such that $\lim\alpha_k=n$. Thus, we have $\lim
  h(w^{(\alpha_k)})=\lim \log_2 \alpha_k=\log_2 n$. Hence, if
  $w^{(n)}$ is an accumulation point of the sequence
  $(w^{(\alpha_k)})_{k}$, then $h(w^{(n)})=\log_2n$ by
  Lemma~\ref{l:entropy-is-upper-semi-continuous} and
  Remark~\ref{r:entropy-full-shift}. By
  Proposition~\ref{p:maximal-entropy-is-minimum-ideal}, the pseudoword
  $w^{(n)}$ then belongs to the minimum ideal of $\Om AV$. Since
  $\leq_\Cl R$ is a closed relation, we have $w^{(n)}\leq_\Cl R
  w^{(\beta)}$ for all $\beta\in {]}1,n{[}$. Then, taking into
  account~\eqref{eq:translation-for-w}, we conclude that the net
  $(w^{(\beta)})_{\beta\in {]}1,n{]}}$ satisfies
  conditions~\ref{item:pseudword-minimum-ideal-r-chain-limit-1}-\ref{item:pseudword-minimum-ideal-r-chain-limit-3}
  in Theorem~\ref{t:pseudword-minimum-ideal-r-chain-limit}.

  As $w^{(n)}\leq_{\Cl R}w^{(\beta)}$, there is $u^{(\beta)}$ with
  $w^{(n)}=w^{(\beta)}u^{(\beta)}$. Since $w^{(\beta)}$ is regular,
  there is an idempotent $f^{(\beta)}$ in the $\Cl L$-class of
  $w^{(\beta)}$. Take $v^{(\beta)}=f^{(\beta)}u^{(\beta)}$. Then
  $(w^{(\beta)},v^{(\beta)})$ is an element of $\mathfrak
  {F}(w^{(n)})$ stabilized by $f^{(\beta)}$. By
  Proposition~\ref{p:shifting-allowed-implies-type-2}, this implies
  that $q_\beta=(w^{(\beta)},v^{(\beta)})/{\sim}$ is a stationary
  point.
  The elements of $J_{q_\beta}$ are factors of $w^{(\beta)}$, and so
  by the minimality of $J_{q_\beta}$ we have $f^{(\beta)}\in
  J_{q_\beta}$. Hence,
  condition~\ref{item:pseudword-minimum-ideal-r-chain-limit-4} in
  Theorem~\ref{t:pseudword-minimum-ideal-r-chain-limit} holds.

  Let $\alpha,\beta\in {]}1,n{]}$. Since $q_\alpha\leq
  q_\beta\Rightarrow w^{(\beta)}\leq_{\Cl R} w^{(\alpha)}$, we deduce
  from~\eqref{eq:translation-for-w} that $q_\alpha\leq
  q_\beta\Rightarrow\alpha\leq \beta$. Thus, if $\alpha<\beta$
  then we cannot have $q_\beta\leq q_\alpha$, and so assuming $\Om AV$
  is equidivisible, we get $q_\alpha< q_\beta$, thereby establishing
  condition~\ref{item:pseudword-minimum-ideal-r-chain-limit-5} in
  Theorem~\ref{t:pseudword-minimum-ideal-r-chain-limit}.
\end{proof}

\section*{Acknowledgments}

The work of the first, third, and fourth authors was partly supported
by the \textsc{Pessoa} French-Portuguese project ``Separation in
automata theory: algebraic, logical, and combinatorial aspects''.
The work of the first three authors was also partially supported
respectively by CMUP (UID/MAT/ 00144/2013), CMUC (UID/MAT/00324/2013),
and CMAT (UID/MAT/ 00013/2013), which are funded by FCT (Portugal) with
national (MCTES) and European structural funds (FEDER), under the
partnership agreement PT2020.
The work of the fourth author was partly supported by ANR 2010
BLAN 0202 01 FREC and by the DeLTA project ANR-16-CE40-0007.

\bibliographystyle{amsplain}

\providecommand{\bysame}{\leavevmode\hbox to3em{\hrulefill}\thinspace}
\providecommand{\MR}{\relax\ifhmode\unskip\space\fi MR }
\providecommand{\MRhref}[2]{%
  \href{http://www.ams.org/mathscinet-getitem?mr=#1}{#2}
}
\providecommand{\href}[2]{#2}

\end{document}